\documentclass[a4paper,UKenglish,cleveref, autoref, thm-restate,numberwithinsect]{lipics-v2021}

\hideLIPIcs  

\title{A Factorization Theorem for Forest Algebras} 

\author{Shaull Almagor}{Department of Computer Science, Technion, Israel \and \url{https://shaull.cswp.cs.technion.ac.il/}} {shaull@technion.ac.il}{https://orcid.org/0000-0001-9021-1175}{supported by the ISRAEL SCIENCE FOUNDATION (grant No. 989/22)}

\author{Michaël Cadilhac}{DePaul University, USA}{michael@cadilhac.name}{https://orcid.org/0000-0001-9828-9129}{}

\author{Asaf Shoham}{Department of Computer Science, Technion, Israel} {asaf.shoham@campus.technion.ac.il}{}{}

\authorrunning{S. Almagor, M. Cadilhac, A. Shoham}

\Copyright{Shaull Almagor, Michaël Cadilhac, and Aasf Shoham} 

\ccsdesc[100]{\textcolor{red}{Replace ccsdesc macro with valid one}} 

\keywords{Forest Algebra, Factorization Forest, Simon, Syntactic Monoid} 

\category{} 

\relatedversion{} 

\nolinenumbers 

\EventEditors{John Q. Open and Joan R. Access}
\EventNoEds{2}
\EventLongTitle{42nd Conference on Very Important Topics (CVIT 2016)}
\EventShortTitle{CVIT 2016}
\EventAcronym{CVIT}
\EventYear{2016}
\EventDate{December 24--27, 2016}
\EventLocation{Little Whinging, United Kingdom}
\EventLogo{}
\SeriesVolume{42}
\ArticleNo{23}

\usepackage{amssymb}
\usepackage{amsmath}
\usepackage[disable]{todonotes}
\usepackage{amsthm}
\usepackage{bbold} 
\usepackage{bbm}
\usepackage{mathtools}
\usepackage{subcaption}
\usepackage{soul}
\usepackage[linguistics]{forest}
\hypersetup{
    colorlinks=true, 
    allcolors=blue
}
\usepackage{csquotes}
\usepackage{xcolor}  

\usepackage{graphicx}
\usepackage{xspace}
\usepackage{scalerel}

\newcommand{\unitel}{\mathbb{1}}
\newcommand{\zeroel}{\mathbb{0}} 

\let\phi=\varphi

\newcommand{\tup}[1]{\langle #1 \rangle}
\renewcommand{\phi}{\varphi}

\newcommand{\sqr}{\mathtt{to\_square}}
\newcommand{\unr}{\mathtt{unravel}}

\newcommand{\id}{\mathtt{id}}

\newcommand{\depth}{\mathtt{depth}}

\newcommand{\rdef}{\mathtt{ref}}

 \newcommand{\ns}{\mathtt{NxtCtx}}
 \newcommand{\factor}{\mathtt{Factor}}
 \newcommand{\nf}{\mathtt{NxtFctrs}}
 \newcommand{\nts}{\mathtt{NxtTwoCtx}}
\newcommand{\clnk}{C_{link}}
\newcommand{\lr}{\mathtt{LR}}

\newcommand{\trec}{T_{2\to1}}

\newcommand{\AncEmb}{\mathtt{AncEmb}}

\newcommand{\nds}{\mathtt{Nodes}}
\newcommand{\frst}{\mathtt{forest}}
\newcommand{\ctx}{\mathtt{ctx}}
\newcommand{\type}{\mathtt{type}}

\newcommand{\val}{\mathtt{val}}

\newcommand{\lf}{\mathfrak{L}}
\newcommand{\B}{\mathfrak{B}}
\newcommand{\C}{\mathfrak{C}}
\newcommand{\I}{\mathfrak{I}}

\newcommand{\fd}{\mathtt{FullDec}}
\newcommand{\pd}{\mathtt{PartialDec}}

\newcommand{\dec}{\mathtt{Dec}}

\newcommand{\dom}{\mathtt{Domain}}

\newcommand{\cleq}{\preceq_{\textnormal{pre}}}
\newcommand{\lcrs}{\leq_{LCRS}}

\newcommand{\sing}{\Upsilon}

\newcommand{\hf}{{H_A}}
\newcommand{\vf}{V_A}
                
\newcommand{\hdef}{\tilde H_A}

\newcommand{\bbN}{\mathbb{N}}
\newcommand{\bns}{\bbN^*}
\newcommand{\bnp}{\bbN^+}

\newcommand{\mj}{\mathcal J}
\newcommand{\mr}{\mathcal R}
\newcommand{\mh}{\mathcal H}
\newcommand{\ml}{\mathcal L}

\newcommand{\zsimple}[0]{$\zeroel$-simple\xspace}

\newcommand{\frs}[1]{\mathtt{Forests}(#1)}
\newcommand{\cts}[1]{\mathtt{Contexts}(#1)}
\newcommand{\tds}{\mathtt{TDS}}
\newcommand{\ltds}{\mathtt{LTDS}}
\newcommand{\fds}{\mathtt{FDS}}
\newcommand{\lfds}{\mathtt{LFDS}}

\newcommand{\llfds}[1]{#1\text{-\texttt{LFDS}}}
\newcommand{\lltds}[1]{#1\text{-\texttt{LTDS}}}

\newcommand{\az}{0_\triangle}
\newcommand{\act}{\mathtt{act}}
\newcommand{\inl}{\mathtt{in_l}}
\newcommand{\inr}{\mathtt{in_r}}

\newcommand{\norm}[1]{\left\lVert#1\right\rVert}

\newcommand{\wHole}[1]{\widetilde{#1}}
\newcommand{\tilH}{\wHole{H_A}}
\newcommand{\tilV}{\wHole{V_A}}

\newcommand{\bs}{B_A}

\newcommand{\eqUnr}{\approx}

\newcommand{\tofd}{\mathsf{ForestDom}}
\newcommand{\lab}{\mathsf{lab}}

\newcommand{\equpto}[3]{#1 \mathrel{\equiv_{\!\setminus #2}} #3}

\newcommand{\spa}{\circledast}
\newcommand{\spos}{\mathtt{sq\_pos}}

\newcommand{\width}{\mathtt{Width}}

\newtheorem*{theorem*}{Theorem}

\newcommand{\mc}{\mathtt{V}}

\newcommand{\quo}{\texttt{quo}}

\newcommand{\resp}{resp.\xspace}

\newcommand{\ela}{\ell_a}
\newcommand{\elb}{\ell_b}

\newcommand{\semsymb}[1]{{\boldsymbol{\color{red!75!black} #1}}}
\newcommand{\sela}{\semsymb{\ela}}
\newcommand{\selb}{\semsymb{\elb}}
\newcommand{\sell}{\semsymb{\ell}}

\newcommand{\alaT}{\semsymb{a\ela T}}
\newcommand{\alaF}{\semsymb{a\ela F}}
\newcommand{\albT}{\semsymb{a\elb T}}
\newcommand{\albF}{\semsymb{a\elb F}}
\newcommand{\alT}{\semsymb{a\ell T}}
\newcommand{\blaT}{\semsymb{b\ela T}}
\newcommand{\blaF}{\semsymb{b\ela F}}
\newcommand{\blbT}{\semsymb{b\elb T}}
\newcommand{\blbF}{\semsymb{b\elb F}}
\newcommand{\blT}{\semsymb{b\ell T}}

\newcommand{\latwoT}{\semsymb{\ela 2T}}
\newcommand{\latwoF}{\semsymb{\ela 2F}}
\newcommand{\lbtwoT}{\semsymb{\elb 2T}}
\newcommand{\lbtwoF}{\semsymb{\elb 2F}}

\newcommand{\TwRoot}{\semsymb{T1}}
\newcommand{\TwoRoot}{\semsymb{T2}}
\newcommand{\allell}{\semsymb{\ell 2}}
\newcommand{\sbot}{\semsymb{\bot}}
\newcommand{\sunit}{\semsymb{\epsilon}}

\newcommand{\conta}{\semsymb{f_{a}}}
\newcommand{\contb}{\semsymb{f_{b}}}
\newcommand{\contela}{\semsymb{f_{\ela}}}
\newcommand{\contelb}{\semsymb{f_{\elb}}}
\newcommand{\contel}{\semsymb{f_{\ell}}}

\newcommand{\contbot}{\semsymb{f_{\bot}}}

\newcommand{\plusa}{\semsymb{+a}}
\newcommand{\plusb}{\semsymb{+b}}
\newcommand{\aplus}{\semsymb{a+}}
\newcommand{\bplus}{\semsymb{b+}}
\newcommand{\allellplus}{\allell\semsymb{+}}

\newcommand{\twotrianglesraw}{%
  \tikz[baseline=-0.1ex, x=0.35ex, y=0.37ex]{
    \fill[black]
      (0,0) -- (1,0) -- (0.5,1) -- cycle;
    \fill[gray]
      (0.5,0) -- (1.5,0) -- (1,1) -- cycle;
  }%
}

\newcommand{\twowhitetrianglesraw}{%
  \tikz[baseline=-0.1ex, x=0.3ex, y=0.3ex]{
    \draw[black, line width=0.03ex, fill=white]
      (0,0) -- (1,0) -- (0.5,1) -- cycle;
    \draw[black, line width=0.03ex, fill=white]
      (0.5,0) -- (1.5,0) -- (1,1) -- cycle;
  }%
}

\newcommand{\whitetriangleraw}{%
  \tikz[baseline=-0.1ex, x=0.3ex, y=0.3ex]{
    \draw[black, line width=0.03ex, fill=white]
      (0,0) -- (1,0) -- (0.5,1) -- cycle;
  }%
}

\newcommand{\blacktriangleraw}{%
\tikz[baseline=-0.1ex, x=0.35ex, y=0.37ex]{
    \fill[black]
      (0,0) -- (1,0) -- (0.5,1) -- cycle;
  }%
}

\newcommand{\twotriangles}{%
  \mathord{\ThisStyle{\scalerel*{\twotrianglesraw}{\ensuremath{\SavedStyle\triangle}}}}%
}

\newcommand{\twowhitetriangles}{%
  \mathord{\ThisStyle{\scalerel*{\twowhitetrianglesraw}{\ensuremath{\SavedStyle\triangle}}}}%
}

\newcommand{\mywhitetriangle}{%
  \mathord{\ThisStyle{\scalerel*{\whitetriangleraw}{\ensuremath{\SavedStyle\triangle}}}}%
}

\newcommand{\myblacktriangle}{%
  \mathord{\ThisStyle{\scalerel*{\blacktriangleraw}{\ensuremath{\SavedStyle\triangle}}}}%
}

\newcommand{\TBFt}{{\mywhitetriangle}}
\newcommand{\TBFf}{{\myblacktriangle}}
\newcommand{\TBFtf}{{\twotriangles}}
\newcommand{\TBFtt}{{\twowhitetriangles}}

\newcommand{\fpt}{f_{\TBFt}}
\newcommand{\fpf}{f_{\TBFf}}

\newcommand{\fnot}{f_{\lnot}}

\newcommand{\fand}{f_{\wedge}}

\begin{document}

\maketitle

\begin{abstract}
Simon’s factorization theorem is a celebrated tool in algebraic automata theory, providing bounded-depth decompositions of words with respect to morphisms into finite semigroups.

We develop an analogue of Simon's theorem for \emph{forests} in the setting of forest algebras. In contrast with words, this presents a basic difficulty: recursively factoring a forest requires keeping track of where each subforest ``fits''. This difficulty ripples throughout the proof, and we overcome it by augmenting the free forest algebra and by developing a framework that supports recursive factorization of forests, along with its semantic implications.

Our main result identifies a new semantic restriction on morphisms (called $\mr$-alignment) which intuitively ensures that different ways of cutting a forest remain compatible (in a certain sense) at the semigroup level. 
Under this condition, we prove that every morphism admits decompositions of bounded depth. 
We also prove that without this restriction, there are morphisms for which no bounded-depth decomposition exists (under our notion of decomposition). 
\end{abstract}

\section{Introduction}
\label{sec:abs:intro}

Regular languages can be studied through automata, logic, or algebra, and each
of these viewpoints comes with its own structural tools. In the algebraic study
of word languages, one of the most influential such tools is Simon's
factorization forest theorem~\cite{Simon1990,Colcombet2021}. If one fixes a
morphism $\eta \colon A^+ \to S$ into a finite semigroup, then $\eta$ does not merely assign
to each word an element of $S$; it also constrains how words can be recursively
factored. Informally, Simon's theorem says that every word admits a rooted
decomposition tree of depth bounded solely in terms of the target semigroup. The
leaves are letters, the internal nodes correspond to concatenations of
consecutive factors, and the high-arity nodes arise when all children have the
same idempotent image under $\eta$. In this way, algebraic recognition yields not
only a classification of words by their images, but also a bounded-depth
structural decomposition of every input.

The purpose of this paper is to obtain an analogous bounded-depth decomposition
theorem for \emph{forests}. At a broad level, the intended parallel is clear: one would
like a morphism recognizing a forest language to impose on every forest a
decomposition whose depth is controlled only by the finite target algebra. What
makes this difficult is that the very notion of factorization changes in the
branching setting. For words, one splits into consecutive factors, and each
factor is again a word, so the recursive process remains in the same class of
objects. For forests, by contrast, removing a subtree leaves behind, in the
standard terminology, a \emph{context}: a forest with a distinguished hole
indicating where the missing subtree is to be reinserted. Once one continues
decomposing, several previously removed subforests may have to be remembered at
once. This bookkeeping problem has no counterpart in the word case, and it is
the first obstruction to a direct forest version of Simon's theorem.

The natural algebraic framework for this problem is that of \emph{forest
algebras}, introduced by Bojańczyk and
Walukiewicz~\cite{BojanczykWalukiewicz2007}. A forest algebra is a pair $(H,V)$
of monoids, usually called the \emph{horizontal} and \emph{vertical} monoids. The horizontal
monoid $H$ captures forest types, the vertical monoid $V$ captures context
types, and $V$ acts on $H$ by substitution (see \cref{sec:abs:forest algebra}).
In the \emph{free} forest algebra over an alphabet, this picture
becomes completely concrete: the elements of $H$ are forests, the elements of
$V$ are contexts, and the action is exactly the operation of plugging a forest
into the hole of a context. Forest algebras have proved to be a natural
algebraic formalism for regular languages of unranked forests and trees, and
they have already supported substantial interactions with logic and algorithms
on trees~\cite{BojanczykStraubingWalukiewicz2012,Bojanczyk2012}.

Our first contribution is a formalism that makes recursive decomposition
possible in this setting. We introduce an augmented free forest algebra in
which, besides the ordinary hole of a context, one may use \emph{default holes}
$\square_h$. Intuitively, $\square_h$ marks a specific subforest that has already been
extracted and is to be treated as fixed while the rest of the decomposition
continues. These markers let one continue decomposing the residue object
without losing track of how previously removed pieces fit back into the whole
forest, and they provide the bookkeeping device that is missing from
the naive forest analogue of factorization forests.

On top of this augmented algebra, we first define \emph{binary
decompositions}. Their role is not simply to encode binary splitting, but to
recover certain monotonicity properties that the classical proof for words gets ``for free''. 
Specifically, for a word $w$, if $P_w$ denotes the set of images under $\eta$ of the prefixes of $w$, then the implication
``$u \text{ prefix of } w \implies P_u \subseteq P_w$''
is immediate from properties of concatenation and morphisms.
In the forest setting, the relevant quantity is subtler, and only holds under certain conditions, which we therefore first ensure (see \cref{sec:abs:next contexts}). 
%

Binary decompositions are therefore fundamental in order to recover a notion of
decomposition that, in the setting of words, comes for free by the definition of
concatenation.  Equipped with these, we then introduce \emph{general
  decompositions}, the forest analog of Simon's decomposition trees. The nodes in a  general decomposition may be leaves, binary nodes, or idempotent
nodes. An idempotent ($\I$) node compresses an entire region that already admits a
binary decomposition all of whose relevant contexts map to a single idempotent
of the target semigroup; this is the forest analogue of the high-arity
idempotent node in Simon's theorem.  For technical convenience, and to obtain
better bounds, we also introduce a third type of inner nodes: centipede $(\C)$
nodes.

\begin{remark}[Visual representation of forests v.s. decomposition trees]
\label{rmk:abs:depiction forests vs trees}
In order to visually separate forests (elements in a forest algebra, including trees) from decomposition trees (the tree-structures we define), we depict forests with diagonal downward edges, and trees with right and down edges (as depicted in e.g., \cref{fig:abs:TBF_combined}).
\end{remark}

\begin{example}[TBF]
    \label{xmp:abs:TBF1}
    We demonstrate the setting with a concrete example. We refer back to this example as the concepts are clarified throughout the paper.
    Consider the language of True Boolean Formulas (TBF), which consists of binary trees labeled by $A=\{T,\neg,\wedge\}$ that represent a formula evaluating to true. 

    This language can be formulated using forest algebras. For example, $\neg  (\wedge (T + \neg(T))$  is a forest (in fact, a tree) in the free forest algebra $\tup{H_A,V_A}$ over $A$. It represents a formula that evaluates to true.
    We can write this tree as a composition of the \emph{context} $\neg (\wedge (\square))$ with the forest $T+\neg(T)$ (note that $+$ is the horizontal operator, and represents putting two forests ``side by side'').
    
    For a more elaborate forest, see e.g., the node labeled $\epsilon$ in \cref{fig:abs:TBF_I_decomposition}. There we also illustrate a binary decomposition of the forest: in each node we ``pluck'' a sub-forest to the right child, and put the ``residue'' in the bottom child. The contexts used in the decomposition are illustrated in the yellow ovals. In order to keep track of plucked sub-forests, we use default holes. For example, node $00$ has default holes corresponding to both node $01$ ($\square_f$) and to node $1$ ($\square_g$). 
    In \cref{fig:abs:TBF_full} we illustrate a general decomposition, starting with an $\I$-node at the root. We revisit this in due course.

    
    \begin{figure}[ht]
    \centering
    \begin{subfigure}{0.25\linewidth}
        \centering
        \includegraphics[width=\linewidth]{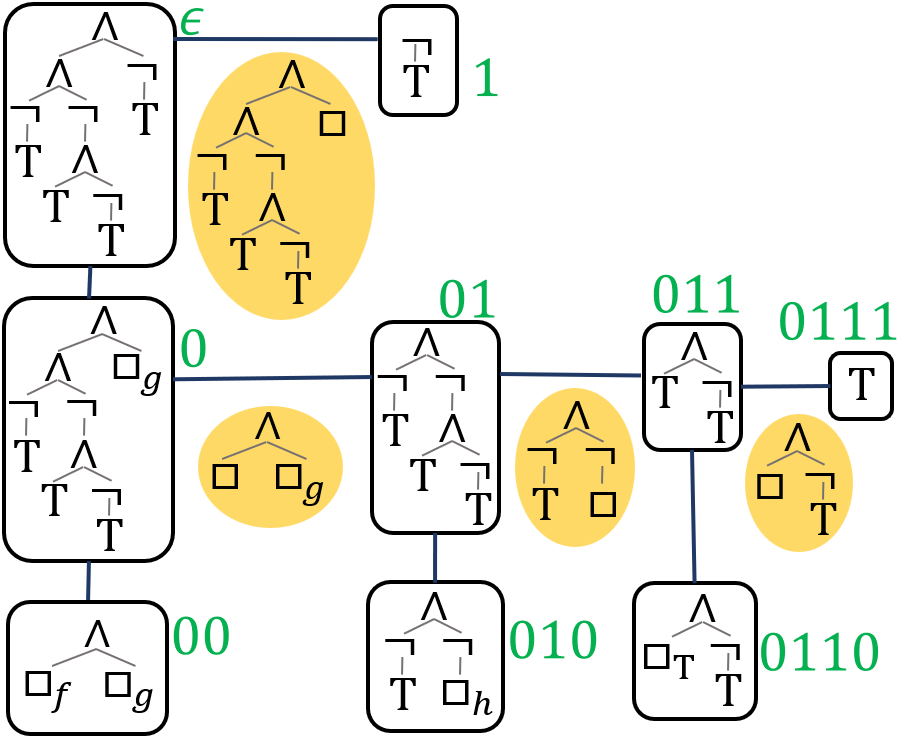}
        \caption{A partial binary decomposition. For TBF, this is a decomposition corresponding to the $\I$-node in \cref{fig:abs:TBF_full}.}
        \label{fig:abs:TBF_I_decomposition}
    \end{subfigure}
    \hfill
    \begin{subfigure}{0.7\linewidth}
        \centering
        \includegraphics[width=\linewidth]{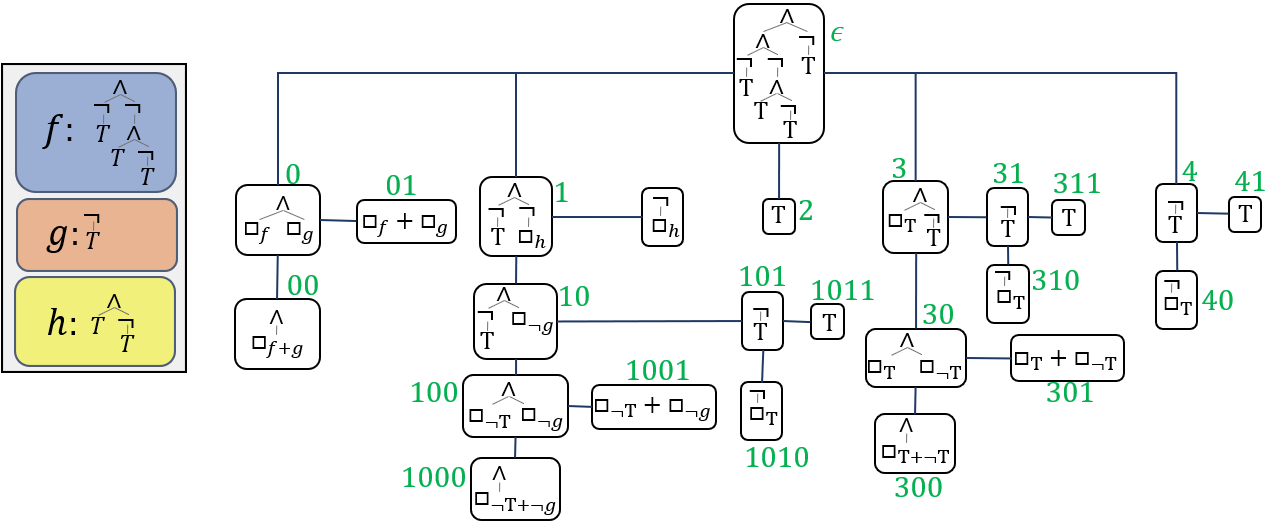}
        \caption{A general decomposition for TBF. The node at $\epsilon$ is an $\I$-node.}
        \label{fig:abs:TBF_full}
    \end{subfigure}
    \caption{Decompositions for TBF. $f,g,h$ are given by the legend.}
    \label{fig:abs:TBF_combined}
\end{figure}
\end{example}

With this machinery in place, we identify a semantic condition on morphisms,
called \emph{$\mr$-alignment}. Its definition appears in \cref{sec:abs:mr alignment}, but its role can be described informally as follows. When the same forest can be obtained by cutting in two different ways, one obtains two residue contexts. If their images lie in the same $\mj$-class, $\mr$-alignment requires these images to remain compatible in a stronger one-sided sense. In the word setting, the corresponding coherence is largely built into the linear order of factors; in the forest setting, it has to be identified explicitly as a property of the recognizing morphism.

Our main theorem shows that this condition is sufficient for bounded-depth
decomposition.

\begin{theorem*}[informal] Let $\varphi$ be a forest morphism into a finite forest
  algebra $(H,V)$, and assume that the $V$-component of $\varphi$ is $\mr$-aligned. 
  Then every forest admits a general decomposition of depth at most $4|V|-3$.
\end{theorem*}

We also show that, without $\mr$-alignment, the notion of decomposition
developed here does not admit any uniform depth bound in general. Thus the
restriction is not merely a byproduct of the proof technique.
The proof follows the same broad lines as semigroup-theoretic proofs of Simon's
theorem, in that it is organized around the $\mj$-order on the target (vertical)
semigroup~\cite{Simon1990,Colcombet2021}. More precisely, we establish bounded
decomposition theorems for three building-block cases---groups, simple or
\zsimple semigroups, and null semigroups---and then combine them by an induction
that climbs up the hierarchy of $\mj$-classes via a $\mj$-minimal nonzero
element, the ideal it generates, and the corresponding Rees quotient. 
Among these ingredients, the simple or \zsimple semigroup case is technically the most involved, and it is also the only point at which $\mr$-alignment is used in an essential way; the group and null cases are treated separately and feed into the later induction.


We view this theorem as a significant foundational step 
toward a decomposition theory for forests,
parallel to the role that factorization forests have long played for words. Such
a theory is interesting already as a structural tool for inductive arguments
over forest morphisms, but it also suggests algorithmic directions, since
bounded-depth decompositions are natural candidates for hierarchical summaries
of large tree-shaped inputs~\cite{Bojanczyk2012}. We do not pursue those
applications here in full, but the framework developed in this paper is designed
with that broader perspective in mind.

The extended abstract provides an intuitive overview of the main ideas, and 
is organized as follows. We begin by recalling forests,
contexts, and forest algebras, and by introducing the augmented free forest
algebra with default holes. We then define binary decompositions and show, using
references and rotations, how to recover the monotonicity properties they
need. After that we introduce general decompositions, which are the
bounded-depth objects used in the theorem. In the second half of the paper, we
define $\mr$-alignment, establish the group and simple semigroup decomposition
theorems, and combine them into the main result by induction on the target
semigroup. We conclude with the reduction lemma and with a lower bound showing
that without $\mr$-alignment no bounded-depth theorem can hold in general.
Following the extended abstract, the full technical details start in
\Cref{sec: appendix organization}.

\section{Preliminaries}
\label{sec:abs:prelim}
In this extended abstract we omit most technical details. For the technical preliminaries see \cref{sec:prelim,sec:green prelim}. 

Let $\bbN$ be the non-negative integers. For a set $A$ let $A^*$ and $A^{+}$ be the sets of finite words and finite non-empty words over $A$. Denote by $\epsilon$ the empty word. For $x,y\in A^*$ we write $x\le y$ if $x$ is a prefix of $y$.
A \emph{tree} is a prefix-closed subset of $\bbN^*$. A \emph{forest} is a prefix-closed subset of $\bbN^+$. Note that trees are rooted at $\epsilon$, whereas forests have roots in $\bbN$. 
For nodes $x\le y$ in a forest $f$, we say that $y$ is a \emph{descendant} of $x$. If $y=x\cdot i$ for some $i\in \bbN$, then $y$ is a \emph{child} of $x$.
The \emph{depth} of a forest $f$ is $\max\{d \geq 0 : \bbN^d\cap f\neq \emptyset\}$.
A \emph{labeled forest over alphabet $A$} is a forest $f$ equipped with a labeling function $\lab\colon f\to A$. We overload notation and write $f[x]$ as $\lab(x)$ for $x\in f$. We say that a labeled forest $g$ \emph{extends} forest $f$ if $f\subseteq g$ as forests, and $f[x]=g[x]$ for all $x\in f\cap g$.

\subparagraph*{Semigroups and Green's Relations}
A \emph{semigroup} is a structure $\tup{S,\cdot}$ such that $S$ is a set and $\cdot\colon S\times S\to S$ is an associative operation. We often write $uv$ instead of $u\cdot v$ for $u,v\in S$. An element $e\in S$ is \emph{idempotent} if $e^2=e$. An element $z$ is a \emph{zero}, denoted $\zeroel$ if $zx=xz=z$ for all $x\in S$. An element $u$ is a \emph{unit}, denoted $\unitel$, if $ux=xu=x$ for all $x\in S$. We denote by $S^1$ the \emph{monoid} obtained by adding a unit to $S$, if it does not have one.

Consider two elements $u,v\in S$. Green's relations are defined as follows.\\
\begin{tabular}{p{0.45\textwidth} p{0.45\textwidth}}
\labelitemi\ $u\leq_\ml v$ if $u\in S^1v$.
&
\labelitemi\ $u\leq_\mj v$ if $u\in S^1vS^1$. \\
\labelitemi\ $u\leq_\mr v$ if $u\in vS^1$. 
&
\labelitemi\ $u\leq_\mh v$ if $u\leq_\ml v$ and $u\leq_\mr v$.
\end{tabular}\\
Every Green relation $G\in \{\mj,\ml,\mr,\mh\}$ induces an equivalence relation by $u G v\iff u\leq_Gv\wedge v\leq_G u$, and we denote by $[u]_G$ the equivalence class of $u$, dubbed the \emph{$G$-class of $u$}. 

The semigroup $S$ is a \emph{simple} (resp. \emph{\zsimple}) if it has only one $\mj$-class (resp., only $\{\zeroel\}$ and $S\setminus \{\zeroel\}$ are the $\mj$-classes).
In \cref{sec:green prelim} we present several results about Green's relations and simple semigroups, used throughout the proof. These are omitted here as we do not reach this level of detail.

\subparagraph*{Forest Algebras (\cref{sec:forest algebra})} Glossing over some technical details, a forest algebra is a tuple $\tup{H,V,+,\cdot}$ where $\tup{H,+}$ and $\tup{V,\cdot}$ are monoids, along with an action of $V$ on $H$ (which we denote as $\cdot\colon V\times H\to H$ or omit) such that for every $u,v\in V$ and $h\in H$ we have $(u\cdot v)\cdot h=u\cdot (v\cdot h)$.
A \emph{forest morphism} from $\tup{H_1,V_1}$ to $\tup{H_2,V_2}$ is a pair
$\tup{\eta,\varphi}$ where $\eta\colon H_1\to H_2$ and $\varphi\colon V_1\to V_2$, with $\eta(v \cdot h) =
\varphi(v) \cdot \eta(h)$ for every $v \in V, h \in H.$

The \emph{Free Forest Algebra} over alphabet $A$ is a concrete forest algebra $\tup{H_A,V_A}$ where $H_A$ is the set of $A$-labeled forests, and $V_A$ is the set of $A\cup \{\square\}$-labeled forests, where $\square$ appears exactly in one leaf. The action of $V_A$ on $H_A$ is the natural composition where we replace $\square$ with the given subforest. We remark that \cref{sec:forest contexts and algebras,sec:toolbox for free algebras} are devoted to establishing a convenient formalism to capture and reason about free forest algebras. These formalisms enable our precise proofs, but encumber notations, hence our informality here.

\begin{example}[TBF as a Forest Algebra]
\label{xmp:tbf prelim forest algebra}
We formulate TBF from \cref{xmp:abs:TBF1} as a forest morphism $\tup{H_A,V_A}\to \tup{H,V}$ for $A=\{\wedge,\neg,T\}$. We start with describing $\tup{H,V}$.
The forests are $H=\{\epsilon,\TBFt,\TBFf,\TBFtf,\TBFtt,\bot\}$ with the following intuition: $\epsilon$ captures the empty forest, and $\bot$ captures all the ``hopelessly invalid'' forests (e.g., forests that are malformed, ternary nodes, etc.).
$\TBFt$ and $\TBFf$ capture, respectively, trees whose root evaluates to true and false, respectively. $\TBFtt$ and $\TBFtf$ capture forests of the form $t_1+t_2$ where, respectively, both $t_1,t_2$ evaluate to true and where at least one evaluates to false. Note that such forests are then plugged into a $\wedge$ node, and we therefore need a way of tracking them, hence the latter two values.

To describe $V$, we identify it with a subset of $H^H$, so that each element describes its action on $H$. Since this set is elaborate, we only demonstrate a few important elements in $V$. Note that each context ``expects'' a certain structure of its input element, specifically whether it should have one or two roots. For example, the elements $\fpt,\fpf,\fand,\fnot$ intuitively capture contexts such as $\square+T$, $\square+\neg(T)$, $\wedge(\square)$, and $\neg(\square)$ respectively. Thus, they are given by the following functions (undefined inputs map to $\bot$): 
    \[\fpt(h)=\begin{cases}
        \TBFt & \epsilon\\
        \TBFtf & \TBFf\\
        \TBFtt & \TBFt
    \end{cases}\quad 
    \fpf(h)=\begin{cases}
        \TBFf  & \epsilon\\
        \TBFtf & \TBFt \text{ or }\TBFf
    \end{cases}\quad 
    \fand(h)=\begin{cases}
            \TBFf & \TBFtf\\
            \TBFt & \TBFtt
    \end{cases}\quad
    \fnot(h)=\begin{cases}
        \TBFt & \TBFf\\
        \TBFf & \TBFt
    \end{cases}\]
Compositions then give rise to other elements, e.g., $\fand(\fpf)$ corresponds to e.g., the context $\wedge(\square+\neg(T))$, and is defined by $\fand(\fpf(h))=\TBFf$ for $h\in \{\TBFt,\TBFf\}$ and $\bot$ otherwise. Observe that this is an \emph{idempotent} element in $V$, i.e., $\fand(\fpf(\fand(\fpf)))=\fand(\fpf)$.  

The morphisms $\eta,\varphi$ are then naturally obtained by assigning each context/forest to the element it represents. 
Observe that all the contexts (in yellow) in \cref{fig:abs:TBF_I_decomposition} map to $\fand(\fpf)$.
\end{example}

\section{Augmented Forest Algebras (\cref{sec:augmented forest algebra})}
\label{sec:abs:forest algebra}
In the setting of words there is a natural notion of decomposition: a word $w$ can be split into $w=u\cdot v$, where both $u$ and $v$ are subwords (\cref{fig:abs:decomp-word}). Then, one can recursively continue to decompose. This decomposition supports the \emph{local} nature of decompositions -- in order to evaluate the morphism on an infix, one only needs to look at the corresponding nodes (indeed, this is the basis for the \emph{fast infix evaluation} algorithm~\cite{FactorForBoj}).

For forests, a fundamental problem arises: when decomposing a forest $h$, we want to  ``pluck'' a subtree rooted at a certain node $x$. This leaves us with a \emph{context}, rather than a forest, so that we know where to restore the plucked subtree.
However, if we now wish to continue decomposing, we may want to pluck another
subtree from the remaining context. This results in two distinct ``holes'' to
restore the subtrees (\cref{fig:abs:decomp-forest}), and this would render the
decomposition ill-defined, since it is not clear where to plug in each subforest
--- this is the main reason contexts have only one hole.
\begin{figure}[ht]
    \centering
    \begin{subfigure}{0.15\linewidth}
        \centering
        \includegraphics[width=\linewidth]{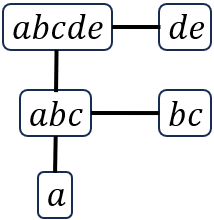}
        \caption{Words}
        \label{fig:abs:decomp-word}
    \end{subfigure}
    \hfill
    \begin{subfigure}{0.2\linewidth}
        \centering
        \includegraphics[width=\linewidth]{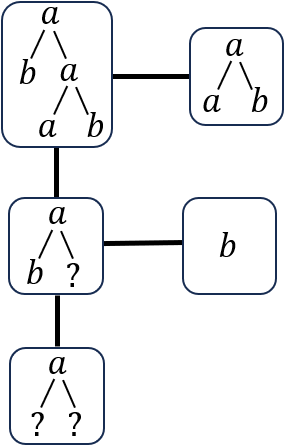}
        \caption{Forests without augmentation}
        \label{fig:abs:decomp-forest}
    \end{subfigure}
    \hfill
    \begin{subfigure}{0.2\linewidth}
        \centering
        \includegraphics[width=\linewidth]{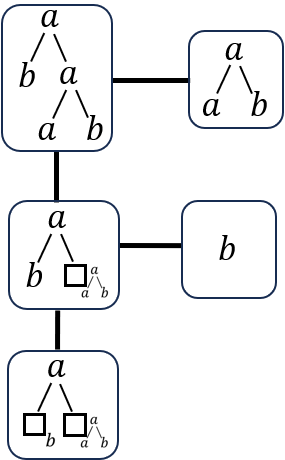}
        \caption{Forests with default holes}
        \label{fig:abs:decomp-augmented}
    \end{subfigure}
    \caption{For words (\cref{fig:abs:decomp-word}), each node can be evaluated separately according to the morphism. In particular, the value of a node is uniquely determined by the values of its children. \\
    For forests (\cref{fig:abs:decomp-forest}, it is not even clear how to represent the residue forest after ``plucking'' a subforest. For example, the bottom leaf cannot be evaluated, and looking at it locally does not provide enough information to know how the two right leaves plug into it.\\
    To overcome this, we use default holes (\cref{fig:abs:decomp-augmented}).}
    \label{fig:abs:problem-decomp}
\end{figure}

In order to overcome this problem, we introduce the \emph{augmented free forest algebra}. Intuitively, we add a new symbol $\square_h$ for every $h\in \hf$, called a ``default hole'': these are holes that can only be filled with the subforest $h$ (\cref{fig:abs:decomp-augmented}). 
A priori, it may seem as though this achieves nothing -- we just write $\square_h$ instead of $h$. Nonetheless, this is the basic building block for our notion of decomposition: the elements $\square_h$ are used as ``markers'' within a forest (or context) to denote a part which is \emph{elementary} and cannot be factored further.
We make this intuition concrete when defining decompositions in \cref{sec:abs:binary decomp}.

Technically, we augment the free forest algebra over alphabet $A$ with \emph{default holes}: let $\square_{H_A}\coloneqq \{\square_h:h\in H_A\}$, then we replace $A$ with the augmented infinite alphabet $\wHole{A}=A\cup \square_{H_A}$, and consider the free forest algebra over $\wHole{A}$, denoted $\tup{\tilH,\tilV}$, with one important restriction: the elements $\square_h$ may appear \emph{only in the leaves} of a forest or context (as intuitively they represent a suffix subforest, and do not have further children).
In \cref{xmp:abs:TBF1} above we demonstrate such default holes.

    
    

Default holes raise a problem: a morphism $\varphi\colon\tup{H_A,V_A}\to \tup{H,V}$ is not defined on default holes. However, a natural way to retain the semantics of a morphism is to \emph{unravel} a default hole back to its original forest. Thus, for $f\in \tilH\cup \tilV$ we define $\unr(f)\in H_A\cup V_A$ to be the corresponding element where every default hole $\square_h$ is replaced with $h$. 
Then, we define $\varphi(f)=\varphi(\unr(f))$.
This allows us to refer to default holes recursively: recall that we do not allow nested default holes by the definition of $H_A$. However, we can now safely define for $f\in \tilH$ that $\square_{f}=\square_{\unr(f)}$. 
Unraveling naturally induces a notion of equivalence on contexts and forests, by which $f$ and $g$ are \emph{unravel equivalent}, denoted $f\eqUnr g$, if $\unr(f)=\unr(g)$. 

In the following, we often consider specific subsets of $\tilV$ over which we
factor a forest. However, for these purposes we do not wish to distinguish
between equivalent contexts under $\unr$. Moreover, we want the set to be closed
under composition.  To this end, we call a set $\wHole{V}\subseteq \tilV$ \emph{stable}
if it is a subsemigroup of $\tilV$, and for every
$C_1\in \wHole{V}, C_2 \in \tilV$, if $C_1\eqUnr C_2$ then $C_2\in \wHole{V}$.

\section{Binary Decompositions (\cref{sec:bin decomp})}
\label{sec:abs:binary decomp}
Equipped with default holes, we can now define binary decompositions. For words, binary decompositions are trivial -- one merely splits a word recursively by concatenation (\cref{fig:abs:decomp-word}). 
For forests, we note that splitting a forest does not yield two forests, but rather a forest and a context. Then, we intuitively need to keep decomposing both parts. 
We illustrate how this is done before proceeding to the precise details.
\begin{example}
\label{xmp:abs:binary decomposition}
Consider the forest $h=b(a+a)+a(b+a(a+b))$. We wish to decompose it as \colorbox{blue!30}{$b(a+a)+a($}\colorbox{red!30}{$b+a(a+b)$}\colorbox{blue!30}{$)$}.
To this end, we write $h=C\cdot h'$ where $C=b(a+a)+a(\square)$ and $h'=b+a(a+b)$. The corresponding binary decomposition (with some extension) is depicted in \cref{fig:abs:bin decomp example with references}.
As per \cref{rmk:abs:depiction forests vs trees}, in the decomposition we refer to the  right-child of a node as the \emph{factor}, and to the down-child as the \emph{residue}. 
The decomposition tree is constructed as follows: we start with $h$ at the root of the decomposition, which is given address {\color{green!50!black} $\epsilon$}. 
We then ``pluck'' from $h$ the sub-forest factor $h'$ and place it as the right-child, with address {\color{green!50!black} $1$}. The node $1$ is also assigned a \emph{context}, which is the residue $C$ (so that $C\cdot h'=h$). Then, in the down-child (addressed {\color{green!50!black} $0$}) we place the tree that remains to be factored, which is $C$ only with $\square_{h'}$ instead of $\square$, to signify that this remaining hole is actually $h'$, but $h'$ itself is already being factored elsewhere (namely from node $1$).
Thus, the forest at node $0$ is $C\cdot \square_{h'}$.
\end{example}

\begin{figure}[ht]
    \centering
    \begin{minipage}{0.6\linewidth}
        \centering
        \includegraphics[width=0.5\linewidth]{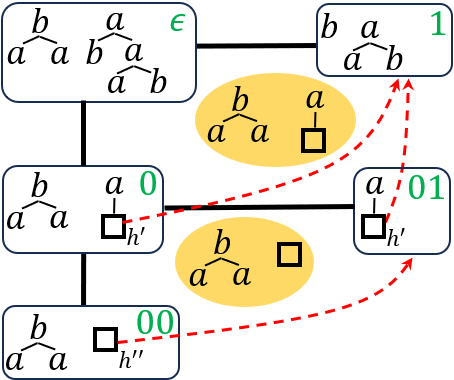}
        \caption{A (partial) binary decomposition with depicted references (dashed arrows). The first step corresponds to \cref{xmp:abs:binary decomposition}. The plucked factor $h'$ appears in Node $1$. The context $C$ appears highlighted below the $\epsilon\to 1$ edge.}
        \label{fig:abs:bin decomp example with references}
    \end{minipage}
    \hfill
    \begin{minipage}{0.3\linewidth}
        \centering
        \includegraphics[width=\linewidth]{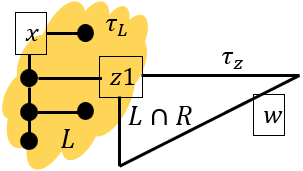}
        \caption{Decomposition structure of \cref{lem:abs:simple decomposition}.}
        \label{fig:abs:simple semigroup}
    \end{minipage}
\end{figure}

\begin{definition}[Binary Decomposition -- informal]
A \emph{binary decomposition} of a forest $h\in \tilH$ is a binary tree $\tau$ where each node $x$ is labeled with a forest $\tau.\frst(x)\in \tilH$ and has a type that is either a \emph{leaf} ($\tau.\type(x)=\lf$) or a \emph{binary node} $\tau.\type(x)=\B$. 
Binary nodes are associated with a context $\tau.\ctx(x)\in \tilV$, and the following conditions hold:
\begin{itemize}
    \item The root is labeled with $h$, i.e., $\tau.\frst(\epsilon)=h$.
    \item The forest of a binary node is decomposed correctly: $\tau.\frst(x)=\tau.\ctx(x1)\cdot \tau.\frst(x1)$.
    \item The residue matches the context:
    $\tau.\frst(x0)=\tau.\ctx(x1)\cdot \square_{\tau.\frst(x1)}$.
    \item Decompositions pluck out non-trivial forests: $\tau.\frst(x)\neq \tau.\frst(xi)$ for $i\in \{0,1\}$. 
\end{itemize}
\end{definition}
See \cref{fig:abs:bin decomp example with references} for an illustration, and \cref{def:bin decomp} for the formal details. A more elaborate example is in \cref{fig:abs:TBF_I_decomposition}.

An important observation to be made about binary decompositions is that for a binary node $x$, the context $\ctx(x1)$ actually determines both the factor and the residue.
Thus, most of our reasoning focuses on the contexts, rather than the forests.

\subsection{References and Stability (\cref{subsec:references and ancestor embedding})}
\label{sec:abs:references and stability}
The proof of Simon's factorization for words crucially relies on a certain monotonicity with respect to the image of the morphism. Specifically, the following property is used: for a word $w$ and a morphism $\varphi$, let $P_w=\{\varphi(u): w=u\cdot v\}$ be the set of images of the prefixes of $w$, then if $w=w_1w_2$, we trivially have that $P_{w_1}\subseteq P_{w}$. 
That is, the set of images of prefixes of $w_1$ is contained in that of $w$. This observation follows immediately by the definitions of a morphism and concatenation.
The dual observation for \(P_{w_2}\) is slightly more involved: we have that $\varphi(w_1)\cdot P_{w_2}\subseteq P_{w}$, which is again obvious.
Framed in the setting of binary decompositions, this relates the ``available ways'' of factoring a word $w$ into subwords, with the available ways of factoring each of its children.

Recovering this property under our binary decompositions is, as it turns out, nontrivial, and leads to new definitions that complicate the remainder of the proof.
At its core, the problem is that (unlike in the setting of words) each node in our binary decomposition may ``refer'' to other nodes by means of default holes. Moreover, this reference is meaningful as a possible decomposition, as we show in \cref{sec:abs:rotations,sec:abs:next contexts}.
To illustrate this, consider the decomposition in \cref{fig:abs:decomp-augmented}. Intuitively, we first pluck out $a(a+b)$ and then pluck out $b$ from the remaining forest, resulting in the two default holes at the bottom node. However, the default hole $\square_{a(a+b)}$ still ``references'' the node $a(a+b)$. 

Consider a binary decomposition $\tau$. To capture the notion of \emph{reference}, for every node $x$ and address $u$ in $\tau.\frst(x)$, we write $x@u\to y$ if $x[u]=\square_h$ such that $\tau.\frst(y)=h$ and $h$ was plucked to node $y$ in some parent of $x$, creating this specific default hole.  We write $x@*\to y$ when $u$ is immaterial.
In \cref{fig:abs:bin decomp example with references}, we depict references as dashed arrows. For example, $01@00\to 1$ means that node $01$, at address $00$ (i.e., the label $\square_{h'}$) references node $1$. 
We remark that capturing this definition formally involves a lot of accounting, see \cref{def:decomp references intuitive,def:decomp references}.
Of special importance are references from right-child nodes, called \emph{inherited references} (i.e, $x1@u\to y$). These are caused when a subforest is plucked, and later on a parent of this subforest is plucked (e.g., the reference $01@00\to 1$ mentioned above).  
Intuitively, these represent a ``bad'' way of factorization (in that it makes more ``sense'' to first pluck out the larger subforest). We make this precise in \cref{lem:abs:no inherited decomposition is stable}.

Next, for a reference $x@u\to y$ we define the \emph{linking context from $x$ to
  $y$} $\clnk(\tau,x,y)=\tau.\frst(x)[u\to \square]$, which formalizes the connection between
the nodes. Continuing the example of \cref{fig:abs:bin decomp example with
  references} we have $\clnk(\tau,01,1)=a(\square)$.
Intuitively, when considering a binary decomposition, the elements it represents are really the linking contexts, as these capture 
the ``infixes'' of the forest. Note that this has no analogue in words --
there, simply looking at the leaves of the decomposition
suffices to track the infixes.

In particular, when we wish to restrict a binary decomposition to use a specific set of contexts (e.g., when inductively decomposing a forest over an idempotent element),
we apply this requirement to the linking contexts as
well. 
\begin{definition}[Stable Decomposition]
    \label{def:abs:stable decomp}
    For a stable context set $\wHole{V}\subseteq \tilV$, we say that a decomposition $\tau$ is \emph{stable over $\wHole{V}$} if $\clnk(\tau,x,y)\in \wHole{V}$ for every $x,y$ such that $x@*\to y$.     
\end{definition}
Observe that in a decomposition without inherited references, the linking contexts are compositions of contexts that appear as $\tau.\ctx(\cdot)$. Since we assume a \emph{stable} context sets, such compositions are already $\wHole{V}$. We therefore have the following.
\begin{lemma}[\cref{lem:no inherited decomposition is stable}]
\label{lem:abs:no inherited decomposition is stable}
    Every decomposition $\tau$ with no inherited references is stable.
\end{lemma}
As an example, observe that the decomposition in \cref{fig:abs:TBF_I_decomposition} has no inherited references, and is therefore stable. Moreover, as mentioned in \cref{xmp:tbf prelim forest algebra}, all the contexts there map to a single idempotent element. 

On top of restricting the context set $\wHole{V}$, in order to support inductive arguments we sometimes restrict the set
of leaves that we allow to some set $B\subseteq \tilH$ called a \emph{basis}. 
We then call a decomposition $\tau$ \emph{full over $\wHole{V},B$} if
it is stable with respect to $\wHole{V}$ and $\tau.\frst(x)\in B$ for every $x$ with
$\tau.\type(x)=\lf$. We call $\tau$ \emph{partial over $\wHole{V},B$} if it can be
extended by further decomposing its leaves 
to a full decomposition over $\wHole{V},B$ (and in particular is
stable). We denote by $\fd(\wHole{V},B)$ and $\pd(\wHole{V},B)$ the sets of full and partial decompositions, respectively (see \cref{def:full decomposition}).




\subsection{Tree Rotations and Monotonicity (\cref{sec:rotations,sec:dec lemmas})}
\label{sec:abs:rotations}
The main tool we develop for working with binary decompositions, and for
constructing them, is \emph{rotations}. Intuitively, these provide a method for
locally changing the shape of a tree while maintaining some desirable properties
(in particular, no inherited references,
  hence stability by \cref{lem:abs:no inherited decomposition is
    stable}).
Fundamentally, rotations can be thought of as changing the order
in which we pluck out factors.
In this section we illustrate the two types of rotations, and formulate their
properties. We remark that despite their intuitive nature, proving these
properties requires a careful analysis of how decompositions are built, and is
the source of a lot of the technical challenges, presented in appendix.

Fix some context set $\wHole{V}$.
The first rotation, dubbed \emph{sequential rotation}, is illustrated in \cref{fig:abs:factorizations for TB BT rotation lemma}. For this rotation, we start with a forest $f=C_1C_2h$ where $C_1,C_2\in \wHole{V}$ and $h\in \tilH$, and we factor it first by plucking $h$ (leaving the residue $C_1C_2\square_h$), and then plucking $C_2\square_h$ from the residue, leaving a further residue of $C_1\square_{C_2h}$. 
We refer to this as a \emph{Bottom-Top} factorization (BT). 
The rotation involves changing the order of plucking to \emph{Top-Bottom} (TB), so that we first pluck $C_2h$, leaving residue $C_1\square_{C_2h}$, and then from $C_2h$ pluck $h$, leaving residue $C_2\square_h$.
\begin{figure}[ht]
    \centering
    \begin{subfigure}[t]{0.35\textwidth}
        \centering
        \includegraphics[width=\linewidth]{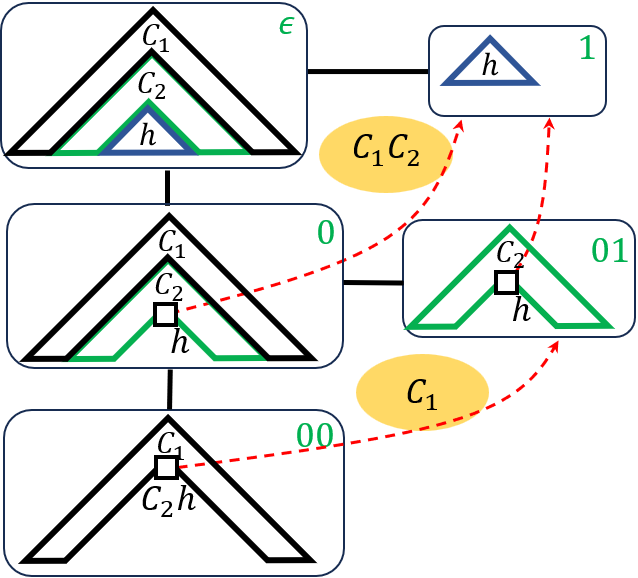}
        \caption{BT factorization.}
    \end{subfigure}
    \hfill
    \begin{subfigure}[t]{0.6\textwidth}
        \centering
        \includegraphics[width=\linewidth]{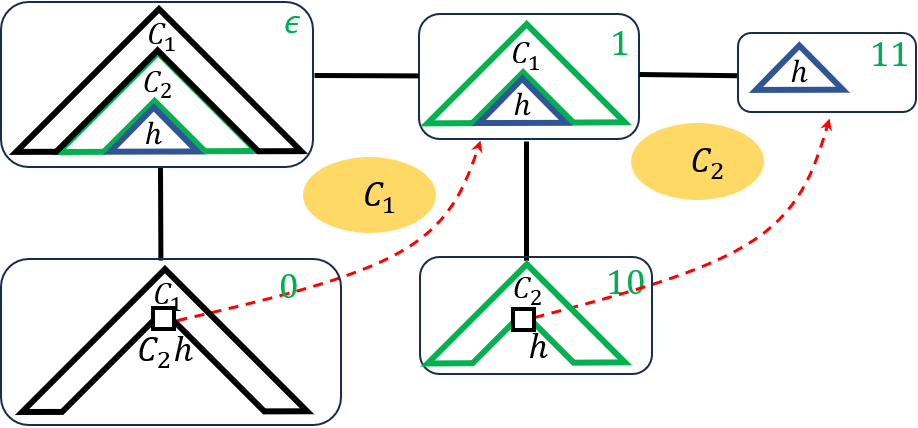}
        \caption{TB factorization}
    \end{subfigure}
    \caption{Two factorizations (BT and TB) for $C_1C_2h$. The translation between them is a \emph{sequential rotation}.}
    \label{fig:abs:factorizations for TB BT rotation lemma}
\end{figure}
It is easy to notice that both BT and TB have the same leaves. Further scrutiny
shows that in the BT decomposition
$\tau$, we have $\tau.01@*\to 1$, with $\clnk(\tau,01,1)=C_1$. In addition, there are no inherited references in the TB decomposition (within the five nodes under consideration). 
We can actually apply this rotation starting at a node $x$ that might not be the root. 

To discuss such tree manipulations, we denote by $\equpto{\tau}{x}{\tau'}$ the fact that two decompositions are identical except for the subtree under the node $x$.
Then, the properties of BT to TB, and its dual TB to BT, are summarized as follows.
\begin{lemma}[BT to TB -- informal]
    \label{lem:abs:BT to TB}
    Consider a binary decomposition $\tau$ and nodes $x,x0,x1,x00,x01$ in a BT formation. Then there exists a binary decomposition $\equpto{\tau'}{x}{\tau}$ such that:
    \begin{itemize}
        \item $\tau'.\ctx(x1) = \tau.\ctx(x01)$
        and $\tau'.\ctx(x11) = \clnk(\tau,x01,x1)$
        \item $\tau'[x0]= \tau[x00]$, $\tau'[x10]=\tau[x01]$, and $\tau'[x11]=\tau[x1]$.
    \end{itemize}
\end{lemma}
\begin{lemma}[TB to BT -- informal]
    \label{lem:abs:TB to BT}
    Consider a binary decomposition $\tau$ and nodes $x,x0,x1,x11,x10$ in a TB formation. Then there exists a binary decomposition $\equpto{\tau'}{x}{\tau}$ such that:
    \begin{itemize}
        \item $\tau'.\ctx(x1)=\tau.\ctx(x1)\cdot \tau.\ctx(x11)$
        and $\tau'.\ctx(x01)=\tau.\ctx(x1)$
        \item $\tau'[x00]= \tau[x0]$, $\tau'[x01]=\tau[x10]$, and $\tau'[x1]=\tau[x11]$.
    \end{itemize}
\end{lemma}
In \cref{sec:parallel rotation} we consider an additional type of rotation called \emph{parallel rotation} and abbreviated RL to LR, which deals with two subtrees that are not nested. 


\subsection{Next Contexts and Monotonicity (\cref{sec:dec lemmas})}
\label{sec:abs:next contexts}
Equipped with the framework of rotations and stability, we can finally devise a notion of monotonicity that takes us a step closer to the proof. Here we present an abridged version, see \cref{sec:dec lemmas} for the complete picture.

Consider a partial decomposition $\tau\in \pd(\wHole{V},B)$ and a node $x$. We define the \emph{Next Contexts} of $x$ to be the set of contexts that can be used to factor $x$ such that the extended decomposition can still be completed to a full decomposition over $\wHole{V}$ and $B$ (\cref{def:next contexts}). 
\[
\ns(\tau,\wHole V, B, x)= \{\tau'.\ctx(x\cdot1): \equpto{\tau}{x}{\tau'},\ \tau'\in \pd(\wHole V, B),\ \tau'.\type(x)=\B\}
\]
When clear from context, we write $\ns(\tau,x)$ for that set.
With these notations, monotonicity is captured as follows (see \cref{cor:decomp semantically decrease}).
\begin{lemma}
    \label{lem:abs:decomp semantically decrease}
        Consider a morphism $\varphi\colon\wHole{V}\to V$ for some semigroup $V$.
    \begin{enumerate}
        \item $\varphi[\ns(\tau, x0)] \subseteq \varphi[\ns(\tau, x)]$
        \item $\varphi(\tau.\ctx(x1))\cdot \varphi[\ns(\tau, x1)] \subseteq \varphi[\ns(\tau, x)]$
    \end{enumerate}
\end{lemma}
\begin{proof}[Proof Sketch]
  Consider a node $x$ and some $C\in \ns(\tau,x0)$. Ideally, we would want to show a stronger claim than the statement, namely that 
  $C\in \ns(\tau,x)$. We demonstrate such a case.
  By definition, we can extend $\tau$ to $\tau'$ so that
  $\tau'.\ctx(x01)=C$. We then split our reasoning according to which type of
  local structure (among BT, TB and LR)
  is below $x$.  For example, if the
  structure is BT, then by rotating $\tau'$ to a TB structure $\tau''$ as per
  \cref{lem:abs:BT to TB} we have $\tau''.\ctx(x1)=\tau'.\ctx(x01)=C$ (see the
  context $C_1$ in \cref{fig:abs:factorizations for TB BT rotation
    lemma}). Since $\equpto{\tau''}{x}{\tau}$, we have
  $C\in \varphi[\ns(\tau, x)]$, and in particular $\varphi(C)\in \varphi[\ns(\tau, x)]$.

    For the rotation LR to RL, we cannot actually ``move'' the same context $C$ from $x01$ to $x1$, but rather obtain at $x1$ a different context $C'$ such that $C'\eqUnr C$. But since $\varphi$ is oblivious to unraveling, the result still follows.
\end{proof}

A particularly useful application of rotations is to obtain stable decompositions. To do 
so, we define a canonical way for decomposing a node, based on an order on contexts. 
For two contexts $C_1,C_2\in \wHole{V}$, write $C_1\cleq C_2$ if $C_2=C_1\cdot C$ for some $C$, i.e., if $C_1$ is a ``prefix'' of $C$ (see \cref{def:context prefix order}).
Then, we say that a partial decomposition $\tau\in \pd(\wHole{V},B)$ is \emph{minimal below node $x$} if for every descendant $y$ of $x$ we have that $\tau.\ctx(y1)$ is $\cleq$-minimal among $\ns(\tau,y)$.
In \cref{lem:min full dec} we show that every partial decomposition can be extended below any node to a minimal decomposition.

Using rotations, we can show that a minimal rotation has a simple structure, in the following sense.
\begin{lemma}[Minimal Decompositions have No Inherited References]
\label{lem:abs:min dec no inherited}
Consider a decomposition $\tau$ that is minimal below $x$, then $\tau$ has no inherited references in the subtree of $x$.
In particular, $\tau$ is stable below $x$ (by \cref{lem:abs:no inherited decomposition is stable}).
\end{lemma}
\begin{proof}[Proof Sketch]
The complete proof is involved (see \cref{lem:min factorization dec no references}). However, as ``supporting evidence'', observe that the BT to TB rotation in \cref{fig:abs:factorizations for TB BT rotation lemma} eliminates the inherited reference $01@*\to 1$, and replaces the context $C_1C_2$ with the smaller $C_1$, progressing towards minimality.
\end{proof}

This completes our set of tools for binary decompositions. Note that for the setting of words, all the tools developed thus far are either irrelevant or follow trivially from the definition of concatenation and morphism (in particular, \cref{lem:abs:decomp semantically decrease}).

\section{General Decompositions}
\label{sec:abs:general decompositions}
In this section we introduce our notion of general decompositions. As in the setting of words, the main addition on top of binary decompositions are $\I$dempotent nodes: nodes with unbounded arity that encapsulate a decomposition where all the contexts map to a single idempotent.
However, for technical convenience we also introduce a new type of nodes called $\C$entipedes. 
Intuitively, centipedes encapsulate a decomposition of several non-nested subforests. More precisely, a centipede is a binary decomposition whose nodes are in $0^*(\epsilon+1)$, and it does not have inherited references (see \cref{fig:centipede node} for an illustration). We remark that centipede nodes can also be introduced in the setting of words. We discuss this in \cref{sec:abs:discussion}.

We start with the main definition (\cref{def:general decomposition}).
\begin{definition}[General Decomposition -- informal]
\label{def:abs:general decomp}
A \emph{general decomposition} with respect to a morphism $\varphi\colon\wHole{V}\to V$ is a
labeled tree $\tau$ where each node $x$ has a type $\tau.\type(x) \in \{\lf,\B,\I,\C\}$
and is labeled with a forest $\tau.\frst(x)\in \tilH$ and context $\tau.\ctx(x)\in
\wHole{V}$, and the following conditions hold.
\begin{itemize}
    \item $\lf$ and $\B$ nodes follow the conditions of binary decompositions.
    \item An $\I$-node $x$ has a set of children $S=\{x0,\ldots, xm\}$, such that there is an idempotent element $v\in V$ and a full decomposition $\tau'\in \fd(\varphi^{-1}[v],\tau.\frst[S])$ of $\tau.\frst(x)$ whose leaves bijectively match the children of $x$ in $\tau$.
    \item A $\C$-node $x$ has a set of children $S=\{x0,\ldots, xm\}$, such that there is a centipede decomposition $\tau'$ of $\tau.\frst(x)$ whose leaves bijectively match the children of $x$ in $\tau$.
\end{itemize}
$\tau$ is \emph{full} with respect to a basis $B\subseteq \tilH$ if the forests at  its leaves are all in $B$. We denote the set of full general decompositions of a forest $f$ by $\dec(\wHole{V},B,f,\varphi)$.
\end{definition}
Note that the requirement for $\I$-nodes implies that $\tau'$ is a stable
decomposition. Apart from facilitating proofs, this requirement makes intuitive
sense: an $\I$-node captures an idempotent ``zone'' in the decomposition. In
particular, we expect that any relationships between any two related nodes in
this decomposition also map to the idempotent $v$, not just immediate
contexts. This is captured by stability (\cref{def:abs:stable decomp}),
which requires that all the \emph{linking} contexts also map to $v$.

\begin{example}[TBF Full General Decomposition]
\label{xmp:abs:TBF full general}
In \cref{fig:abs:TBF_full} we present a full general decomposition of a forest in the TBF setting. The binary decomposition in \cref{fig:abs:TBF_I_decomposition} is over an idempotent element, and thus we fold it to an $\I$-node at the root. Notice that the leaves $0,1,2,3,4$ of \cref{fig:abs:TBF_full} bijectively match the leaves $00,010,0111,0110,1$ of \cref{fig:abs:TBF_I_decomposition}. 

The remainder of the composition uses binary nodes. Note that the leaves of the decomposition cannot be decomposed any further. Indeed, they are in the \emph{standard basis}, defined shortly.
\end{example}

The basis $B$ of the decomposition is used for inductive arguments. However, some special attention is given to the \emph{standard basis} $B_A\subseteq \tilH$ given by 
$\bs = A\cup \{\sigma\square_h:\sigma\in A,h\in \hf\}\cup \{\square_{h_1}+\square_{h_2}:h_1,h_2 \in \hf\}$.
Note that the elements in $\bs$ are indeed those that cannot be decomposed further. 

We can now define decomposition bounds.\footnote{The precise definition is more nuanced, see \cref{def:semigroup decomp bound}.} Let $S$ be a semigroup and $\varphi\colon\wHole{V}\to S$ a morphism. 
For a basis $B$ we define 
$\norm{\varphi, B}\coloneqq \sup_{f\in \tilH} \min_{\tau\in \dec(\wHole{V},B,f,\varphi)}\depth(\tau)$ 
and by extension $\norm{\varphi}=\sup_{B}\norm{\varphi,B}$.

As mentioned above, $\C$-nodes are a technical convenience. Indeed, the following result shows that we can eliminate them, at the cost of increasing the depth of a decomposition exponentially. Note that for our main result, this depth is constant, and therefore this exponential is still a constant.
\begin{lemma}[See \cref{lem:elimination of C nodes}]
\label{lem:abs:elimination of C nodes}
For every $\tau\in \dec(\wHole{V},\bs,f,\varphi)$ there exists $\tau'\in
\dec(\wHole{V},\bs,f,\varphi)$ with no centipede nodes, and such that $\depth(\tau')\leq
2^{\depth(\tau)+1}-2$.
\end{lemma}


\section{Bounded Decompositions}
\label{sec:abs:bounded decomp}
As mentioned in \cref{sec:abs:intro}, our bounded-decomposition theorem holds under a semantic condition, which we now present.

\subsection{$\mr$-Alignment}
\label{sec:abs:mr alignment}
\begin{definition}[$\mr$-aligned Forests and Morphisms]
\label{def:abs:unamb property}
Consider a semigroup $S$, a stable context set $\wHole{V}$ and a morphism $\varphi\colon\wHole{V}\rightarrow S$.
A forest $f\in \tilH$ is \emph{$\mr$-aligned} over $\varphi$ if
for every $f_1,f_2\in\tilH$ and $C_1\neq C_2\in \wHole{V}$ such that $f=C_1(f_1)=C_2(f_2)$ and $\varphi(C_1)\mj \varphi(C_2)$ there exists $r\in S$ such that $r\mj \varphi(C_1)$ and $\varphi(C_2)=\varphi(C_1)\cdot r$. 

A morphism $\varphi$ is said to be \emph{$\mr$-aligned} if all forests $f\in \tilH$ are
$\mr$-aligned over \(\varphi\).
\end{definition}
Note that since the requirement in the definition holds for every $C_1,C_2$, it is in fact symmetric: there is also $r'\mj \varphi(C_1)$ such that $\varphi(C_1)=\varphi(C_2)\cdot r'$. 

Intuitively, the definition states the following: suppose we decompose $f$ in two different ways (by plucking out different sub-forests), such that the remaining contexts are $\mj$-equivalent, then not only are these contexts also $\mr$-equivalent, but they are $\mr$-equivalent using multiplicands that are within the same $\mj$-class (hence the term $\mr$-aligned). 

The requirement that $r\mj \varphi(C_1)$ allows us to show that $\mr$-alignment survives through taking certain quotients and subsemigroups, as follows.
\begin{lemma}[\cref{lem:rees unambiguous,lem:min ideal unambiguous}]
\label{lem:abs:R aligned quotient subsemigroup}
    Consider an $\mr$-aligned morphism $\varphi\colon\wHole{V}\to S$ and an ideal $I\subseteq S$, then the following hold.
    \begin{itemize}
        \item Let $\quo_I\colon S\to S/I$ be the quotient morphism, then $\quo_I\circ \varphi\colon\wHole{V}\to S/I$ is $\mr$-aligned.
        \item If $I$ is a $\zeroel$-minimal ideal\footnote{See \cref{sec:green prelim} for the definition. Here it can be taken as ``some condition on $I$''.}, let $\wHole{V'}=\varphi^{-1}[I]$  then $\varphi|_{\wHole{V'}}\colon\wHole{V'}\to I$ is $\mr$-aligned.
    \end{itemize}
\end{lemma}

We can finally state our main theorem.
\begin{theorem}
    \label{thm:abs:main}
    Consider a finite semigroup $S$ and an $\mr$-aligned morphism $\varphi\colon\wHole{V}\to S$, then $\norm{\varphi}\le 4|S|-3$.
\end{theorem}
The proof of \cref{thm:abs:main} is by induction, with three base cases depending on the structure of $S$: a group, a simple or \zsimple semigroup, or a null semigroup. The latter is less involved and appears in \cref{sec:null semigroups}. We turn to sketch the main ideas of the interesting cases.
\subsection{Decomposition over a Group}
\label{sec:abs:group decomp}
\begin{lemma}
\label{lem:abs:group decomposition}
For a group $V$ we have $\norm{V}\le 2|V|-1$. 
\end{lemma}
\begin{proof}[Proof Sketch]
The proof is by induction over $|\ns(\tau,x)|$ at each node $x$ of a decomposition (see \cref{sec:group decomposition} for details).
Let $f\in \tilH$, we start by decomposing $f$ ``as much as possible'' using $\varphi^{-1}[\unitel]$ (i.e., over the unit of $V$, which is idempotent). We can then fold this decomposition to a single $\I$-node. Using \cref{lem:abs:decomp semantically decrease}, we can show that the children of this $\I$-node cannot be decomposed with $\unitel$, and therefore are decomposable over a smaller set of contexts, so continue by induction.

If $f$ is not decomposable over $\varphi^{-1}[\unitel]$, we start by decomposing $f$ using a maximal-depth centipede decomposition $\tau$, such that $\tau$ is \emph{minimal}. We then use the monotonicity of \cref{lem:abs:decomp semantically decrease} together with the properties of groups and the centipede structure to conclude that each leaf $y$ of the centipede either has a smaller $\ns(\tau,y)$ set (so we can apply induction), or $\unitel\in \ns(\tau,y)$, and we apply the case  proved above. Then, we can fold the centipede to a single $\C$-node, yielding a full general decomposition of the desired depth.    
\end{proof}

\subsection{Decomposition over Simple and $\zeroel$-Simple Semigroup}
\label{sec:abs:simple decomp}
This case is the most challenging part of our contribution. To frame it, we recall a central part of the proof of Simon's theorem~\cite{FactorForBoj}. In broad strokes, a word $w$ is written as $w=u_0abu_1abu_2ab\cdots abu_k$ for some letters $a,b$ such that the $u_i$ do not contain the subsequence $ab$. Then, $w$ is viewed as $w=u_0a(bu_1a)(bu_2a)\cdots (bu_{k-1}a)bu_k$ and it is noticed that all the $bu_ia$ have the same $\mr$- and $\ml$-classes, and thus belong to the same $\mh$-class, which is moreover idempotent (as an element in a quotient semigroup).
The takeaway for us is that we wish to split a forest by consecutive elements that dictate a pair $(L,R)$ of $\ml$- and $\mr$-classes. 

Unfortunately, for forests this approach becomes much more complicated. In fact, it generally breaks down (as we show in \cref{sec:abs:lower bound}), and is the reason for the introduction of $\mr$-alignment. Importantly, we notice that for simple semigroups without $\zeroel$, there is a single $\mj$-class. Thus, for an $\mr$-aligned morphism, for every $C,C'\in \tilV$ and $f,f'\in \tilH$ such that $C(f)=C'(f')$, we have $\varphi(C)\mr \varphi(C')$. 

We now give an overview of our approach.
The first point to consider is that for forests, the notion of ``two consecutive letters'' makes only partial sense. Consider a node $x$ in a binary decomposition $\tau$. Similarly to $\ns(\tau,x)$, we can consider the next \emph{two} contexts, i.e., $\ctx(x1)$ and $\ctx(x11)$. Since we are interested in $\ml$- and $\mr$-classes, we define $\lr(\tau,x)$ to be the set of $(L,R)$ pairs such that $L=[\tau'.\ctx(x1)]_{\ml}$ and $R=[\tau'.\ctx(x11)]_{\mr}$ in some extension $\tau'$ of $\tau$. That is, the possible ``consecutive'' $(L,R)$ classes that can be used to decompose $x$.

Using \cref{lem:abs:decomp semantically decrease} and results about Green's relations we show in \cref{lem:lr decreasing} a monotonicity property of $\lr(\tau,x)$: if $y$ is a descendant of $x$ then $\lr(\tau,y)\subseteq \lr(\tau,x)$.
We then proceed to show the bound on decomposition depth.
\begin{lemma}[Decomposition over Simple Semigroup -- abridged]
    \label{lem:abs:simple decomposition}
    Let $\varphi\colon\wHole{V}\to S$ with $S$ a simple semigroup and $\varphi$ $\mr$-aligned. Then $\norm{\varphi}\le 2|S|+1$.
\end{lemma}
\begin{proof}[Proof Sketch]
    Let $f\in \tilH$. 
    As mentioned above, we consider the case that $S$ has no $\zeroel$. If it does, we start by factoring over $\varphi^{-1}[\zeroel]$, similarly to the case of \cref{lem:abs:group decomposition} (see \cref{lem:simple decomp}).

    The proof proceeds by induction over $|\lr(\tau,x)|$, assuming an existing decomposition $\tau$ and a node $x$, where we prove that we can extend $\tau$ below $x$ with low-enough depth. Thus, consider a node $x$, and fix a pair $(L,R)\in \lr(\tau,x)$.
    
    We now decompose $f$ in two parts (see \cref{fig:abs:simple semigroup}).
    First, we decompose $f$ with a maximal-depth centipede $\tau_L$ over the contexts $\varphi^{-1}[L]$, with an additional crucial requirement: that every right-child (i.e., ``leg'' of the centipede) $z1$ has $\tau_L.\ns(z1)\cap R\neq \emptyset$. That is, we can continue decomposing $z1$ with a context mapped to $R$. This captures the analogy of an $(L,R)$ pair for words.

    Then, for every such leaf $z1$ we decompose it using a maximal-depth binary decomposition $\tau_z$ with the following property: $\tau_z$ is minimal (as in \cref{lem:abs:min dec no inherited}) over $L\cap R$, and for every node $y$ in $\tau_z$, every way of factoring $\tau_z.\frst(y)=C'(f')$ satisfies $\varphi(C')\in R$.
    Using this latter condition we show two significant properties: 
    \begin{itemize}
        \item $L\cap R$ is a group, and therefore we can use the group bound from \cref{lem:abs:group decomposition} to replace $\tau_z$ with a low-depth decomposition with the same leaves.
        \item Every \emph{leaf} $w$ of $\tau_z$ has $(L,R)\notin \lr(\tau_z,w)$. 
        To show this we rely on $\mr$-alignment twice: once for showing that any way of factoring a leaf of $w$ must start with a context mapped to $R$, and once for assuming that $(L,R)\in \lr(\tau_z,w)$ and reaching a contradiction. 
        We remark that for leaves of the form $w'1$ the former is trivial, by the property above. The crucial part is that due to $\mr$-alignment this also holds for $w'0$.
    \end{itemize}
    We can therefore decompose every leaf $w$ of $\tau_z$ inductively, since it has a smaller $\lr$ set. Then, we fold $\tau_L$ to a single $\C$-node, obtaining our desired decomposition.   
\end{proof}

\subsection{Decomposition over General Semigroups -- Main Theorem}
\label{sec:abs:main thm}
We can finally sketch the proof of \cref{thm:abs:main}. See \cref{sec:main theorem bound} for the details. 
The proof is by induction over $|S|$, where the base cases are when $S$ is a group (\cref{sec:abs:group decomp}), simple or \zsimple semigroup (\cref{sec:abs:simple decomp}), or null semigroup (\cref{sec:null semigroups}). 

Consider a forest $f$. The inductive argument follows a similar scheme to the proof of Simon's theorem for words. We consider a nonzero element $a\in S$ that is $\mj$-minimal and its corresponding ideal $M=S^1aS^1$. We show that this is a minimal non-zero ideal of $S$ with respect to containment, and that $|M|>1$. 
Then, by \cref{lem:abs:R aligned quotient subsemigroup} we have that $\quo_M\circ \varphi$ is $\mr$-aligned. Moreover, $|S/M|<|S|$. We can therefore apply the induction hypothesis and obtain a decomposition $\tau$ of $f$ over $S/M$. 
We further notice that $S/M$ is fairly similar to $S$, with the only difference being that $M$ is collapsed to a single element. However, this element is idempotent and is $\zeroel_{S/M}$. By carefully analyzing the constructions of the base case, we can actually show that there is at most one $\I$-node mapped to $M$ in $\tau$. Thus, all we need to do is replace this $\I$-node by a decomposition over $S$, with bounded depth.

To this end, we show that $M$ is either null (in which case we are done), or $|M|<|S|$, in which case we show that it is a $\zeroel$-minimal ideal, so by the second part of \cref{lem:abs:R aligned quotient subsemigroup} we can apply the induction hypothesis to conclude the theorem. \hfill \qed

\section{Additional Results, Discussion and Future Research}
\label{sec:abs:discussion}
In this overview we present the main ideas of our analogue of Simon's factorization theorem for forests. We briefly mention some contributions that appear only in the technical appendix. 
\subparagraph*{A Decomposition Reduction Scheme}
In \cref{sec:reduction} we present the following general result (not relying on $\mr$-alignment):
\begin{lemma}[Reduction Lemma -- abridged]
\label{lem:abs:reduction}
Consider two finite semigroups $V_1,V_2$ and a semigroup morphism $\varphi\colon V_1\to V_2$, then for every context set $\wHole{V}$ and a morphism $\beta\colon\wHole{V}\to V_1$ we have 
 $\norm{\beta} \leq \norm {\phi\circ \beta} \cdot \max_{\overset{e\in V_2}{e^2=e}}\norm{\varphi^{-1}[e]}$.
\end{lemma}
This result allows us to obtain bounds on decomposition depth for one morphism ($\beta$) based on known bounds for e.g., quotient semigroups ($\varphi\circ \beta$), assuming knowledge of decompositions over the idempotent subsemigroups. We remark that a similar result is used in the proof of Simon's result for words. 

\subparagraph*{A Lower Bound without $\mr$-Alignment}
\label{sec:abs:lower bound}
A natural question is whether $\mr$-alignment is a necessary condition to obtain bounded-depth decompositions. The answer to this is no -- it is not hard to check that the TBF example from \cref{xmp:abs:TBF1} is not $\mr$-aligned. However, it is possible to show that it has bounded decompositions by following the proof of \cref{thm:main} and noticing that after getting rid of prefixes that map to 
a ``sink-element using a depth-2 decomposition, the remaining forests do satisfy $\mr$-alignment.
Thus, the failure of $\mr$-alignment in TBF is an edge case, and is not fundamental. 

The more important question is whether any condition is even necessary. In this section, we show that without any restrictions on $\varphi$, it might not admit bounded decompositions.
\begin{theorem}[see \cref{cor: counter example}]
    \label{thm:abs:lower bound}
    There is a morphism $\varphi\colon\wHole{V}\to V$ for some $V$ that is not $\mr$-aligned, such that $\norm{\varphi}=\infty$.
\end{theorem}
We refer to \cref{sec:lower bound} for the details, and bring here only the rough idea, without describing the concrete counterexample.
Intuitively, we construct a morphism based on a tree-language of binary trees, with the following property: each idempotent element $v$ is associated with a \emph{side} -- right or left, such that when decomposing a binary tree using a context $C\in \varphi^{-1}[v]$, only the right or left subtree of the tree (according to the associated side) is being decomposed. 

This means that in a balanced tree, each $\I$-node can decompose at most half the tree. Thus, a tree of depth $n$ requires a decomposition of depth at least $\log n$, concluding the claim.

\subsection{Discussion}
Absent from this work are applications of bounded decompositions for forests. One immediate application is an analogue of fast infix query for words~\cite{FactorForBoj}, where we want to compute the image of an ``infix'' of a forest under a morphism. Note that proper infixes of forests are contexts rather than forests, and formulating precisely what an infix is is nontrivial. Unfortunately, this is not useful from an algorithmic approach, since \emph{computing} a decomposition tree is not clearly efficient. Indeed, following our proof, it requires computing $\ns(\tau,x)$ for certain nodes, and it is not clear how to do that. 
This, however, is not surprising, as the same difficulty arises in words. To overcome this, an approach built on \emph{forward Ramsey decomposition} is introduced in~\cite{Colcombet07,Colcombet2021}.
In fact, a similar approach using forward factorizations has been introduced for forests as well, yielding a solution for infix querying on forests~\cite[Appendix E]{Bojaczyk2010}. 

Nonetheless, most applications of Simon's factorization are in proofs (e.g., boundedness of weighted automata~\cite{simon1994semigroups}), where we hope our work would enable lifting results on words to forests. We remark that as a sanity check, a naive application is as a simple pumping argument: given a general decomposition with an $\I$-node, we can modify the binary decomposition that corresponds to this $\I$-node by repeating a binary node (e.g., replacing a forest $C(f)$ with $C(C(f))$). This readily yields a decomposition of a different, ``larger'' forest, with the same image under the morphism.

\bibliography{ref}

\section{Technical Appendix -- Organization}
\label{sec: appendix organization}
In the following we present the details of our construction and proof. This is organized as follows.
In \cref{sec:prelim} we define basic terminology. In \cref{sec:forest contexts and algebras,sec:augmented forest algebra} we give the precise definition of forest algebras, and develop an equivalent formulation via forest domains. This highly-technical part is the foundations of the ``bookkeeping'' mentioned in \cref{sec:abs:intro}. In \cref{sec:bin decomp} we formally introduce binary decompositions, references and additional structures used throughout the proof. 
In \cref{sec:rotations} we develop rotations and prove their correctness. We then use them in \cref{sec:dec lemmas} to obtain the monotonicity results for binary decompositions. The reader may notice that these results are more elaborate than those sketched in \cref{sec:abs:next contexts}, and refer to additional definitions.

In \cref{sec:general decomp} we define general decompositions and decomposition bounds, and prove that $\C$-nodes can be eliminated.
In \cref{sec:green prelim} we present another set of preliminary results about Green's relations, used in the next sections. Then, in \cref{sec:decomposition theorems,sec:main theorem bound} we present $\mr$-alignment, the bases for the main inductive proof, and the proof itself. In particular, in \cref{sec:reduction} we present the reduction scheme mentioned in \cref{sec:abs:discussion}. 
We conclude in \cref{sec:lower bound} with the lower bound on decompositions without $\mr$-alignment.

\section{Preliminaries Part I -- Basic Definitions}
\label{sec:prelim}
We begin with some basic mathematical notations, definitions and results which we use throughout the paper. In \cref{sec:green prelim} we present another set of preliminary results pertaining to Green's relations, that are relevant for later parts of the paper.
\subparagraph*{Sets and Relations}
Let $\bbN$ be the set of non-negative integers. For a set $A$ let $A^*$ and $A^{+}$ be the sets of finite words and finite non-empty words over $A$. Denote by $\epsilon$ the empty word. For $x,y\in A^*$ we write $x\le y$ if $x$ is a prefix of $y$ and $x<y$ if $x\le y$ and $x\neq y$.

Consider sets $A,B$ and a relation $R\subseteq A\times B$. We interchangeably denote $(a,b)\in R$ or $aRb$ for elements $a\in A$ and $b\in B$.
For a set $S\subseteq A$, we denote $R[S]\coloneqq \{b\in B: \exists s\in S\ (sRb)\}$
And for some $a\in A$ define  $R[a]\coloneqq R[\{a\}]$\footnote{we only use this notation when it is clear that $a\in A\wedge \{a\}\notin A$},
In addition define $R^{-1}\coloneqq \{(b,a): aRb\}$
and
$\dom(R)\coloneqq R^{-1}[B]=\{a\in A: \exists b \ aRb\}$.
For an \emph{injective} function $f\colon A\rightarrow B$ and an element $b$ in the image of $f$ we denote by $f^{-1}(b)$ the unique element in $f^{-1}[b]$.

\subparagraph*{Semigroups, Monoids and Morphisms}
A \emph{semigroup} is an algebraic structure $\tup{S,\cdot}$ such that
$S$ is a set and $\cdot\colon S\times S\to S$ is an associative operation. We sometimes write $uv$ instead of $u\cdot v$ for $u,v\in S$, when the operation is clear. We also sometimes omit mentioning the operator (e.g., we refer to ``a semigroup $S$''). Unless stated otherwise, we refer to the operation generically as ``multiplication''.

For an element $x\in S$ and a number $i\in \bbN$ we denote the \emph{power} $x^i=\underbrace{x\cdot x \cdots \cdot x}_{i\text{ times }}$.
An element $e\in S$ is called \emph{idempotent} if it satisfies $e\cdot e = e$.
A simple but fundamental result about \emph{finite} semigroups states that every element has some power which is an idempotent element. Moreover, there is such a uniform power over $S$. 
\begin{lemma}[Idempotent Power (See e.g.,~\cite{pin2010mathematical})]
\label{lem:idm power}
Consider a finite semigroup $\tup{S,\cdot}$. For every $x\in S$ there exists $i\in \bbN$ such that $x^i$ is idempotent.

Moreover, there is a number $\omega_S\in \bbN$ called the \emph{exponent} of $S$ such that for every $x\in S$, the element $x^{\omega_S}$ is idempotent.
\end{lemma}

For a Semigroup $S$, an element $z\in S$ is a \emph{zero} of $S$ if $m\cdot z = z\cdot m=z$ for every $m\in S$. We typically denote such an element by $\zeroel$.
Notice that a Semigroup has at most a single zero, as if both $u,v\in S$ are zeroes we get $u=u\cdot v=v$. We emphasize that not all Semigroups have a zero (for example $\tup{\bbN,+,0}$ -- the element $0$ is an identity for addition, and it is not a zero).

A \emph{monoid} is a semigroup with an identity element $\tup{S, \cdot, \unitel}$ 
satisfying $m\cdot\unitel = \unitel\cdot m = m$ for every $m\in S$.

Fix some finite alphabet $A$.  The \emph{free monoid} 
$(A^*,\cdot)$ is the monoid all finite-length words over $A$ with concatenation and $\epsilon$ (the empty word) as the identity.


A \emph{morphism} is a mapping between two algebraic structures which preserves their operations. Specifically, for two semigroups $\tup{S_1,\oplus},\tup{S_2,\otimes}$ a semigroup morphism is a mapping 
$\varphi\colon S_1 \rightarrow S_2$ such that for every $m_1,m_2\in S_1$ we have 
$\varphi(m_1\oplus m_2) =\varphi(m_1)\otimes\varphi(m_2)$.

\section{Forests, Contexts, and Forest Algebras}
\label{sec:forest contexts and algebras}
Our main object of study is forest algebras~\cite{Bojanczyk2008ForestA}, which we now recall. 
Intuitively, a forest is built from a finite alphabet using two operations: addition of sub-forests (i.e., putting the two forests ``side by side''), and connecting a forest's roots as the children of a node, thus obtaining a tree. However, this approach only allows a ``bottom-up'' approach for constructing forests. The novelty in forest algebras is that they allow composition of forests in various ways, while retaining a nice algebraic structure. Technically, this is done by adding another algebraic structure called \emph{contexts}, where a context is a forest with a ``hole'' ($\square$) in a leaf. Then, given a context, we can plug into that hole another context, or a forest. 
We proceed to give the formal details of this framework.

Fix some finite alphabet $A$. We think of the letters of $A$ as \emph{unranked}, meaning that a letter can be the root of a tree of arbitrary arity. We designate two special symbols: $\az$ is the \emph{empty forest} (intuitively, a forest with no nodes), and $\square$ is a \emph{hole} (intuitively, a placeholder where a forest may be inserted). We assume $\az,\square\notin A$. We define the syntax of our two fundamental objects: \emph{forests} and \emph{contexts}. 

\begin{definition}[Forests and Contexts]
\label{def:forests and contexts}
    The set of \emph{forests} over $A$, denoted $\frs A$ is defined inductively as follows. 
    \begin{description}
        \item[\textbf{Base:}] $\sigma\in \frs A$ for every $\sigma\in A\cup \{\az\}$.
        \item[\textbf{Addition:}] if $h_1,h_2\in \frs A$ then $h_1+h_2\in \frs A$.
        \item[\textbf{Composition:}] if $h\in \frs A$ and $\sigma\in A$ then $\sigma(h)\in \frs A$.
    \end{description}

    The set of \emph{contexts} over $A$, denoted $\cts A$ is defined inductively as follows. 
    \begin{description}
        \item[\textbf{Base:}] $\square\in \cts A$.
        \item[\textbf{Addition:}] if $C\in \cts A$ and $h\in \frs A$ then $C+h\in \cts A$ and $h+C\in\cts A$.
        \item[\textbf{Composition:}] if $c\in \cts A$ and $\sigma\in A$ then $\sigma(C)\in \cts A$.
    \end{description}
\end{definition}

\begin{example}
\label{xmp:forests and contexts}
Let $A=\{a,b\}$. We illustrate some forests and their addition in \cref{fig:forests example1}. Contexts are similar but have exactly one leaf labeled $\square$. For example, $a(\square)$, $b(a+\square)+a$, and $\square$.

To drive the point in, the following are neither forests nor contexts:
\begin{itemize}
    \item $a(\square)+\square$ (only one $\square$ is allowed).
    \item $a+\square(b)+a$ (squares must be leafs).
\end{itemize}

\begin{figure}[ht]
    \centering
    \includegraphics[width=0.65\linewidth]{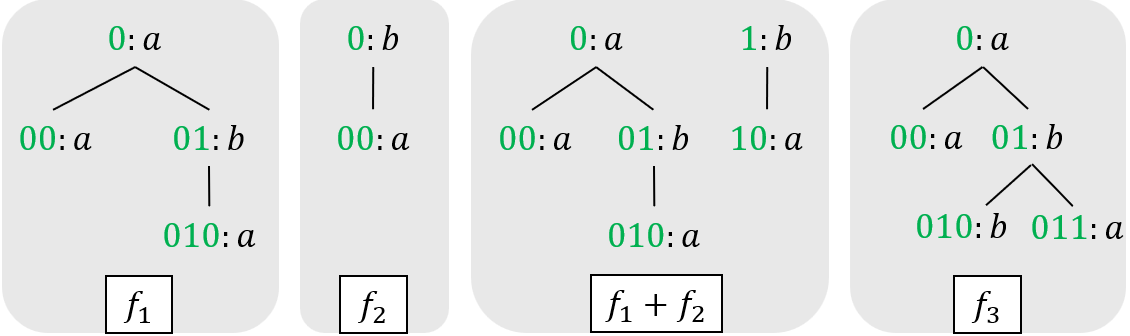}
    \caption{Examples of forests and their addition with $f_1=a(a+(b(a)))$,   $f_2=b(a)$, and $f_3=a(a+b(b+a))$. The numbers indicate the \emph{Address Domains} defined in \cref{sec:labeling intuitive}.}
    \label{fig:forests example1}
\end{figure}
\begin{figure}[ht]
    \centering
    \includegraphics[width=0.5\linewidth]{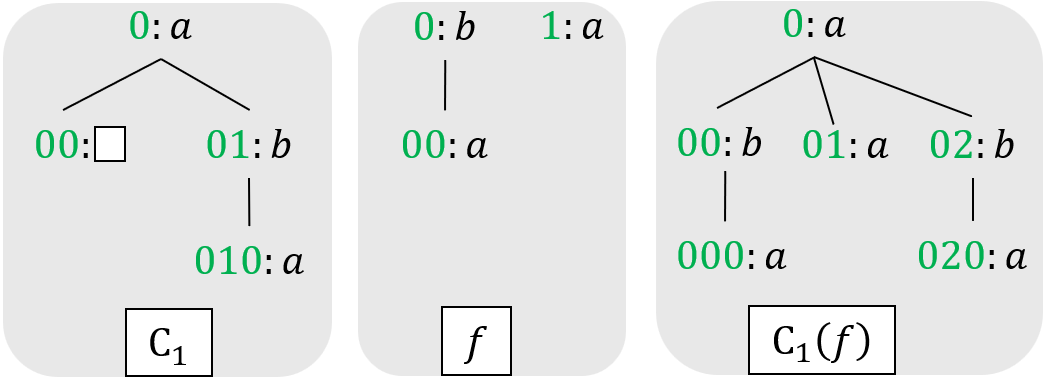}
    \caption{Examples of context composition with $C_1=a(a+(b(a)))$ and  $f_2=b(a)$. The numbers indicate the \emph{Address Domains} defined in \cref{sec:labeling intuitive}.}
    \label{fig:contexts example1}
\end{figure}
\end{example}
The formalism of contexts and forests is designed to allow the composition of a forest or a context with another context. For example, the composition of the context $a(\square)$ with $a+b$ yields the forest $a(a+b)$, since we replace the hole with the composed forest. In \cref{fig:contexts example1} we depict this idea.
\begin{definition}
We define the composition $C(h)$ inductively as follows.
\begin{description}
\label{def:context composition}
        \item[\textbf{Base:}] If $C=\square$ then $\square(h)=h$.
        \item[\textbf{Addition:}] If $C=C'+h'$ then $C(h)=C'(h)+h'$.  If $C=h'+C'$ then $C(h)=h'+C'(h)$.
        \item[\textbf{Composition:}] If $C=\sigma(C')$ then $C(h)=\sigma(C'(h))$.
\end{description}
We often denote $C(h)$ by $C\cdot h$.
\end{definition}

In order to reason about composition formally, and in order to prove our results, it is inconvenient to use the inductive definition, as it does not easily support pinpointing a specific node and tracking it through compositions.

Instead, we wish to refer to forests and contexts with an \emph{address scheme}, where we can point out certain elements, sub-forests, etc. We develop this viewpoint in the next section.

\subsection{Trees and Forests as Labeled Address Domains}
\label{sec:labeling intuitive}
The address scheme we use is standard and intuitive: we mark each node with a sequence of numbers $x\in \bbN^*$, such that the children of node $x$ are marked $x0,x1,\ldots$ (see depiction in \cref{fig:forests example1,fig:contexts example1}). 
We can then define a forest or context using a \emph{labeling function} that assigns to each address a label (e.g., a letter from $A$, or other information). 
\begin{definition}[labeled Address Domains]
    \label{def:labeling intuitive}
    Let $L$ be a non-empty set of labels.
    \begin{itemize}
        \item A \emph{Tree Domain} is a set $d\subseteq \bbN^*$ that represents the addresses of some tree (i.e., closed under left-sibling, and under parent). 
        That is, $\epsilon\in d$, and  if $x\cdot i\in d$ for $i>0$ then $x\cdot (i-1)\in d$ and $x\in d$.

        The set of all tree domains (explicitly including an empty domain $\emptyset$) is denoted $\tds$.
        \item A \emph{Forest Domain} is a set $d\subseteq \bbN^*$ that represents the addresses of some forest (i.e., closed under left-sibling, and under parent, except the root is not $\epsilon$). 
        That is, $\epsilon\notin d$, and  if $x\cdot i\in d$ for $i>0$ and $|x|>0$ then $x\cdot (i-1)\in d$ and $x\in d$.
        For every forest domain $d$, define $\width(d)\coloneqq \max\{d\cap \bbN\}$ (which corresponds to the number of a root in $d$)
        The set of all forest domains (including an empty domain $\emptyset$) is denoted $\fds$.

        \item An \emph{$L$-labeled} tree (\resp forest) domain is a pair $\tup{d,\lab}$ where $d\in \tds$ (\resp $d\in \fds$) and $\lab\colon d\to L$ is a labeling function. We denote the set of $L$-labeled tree (\resp forest) domains by $\lltds{L}$ (\resp $\llfds{L}$). We omit $L$ when it is clear from context.
        For $dl=\tup{d,l}\in \lfds$ denote $\width(dl)\coloneqq \width(d)$
    \end{itemize}
\end{definition}
In \cref{fig:forests example1} we illustrate some forest domains. Note that these are not tree domains, since $\epsilon$ is not the root. Observe that e.g., $\width(f_1+f_2)=1$ whereas $\width(f_1)=\width(f_2)=0$.

\begin{remark}[Tree v.s. Forest Domains]
    \label{rmk: tree vs forest domain}
    Despite the unified definition, our usages of tree domains and forest domains differ entirely: forest domains are used to model forests, whereas tree domains are used for decomposition trees -- the decomposition structures we develop. Conceptually, one can think of them as entirely different definitions.
\end{remark}

Labeled tree and forest domains naturally support certain structural operations (namely analogues of a recursive definition of forests), which we now define. 

\begin{definition}[Operations on Labeled Domains]
\label{def:operations on labeled domains}
Let $f_1,f_2\in \ltds\cup\lfds$ with $f_1=\tup{d_1,\lab_1}$ and $f_2=\tup{d_2,\lab_2}$, and let $x\in d_1$ be an address.
\begin{itemize}
    \item The \emph{addition} $f_1+f_2$ is the $\lfds$ obtained by putting the two (tree or forest) domains ``side by side'', and updating the addresses accordingly
    (see $f_1+f_2$ in \cref{fig:forests example1}).
    \item We denote by $f_1[x]$ the \emph{sub-forest or sub-tree of $f_1$ at $x$}. Note that this ``shifts'' the addresses, 
    e.g., in \cref{fig:forests example1} $f_2=(f_1+f_2)[1]$.
    \item We denote by $f_1[x\mapsto f_2]$ the \emph{composition of $f_1$ with $f_2$ at $x$}, namely the forest obtained by replacing the subtree of $x$ at $f_1$ with the forest $f_2$. For example, in \cref{fig:contexts example1} we have $C_1(f)=C_1[00\mapsto f]$.
\end{itemize}
\end{definition}

Labeled domains also admit a natural partial order, as follows.
\begin{definition}
\label{def:dl order}
For two $L$-labeled address domain $dl_1=\tup{d_1,l_1},dl_2=\tup{d_2,l_2}$ we write $dl_1\leq dl_2$ if $d_1\subseteq d_2$ and $l_2|_{d_1} = l_1$. That is, for every $x\in d_1$ we have $l_2(x)=l_1(x)$.
\end{definition}

With these definitions in place, we can now naturally capture forests and contexts using forest domains. Formally, we define the injection 
$\tofd\colon \frs A\cup\cts A\rightarrow \llfds{(A\cup\{\square\})}$
that associates a tree or forest domain with a given forest or context. Specifically, for an element $h\in \frs A\cup\cts A$, we denote $\tofd(h)=\tup{\dom(h),\lab(h)}$.

\label{def:convinient notion}
This formalism allows us to define composition of forests as follows. For $h\in \frs A\cup \cts{A}$, a forest $h'\in \frs{A}$, and an address $x\in \dom(h)$,  the forest obtained by replacing the $x$-subforest of $h$ with $h'$ is
\[h[x\mapsto h']\coloneqq \tofd^{-1}(\tofd(h)[x\mapsto \tofd(h')])\]
And similarly define:
\[h[x]\coloneqq \tofd^{-1}(\tofd(h)[x])\]

The view of trees as labeled domains allows us to use graph-theoretic notions such as a \emph{node} (an element in the domain), a \emph{child} (a node $x\cdot i$ with respect to the node $x$), a \emph{leaf} (a node without children), a \emph{path} (a sequence starting with a node and continuing with successive children), etc.

We conclude this section with three important notations. 
\begin{enumerate}
    \item Define $\spos\colon\cts{A}\to \bbN^+$ such that for a context $C\in \cts{A}$, $\spos(C)$ is the address of the unique $\square$.
    Observe that by \cref{def:forests and contexts}, a context $C\in \cts{A}$ has exactly one address labeled $\square$, so this is well defined. 
    For example, in \cref{fig:contexts example1} we have $\spos(C_1)=00$.

    \item For $h_1,h_2\in \frs{A}\cup\cts{A}$, we write $\equpto{h_1}{x}{h_2}$ when $h_1$ and $h_2$ are identical apart from (possibly) their subtree rooted at $x$. For example, in \cref{fig:forests example1} we have $\equpto{f_1}{01}{f_3}$
    \item For $h\in \frs A \cup \cts A$ we denote by $\width(h)$ the maximal number in the domain of $\tofd(h)$ (i.e., the maximal ``root'' index in a disjoint union of trees).
    Notice that for a context $C$ and a forest $h$ we have that the width of the forest $C\cdot h$ is determined by: $\width(Ch) = \begin{cases}
        \width(C) & \spos(C)\notin \bbN\\
        \width(C) + \width(h) & \spos(C)\in \bbN
    \end{cases}$. Intuitively, the width is that of $C$, unless $\square$ is itself a root in $C$, in which case the width of $h$ is added.
\end{enumerate}


\subsection{Forest Algebras}
\label{sec:forest algebra}
The algebraic definition of regular languages (as opposed to automata/regular expressions) is that a language is regular if it is the inverse image of some monoid morphism from $\Sigma^*$ to a finite monoid $M$. This view offers a convenient way to reason about certain aspects of regular languages.

For tree languages, it is less clear what a proper algebraic formalism is. The work in~\cite{Bojanczyk2008ForestA} establishes \emph{forest algebras} as an algebraic approach for unranked tree languages. We bring the full definition here, as it is crucial to our results.

\begin{definition}[Forest Algebra]
\label{def:forest algebra}
A \emph{Forest algebra}~\cite{Bojanczyk2008ForestA} is a tuple\footnote{In the context of \emph{universal algebra} this is a structure of \emph{two sorts}.} $\tup{H,V,+,\cdot,\act,\inl,\inr,\zeroel,\unitel}$ 
such that:
\begin{itemize}
    \item $\tup{H,+,\zeroel}$ and $\tup{V,\cdot,\unitel}$ are monoids.\footnote{in \cite{Bojanczyk2008ForestA}  $H$ is referred to as \emph{horizontal monoid} and $V$ as \emph{vertical monoid}, but we do not adopt these names in the paper to reduce confusion with the new structure of \emph{decomposition tree} presented later.}

    \item $\act\colon V\times H\rightarrow H$ is a \emph{faithful monoid action} of $V$ on $H$:
    \begin{description}
    \item[Monoid action:] For every $v,u\in V$ and $h\in H$ we have $\act(v\cdot u,h)=\act(v,\act(u,h))$, and  $\act(\unitel,h)=h$.

    \item[Faithful:] For every $u\neq v\in V$ there exists $h\in H$ such that $\act(u,h)\neq \act(v,h)$.
    \end{description}
    We typically abbreviate $\act(v,h)$ to $vh$.

    \item $\inl,\inr\colon H\rightarrow V$ satisfy $\inl(g)\cdot h=g+h$ and $\inr(g)\cdot h=h+g$.
    (notice that in particular $\inl,\inr$ are monoid morphisms).
    Intuitively, the functions $\inl,\inr$ are semantic versions of the functions $g\mapsto g+\square$ and $g\mapsto \square+g$, respectively.
\end{itemize}
We usually denote a forest algebra by $\tup{H,V}$, omitting the remaining elements.
\end{definition}

The notion of morphism naturally extends to forest algebras: a \emph{forest morphism} between two forest algebras $\tup{H_1,V_1},\tup{H_2,V_2}$ is a tuple $\tup{\alpha,\beta}$ such that $\alpha\colon H_1\rightarrow H_2$ and $\beta\colon V_1\rightarrow V_2$ are monoid morphisms.

The intuitive meaning of $H$ and $V$ in a general forest algebra is difficult to grasp. This is similar to the case of general monoids in the regular-language setting -- the target monoid can be thought of as the state space of a minimal DFA for a language, and may be complicated. 
Sticking to this analogy, the free monoid $\Sigma^*$ is easier to understand, and so is the \emph{free forest algebra}, which we define next.

\begin{definition}[Free Forest Algebra]
\label{def:free forest algebra}
Let $A$ be an unranked alphabet. The \emph{free forest algebra} over $A$ is $\tup{\hf, \vf, +,\cdot,\act,\inl,\inr,\az,\square}$ with the following components.
\begin{itemize}
    \item $\hf=\frs{A}$ and its monoid is $\tup{\hf,+,\az}$ with $+$ being forest addition, and such that $\az$ (the empty forest) is a neutral addition element.
    
    \item $\vf=\cts{A}$ and its monoid is $\tup{\vf,\cdot,\square}$ where $\cdot$ is defined by
    \[C_1\cdot C_2=C_1[\spos(C_1)\mapsto C_2]\]

\item $\inl(h)= h+\square\in \vf$ and $\inr(h)\coloneqq \square+h\in \vf$ 
\item $\act(C_1, h)=C_1[\spos(C_1)\mapsto h]$.

For readability, we denote $\act(C,h)$ by $C\cdot h$ or $C(h)$

\end{itemize}
\end{definition}

\begin{remark}
\label{rmk:free forest algebra is a forest algebra}
    Notice that the free forest algebra is indeed a forest algebra.
    In particular, the multiplication of contexts is well defined and associative, $\square$ is a neutral element for $\cdot$, and the action $\act(C,h)=C(h)$ is faithful since if $C_1\neq C_2$ then e.g., the forest $\epsilon$ satisfies $C_1(\epsilon)\neq C_2(\epsilon)$. 
    In addition, $\act$ satisfies for every $h\in \hf$ that $\square(h)=h$ (which is sometimes referred to as \emph{free action}).
\end{remark}

The free forest algebra is free in the (universal-algebraic) sense that a function $f\colon A\to V$ lifts to a unique morphism.

\section{Augmented Free Forest Algebra}
\label{sec:augmented forest algebra}
In the setting of words there is a natural notion of decomposition: a word $w$ can be split into $w=u\cdot v$, where both $u$ and $v$ are subwords (\cref{fig:decomp-word}). Then, one can recursively continue to decompose. This decomposition supports the \emph{local} nature of decompositions -- in order to evaluate the morphism on an infix, one only needs to look at the corresponding node (indeed, this is the basis for the \emph{fast infix evaluation} algorithm~\cite{FactorForBoj}).

For forests, a fundamental problem arises: when decomposing a forest $h$, we want to remove (or ``pluck'') a subtree rooted at a certain node $x$. This leaves us with a \emph{context}, rather than a forest, so that we know where to restore the plucked subtree.
However, if we now wish to continue decomposing, we may want to pluck another subtree from the remaining context. This raises a problem: we need two separate ``holes'' to restore the subtrees (\cref{fig:decomp-forest}), and this would render the decomposition ill-defined (indeed, this is the reason contexts have only one hole).
\begin{figure}[ht]
    \centering
    \begin{subfigure}{0.15\linewidth}
        \centering
        \includegraphics[width=\linewidth]{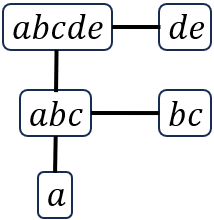}
        \caption{Words}
        \label{fig:decomp-word}
    \end{subfigure}
    \hfill
    \begin{subfigure}{0.2\linewidth}
        \centering
        \includegraphics[width=\linewidth]{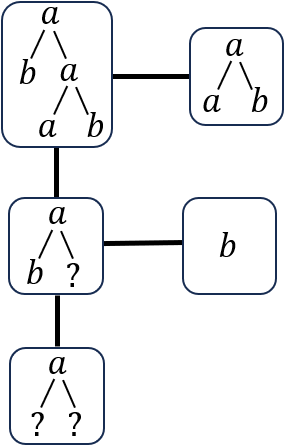}
        \caption{Forests without augmentation}
        \label{fig:decomp-forest}
    \end{subfigure}
    \hfill
    \begin{subfigure}{0.2\linewidth}
        \centering
        \includegraphics[width=\linewidth]{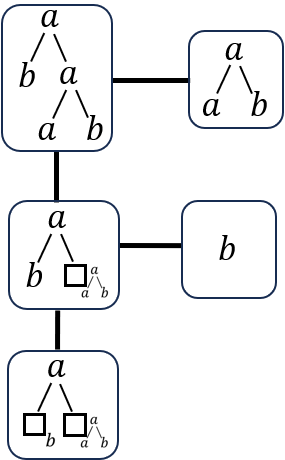}
        \caption{Forests with default holes}
        \label{fig:decomp-augmented}
    \end{subfigure}
    \caption{For words (\cref{fig:decomp-word}), each node can be evaluated separately according to the morphism. In particular, the value of a node is uniquely determined by the values of its children. \\
    For forests (\cref{fig:decomp-forest}, it is not even clear how to represent the residue forest after ``plucking'' a subforest. For example, the bottom leaf cannot be evaluated, and looking at it locally does not provide enough information to know how the two right leaves plug into it.\\
    To overcome this, we use default holes (\cref{fig:decomp-augmented}).}
    \label{fig:problem-decomp}
\end{figure}

In order to overcome this problem, we introduce the \emph{augmented free forest algebra}. Intuitively, we add a new symbol $\square_h$ for every $h\in \hf$, called a ``default hole'': these are holes that can only be filled with the subforest $h$ (\cref{fig:decomp-augmented}). 
A priori, it may seem as though this achieves nothing -- we just write $\square_h$ instead of $h$ (cf.\ \cref{xmp:forests with default holes}). Nonetheless, this is the basic building block for our notion of decomposition: the elements $\square_h$ are used as ``markers'' within a forest (or context) to denote a part which is \emph{elementary} and cannot be factored further.
We make this intuition concrete when defining decompositions in \cref{def:bin decomp}.

\begin{definition}[Forest Algebra with Default Holes]
\label{def:forest algebra with defaults}
Let $A$ be an alphabet and $\tup{H_A,V_A}$ be the free forest algebra over $A$. The set of \emph{default holes} is 
\[\square_{H_A}\coloneqq \{\square_h:h\in H_A\}\]
We augment the alphabet $A$ with default holes by defining $\wHole{A}=A\cup \square_{H_A}$ (note that $\wHole{A}$ is infinite).
We then lift the setting to forests and contexts over $\wHole{A}$ by allowing the symbols in $H_A$ only in the \emph{leaves}. Formally, we define $\tilH$ inductively by:
\begin{itemize}
    \item $\wHole{A}\subseteq \tilH$.
    \item For $h_1,h_2\in \tilH$ we have $h_1+h_2\in \tilH$.
    \item For $h\in \tilH$ and $\sigma\in A$ (not in $\wHole{A}$) we have $\sigma(h)\in \tilH$.
\end{itemize}
The set of \emph{contexts with default holes} $\tilV$ is defined similarly, also requiring a single $\square$ at some leaf.

Finally, we define the \emph{Forest Algebra with Default Holes} $\tup{\tilH,\tilV}$ identically to the free forest algebra in \cref{def:free forest algebra} (i.e., the same operations and constants $\act,\inl,\inr,+,\cdot$).

\end{definition}
\begin{example}
\label{xmp:forests with default holes}
    Fix $A=\{a,b\}$. 
    The following are forests with default holes, i.e., forests in $\tilH$:
    $a,\ \square_a,\ b(a+a)+b,\ b(\square_a),\ b(a+\square_{b(a+a)+b})+\square_{b(b)}$.
    
    The following are contexts in $\tilV$:
    $\square,\ a+\square,\ b(\square_{a(b)}+\square)+\square_{b(a+b(a))}$.
    Note that there may be several default holes, but exactly one standard hole $\square$.

    The following are \textbf{non-examples}, to illustrate the requirements:
    \begin{itemize}
        \item $\square_a(b)$  -- \emph{non leaf} default hole is not allowed.
        \item $b(a+\square_{b(a+\square)})$ -- default hole with \emph{context} value are not allowed.
        \item $b(a+\square_{b(a+\square_a)})$ -- \emph{nested} default holes are not allowed.
    \end{itemize}
\end{example}

\subsection{A Toolbox for Free Forest Algebras}
\label{sec:toolbox for free algebras}
In the remainder of this section we present several useful tools for reasoning about Free Forest Algebras (including our Augmented Algebra). 
Note that $\tilH$ and $\tilV$ can be viewed as address domains, identically to $\tup{H_A,V_A}$ above. In particular, we can use $\tofd$ and $\spos$ in the same manner.

Our first tool is a property of the action $\act\colon\tilV\times\tilH\rightarrow \tilH$ that provides a useful characterization for the equality of two contexts.
\begin{lemma}
\label{lem:equal contexts characterization}
Consider $C_1,C_2\in \tilV$, then the following conditions are equivalent:
\begin{enumerate}
    \item $C_1=C_2$.
    \item There exist $h_1\neq h_2\in \tilH\cup\tilV$ such that $C_1(h_1)=C_2(h_1)$ and $C_1(h_2)=C_2(h_2)$.
    \item $\spos(C_1)=\spos(C_2)$ and there exists $h\in \tilH$ such that $C_1(h)=C_2(h)$.
\end{enumerate}
\end{lemma}
\begin{proof}
    The implication $(1)\implies(2)$ is trivial (indeed, if $C_1=C_2$ then $C_1(h)=C_2(h)$ for every $h$).

    \subparagraph*{$(2)\implies (3)$:}
    Consider $C_1,C_2,h_1,h_2$ such that
        $h_1\neq h_2$ and $C_1(h_1)=C_2(h_1),C_1(h_2)=C_2(h_2)$.
        Denote $x_1=\spos(C_1),x_2=\spos(C_2)$, we claim that $x_1=x_2$. Assume by way of contradiction that $x_1\neq x_2$, we divide into cases according to the ordering between $x_1$ and $x_2$ (as elements in $\bns$).
        If $x_1<x_2$, then $C_1[x_1],C_2[x_1]\in \cts{A}$ but $C_1[x]=\square\neq C_2[x_1]$. In particular, $\dom(\square)=\{0\}\subsetneq \dom(C_2[x_1])$.
        In addition for every context $C$ and $x\leq \spos(C)$ we have that $C(h)[x] = C[x](h)$.  
        Together we have by \cref{def:context composition} that 
        \[C_1(h_1)[x_1]=C_1[x_1](h_1)=h_1\neq C_2[x_1](h_1)=C_2(h_1)[x_1]\]
        where the inequality is because the domain of $C_2[x_1](h_1)$ strictly contains the domain of $\square(h_1)=h_1$.
        This is a contradiction to the assumption $C_1(h_1)=C_2(h_1)$. The case $x_2<x_1$ is dual.

        Otherwise, $x_1$ and $x_2$ are incomparable. Then, intuitively, plugging $h_1$ at $x_1$ does not affect the subforest rooted at $x_2$, and vice-versa. Formally, for every $h\in \tilH$ we have 
        \[C_2[x_2\mapsto h][x_1]=C_2[x_1]\]
        Together with the assumptions of $(2)$ we get
        \begin{align*}
            h_1&=C_1[x_1\mapsto h_1][x_1] = C_2[x_2\mapsto h_1][x_1] = C_2[x_1]\\
            &=C_2[x_2\mapsto h_2][x_1] = C_1[x_1\mapsto h_2][x_1]
            = h_2
        \end{align*}
        which is a contradiction to the assumption $h_1\neq h_2$.
        We conclude that $x_1=x_2$ as required for $(3)$.

    \subparagraph*{$(3)\implies (1)$:}
        Denote $x=\spos(C_1)=\spos(C_2)$, and let $h\in \tilH$ such that $C_1(h)=C_2(h)$. We show that $C_1=C_2$.
        First, recall that for every $C,h$ we have $C(h)=C[\spos(C)\mapsto h]$.
        Denote $\tup{d,l}=C_1(h)=C_2(h)$.
        Notice that $\dom(C_1) = d\setminus x\bnp = \dom(C_2)$. Take $y\in d\setminus x\bnp$. If $y\in d\setminus x\bns$ then $\lab(C_1)(y)=l(y)=\lab(C_2)(y)$. Otherwise, $y=x$ and $\lab(C_1)(x)=\square=\lab(C_2)(x)$.
        So overall we get $C_1=C_2$ which implies the required $C_1=C_2$, concluding $(1)$.
\end{proof}
\begin{remark}[Two forests are needed to separate contexts]
\label{rmk: two forest to separate contexts}
The condition in Item 2 of \cref{lem:equal contexts characterization} may seem redundant, and raises the question of whether equality is already implied when a single $h$ satisfies $C_1(h)=C_2(h)$.
This, however, is not the case. For example, consider:
    \[C_1 = a + \square,\quad C_2 = \square + a,\quad h=a\]
    Then $C_1(h)=C_2(h)$, but $C_1\neq C_2$ as e.g., $C_1(b)=a+b\neq b+a= C_2(b)$. 
    Thus, the condition that equality is maintained under two different $h_1\neq h_2$ is necessary.
\end{remark}

Another basic technical tool we require is an ability to manipulate addresses in a way that corresponds to the behavior of $\spos$ under context composition. In the following we define such an operation $\spa$ and establish several properties of it.
Intuitively, one should think of $x\spa y$ as ``if we take a context $C$ with $\spos(C)=x$, and a node at address $y$ in some forest $h$, then this node moves to  address $x\spa y$ in $C\cdot h$''.
\begin{definition}
\label{def:spa}
We define $\spa\colon\bnp\times \bnp\to \bnp$ as follows: for $x,y\in \bnp$, write $x=x'\cdot \sigma_x,y=\sigma_y\cdot y'$ where $\sigma_x,\sigma_y\in \bbN$. Then  $x\spa y= x'\cdot (\sigma_x+\sigma_y)\cdot y'$. 
\end{definition}
$\spa$ is essentially an extension of $\cdot$, in the following sense.
\begin{proposition}
    \label{prop:spa and concat}
    Let $x\in \bns$ and $y\in \bnp$ then $x\cdot y = (x\cdot 0)\spa y$
\end{proposition}
\begin{proof}
    Denote $y = \sigma_y\cdot y'$, then  and we get $(x\cdot 0)\spa y = x\cdot (0 + \sigma_y)\cdot y' = x\cdot \sigma_y\cdot y' = x\cdot y$
\end{proof}
We now show that $\spa$ behaves well, algebraically. This also captures the intuition that $x\spa y$ is the address of $y$ when plugged into a context at $x$ (by essentially identifying $y$ with $\spos$).
\begin{lemma}
    \label{prop: address spos morphism}
    $\spa$ satisfies the following properties.
    \begin{enumerate}
    \item $\tup{\bnp, \spa, 0}$ is a monoid.
    \item $\spos\colon\tup{\tilV, \cdot, \square} \rightarrow \tup{\bnp, \spa, 0}$ is a monoid morphism.
    \end{enumerate}
\end{lemma}
\begin{proof}
    First we show that $\tup{\bnp, \spa, 0}$ is a monoid.
    \begin{itemize}
            \item $\spa$ is associative:
        we compute $(\overbrace{x\cdot \sigma_x}^{x'})\spa(y\spa (\overbrace{\sigma_z\cdot z}^{z'}))$ according to whether $|y|>1$:
        \begin{align*}
            &x'\spa((y'\cdot \sigma_y)\spa z') = x'\spa(y'\cdot (\sigma_y+\sigma_z)\cdot z))\\
            &=\begin{cases}
                x\cdot (\sigma_x+\sigma_y+\sigma_z)\cdot z = ((x\cdot\sigma_x)\spa\overbrace{\sigma_y}^y)\spa z' &y = \sigma_y\\
                x\cdot (\sigma_x+\sigma'_y)\cdot y''\cdot (\sigma_y+\sigma_z)\cdot z = ((x\cdot \sigma_z)\spa (\overbrace{\sigma'_y\cdot y''\cdot \sigma_y}^{y}))\spa z' &y = \sigma'_y\cdot y''\cdot \sigma_y
            \end{cases}\\
            &=(x'\spa y)\spa z'
        \end{align*}
            \item $0$ is unit of $\bnp$:
            Let some $x\in \bnp$ then denote $x=\sigma\cdot y$ and we get \[0\spa x = (0+\sigma)\cdot y = \sigma\cdot y = x\]
        \end{itemize}
    
    In order to show that $\spos$ is a monoid morphism, we need to show that $\spos(C'\cdot C)=\spos(C')\spa \spos(C)$ for every $C,C'\in \tilV$. The formal proof is by induction below.
    Intuitively, however, denote $\spos(C')=x=x'\sigma_x$ and $\spos(C)=y=\sigma_y y'$ and consider $C'\cdot C$. In $C'$, the leaf $\square$ at position $x'\sigma_X$ is replaced by all the roots of the forest $C$. Therefore, the branch leading to $\square$ in $C'\cdot C$, which in $C$ is rooted at $\sigma_y$, now passes through $x'(\sigma_x+\sigma_y)$. Then, this branch proceeds along $y'$ to reach $\square$. Therefore $\spos(C'\cdot C)=x'\cdot (\sigma_x+\sigma_y)\cdot y'$, as required. 
    
    Before addressing this formally, we make the following basic observations:
    \begin{itemize}
    \item $\spos(\square) = 0$  
    \item $\spos(\sigma(C)) = \spos(\sigma(\square))\spa \spos(C)$
    \item $\spos(h+C) = \spos(h+\square)\spa \spos(C)$
    \item $\spos(C+h) = \spos(\square+h)\spa \spos(C)$
    \end{itemize}
   We can now show that $\spos(C'\cdot C)=\spos(C')\spa \spos(C)$ for every $C,C'\in \tilV$. We show this by induction over the structure of $C'$.
        \begin{itemize}
            \item If $C'=\square$ then
            \begin{align*}
              \spos(\square \cdot C) &= \spos(C) = 0\spa \spos(C)
              = \spos(\square)\spa \spos(C)
            \end{align*}
            \item If $C'=\sigma(C'')$ then 
            \begin{align*}
                \spos(\sigma(C'')\cdot C) &= \spos(\sigma(\square)\cdot C''\cdot C)=\spos(\sigma(\square))\spa \spos(C''\cdot C) \\&=\spos(\sigma(\square))\spa \spos(C'')\spa \spos(C)\\
                &=\spos(\sigma(C''))\spa \spos(C)
            \end{align*}
            \item If $C'=h+C''$ then  
            \begin{align*}
                \spos((h+C'')\cdot C) &= \spos((h+\square)\cdot C''\cdot C)\\
                &=\spos(h+\square)\spa (\spos(C'')\spa \spos(C))\\
                &=\spos(h+\square)\spa \spos(C'')\spa \spos(C)\\
                &=\spos(h+C'')\spa \spos(C)
            \end{align*}
            \item If $C'=C''+h$ then similarly to the above
            \begin{align*}
                \spos((C''+h)\cdot C) &= \spos(\square+h)\spa \spos(C'')\spa \spos(C)
                \\&= \spos(C''+h)\spa \spos(C)
            \end{align*}
        \end{itemize}
\end{proof}
Our next intuitive property is that if we fix $x$, then address shifts with respect to $x$ are injective (i.e., address shifts do not cause any ``collision'' of addresses).
\begin{proposition}
\label{prop:spa injective}
    For every $x\in \bnp$ the functions $y\mapsto x\spa y$ and $y\mapsto y\spa x$ are injective.
\end{proposition}
\begin{proof}
    The proof is based on the following observation: for every $x',y',z'\in \bns$ and $n_1,n_2\in \bbN$, if $x'n_1y'=x'n_2z'$ then $n_1=n_2$ and $y'=z'$.
    
    Consider the function $y\mapsto x\spa y$.
    Denote $x=x'\cdot \sigma_x$, and let $y = \sigma_y\cdot y', z=\sigma_z\cdot z'$ such that $x\spa y = x\spa z$.
    This means that $x'(\sigma_x+\sigma_y)y' = x' (\sigma_x+\sigma_z)z'$, yielding by the observation above that $\sigma_x+\sigma_y=\sigma_x+\sigma_z$ (so $\sigma_y=\sigma_z$) and $y'=z'$. Thus, $y=z$, so the function is injective.

    The case of $y\mapsto y\spa x$ is symmetric.
\end{proof}

Since $\spa$ extends $\cdot$ as per \cref{prop:spa and concat}, it also corresponds to a partial order $\lcrs$ on $\bbN^*$ that generalizes the prefix relation. Intuitively, in $\lcrs$ a node $x$ is smaller than its left child $x0$, and $x\cdot i$ is smaller than $x\cdot (i+1)$. 
\begin{definition}[Left-Child Right-Sibling Order]
\label{def:lcrs order}
The \emph{Left Child Right Sibling order} (LCRS, for short) is a partial order $\lcrs$ on $\bns$ defined as the transitive closure of the relation $\{(x,x):x\in \bns\}\cup\{(x,x0):x\in \bns\}\cup \{(xi,x(i+1)):x\in \bns,i\in \bbN\}$.
\end{definition}
For example: $1\lcrs 1\cdot2$ and $1\cdot2\lcrs 1\cdot3$, but $1\cdot3\not\lcrs 2$. 
Observe that an address domain is simply a \emph{finite} subset $d\subseteq \bns$ which is closed under intervals by $\lcrs$.

The connection between $\spa$ and $\lcrs$ is captured in the following.
\begin{proposition}
\label{prop:lcrs spa}
    Consider $x,y\in \bnp$ then $x\lcrs y$ if and only if exists $u\in \bnp$ such that $y = x\spa u$. Moreover, this $u$ is unique if it exists.
\end{proposition}
\begin{proof}
    The first direction follows by induction on the structure of $\lcrs$ (\cref{def:lcrs order}). Let $x,y$ such that $x\lcrs y$.
    \begin{itemize}
        \item Base:
        \[u=\begin{cases}
        0 & y=x\\
        00 & y=x0\\
        1 & x=zi\wedge y=z(i+1)\end{cases}\]
        \item Step: if $x\lcrs z \lcrs y$ ($x\neq z\neq y$) then by the induction hypothesis we can write $z=x\spa u_z$ and $y = z\spa u_y$, so by the definition of $\spa$ we have $y = x \spa (u_z\spa u_y)$
    \end{itemize}
    The uniqueness of $u$ follows from the injectivity in \cref{prop:spa injective}.
    
    For the second direction, assume $y = x\spa u$. Denote $x=x'\cdot \sigma_x$ and $u = \sigma_u\cdot u'$, then we get $y = x'(\sigma_x+\sigma_u)u'$. We now notice that $x\lcrs x'(\sigma_x+\sigma_u)\lcrs y$.
\end{proof}

The tools introduced thus far apply to the augmented forest algebra as well as to the standard free forest algebra. 
A natural operation unique to the augmented algebra is to replace a default hole $\square_h$ with the concrete forest $h$. In general, we can take a subforest and replace all its default holes with their concrete values. We call this \emph{unraveling}, as follows.
\begin{definition}[Unraveling]
\label{def:unraveling}
Consider $h\in \tilH\cup \tilV$. The \emph{unraveling} of $h$ is 
\[
\unr(h)=
\begin{cases}
    h'&h=\square_{h'}\\
    \square & h=\square\\
    \sigma(\unr(h'))& h=\sigma(h')\\
    \unr(h_1)+\unr(h_2)&h=h_1+h_2
\end{cases}
\]    
Note that by definition $\unr$ is a monoid morphism when viewed as either $\unr\colon\tilH\to H_A$ or $\unr\colon\tilV\to V_A$.
\end{definition}
We show that unraveling also respects the action of $\tilV$ on $\tilV\cup \tilH$.
\begin{lemma}[Unraveling Morphism]
    \label{lem:unr mor}
    For every $C \in \tilV$ and $h\in \tilV\cup \tilH$ we have
    \[\unr(C\cdot h) = \unr(C)\cdot \unr(h)\]
\end{lemma}
\begin{proof}
    By induction on the structure of $C$, using \cref{def:unraveling}.
    \begin{itemize}
        \item $C=\square$: trivial
        \item $C=\sigma(C')$: 
        \begin{align*}
            \unr(C\cdot h) &= \sigma(\unr(C'\cdot h)) = \sigma(\unr(C')\cdot \unr(h))\\
            &=\sigma(\unr(C'))\cdot \unr(h) =\unr(C)\cdot\unr(h)
        \end{align*}
        \item $C = h'+C'$:
        \begin{align*}
            \unr(C\cdot h)&=\unr(h')+\unr(C'\cdot h) \\
            &= \unr(h')+\unr(C)'\cdot \unr(h) \\
            &= \unr(h'+C')\cdot\unr(h)\\
            &=\unr(C)\cdot \unr(h)
        \end{align*}
        \item $C=C'+h'$: analogous to the previous case.
    \end{itemize}
\end{proof}

It is often useful to equate a default hole $\square_h$ with $h$, this can be lifted to an equivalence over forests and contexts, as follows (see \cref{fig:equiv unravel}).
\begin{definition}[Unravel-Equivalence]
    \label{def:equiv unravel}
    We say that $h_1,h_2\in \tilH\cup \tilV$ are \emph{unravel equivalent}, denoted $h_1\eqUnr h_2$ if $\unr(h_1)=\unr(h_2)$.
\end{definition}
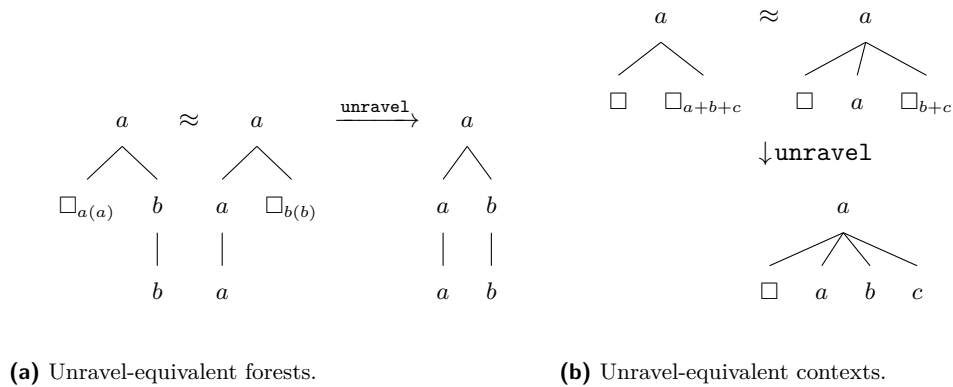
\begin{figure}[ht]
    \centering
    \begin{subfigure}{0.48\textwidth}
        \centering
        \[
        \begin{forest}
            [$a$ [$\square_{a(a)}$] [$b$ [$b$]]]
        \end{forest}
        \eqUnr 
        \begin{forest}
            [$a$ [$a$ [$a$]] [$\square_{b(b)}$]]
        \end{forest}
        \xrightarrow{\unr}
        \begin{forest}
            [$a$ [$a$ [$a$]] [$b$ [$b$]]]
        \end{forest}
        \]
        \caption{Unravel-equivalent forests.}
        \label{fig:equiv unravel 1}
    \end{subfigure}
    \hfill
    \begin{subfigure}{0.48\textwidth}
    \centering
    \[
    \begin{split}
    \begin{forest}
        [$a$ [$\square$][$\square_{a+b+c}$]]
    \end{forest}
    &\eqUnr 
    \begin{forest}
        [$a$ [$\square$] [$a$] [$\square_{b+c}$]]
    \end{forest}\\
    &\downarrow\!{\unr}\\
    &\begin{forest}
        [$a$ [$\square$] [$a$] [$b$] [$c$]]
    \end{forest}
    \end{split}
    \]
    \caption{Unravel-equivalent contexts.}
    \label{fig:equiv unravel 2}
\end{subfigure}
    \caption{Unravel equivalent structures, and their corresponding unraveling.}
    \label{fig:equiv unravel}
\end{figure}

Of particular importance in our decomposition are sets of contexts that are ``stable'' under unraveling (and also under replacing subforests with default holes). We capture this in the following.
\begin{definition}[Stable Context Set]
    \label{def:stable context set}
    A set of contexts $\wHole{V}\subseteq \tilV$ is \emph{stable} if it is a sub-semigroup of $\tilV$, and for every $C_1,C_2\in \tilV$, if $C_1\eqUnr C_2$ then $C_1\in \wHole{V}\iff C_2\in \wHole{V}$.
\end{definition}

Recall that syntactically, we do not allow ``nested default holes'', i.e., $\square_h$ when $h\in \tilH\setminus H_A$. However, in the following it is convenient to refer to such elements, without changing the syntax. To achieve this, we lift the notion of $\square_h$ to $h\in \tilH\setminus H_A$ by defining $\square_h=\square_{\unr(h)}\in \hf$.

\begin{remark}[$\spos$ and $\unr$]
\label{rmk:spos and unravel}
Consider some $C\in \tilV$. The reader may suspect that $\spos(C)$ is simply $\spos(\unr(C))$. This, however, is not always the case. For example, 
take $C=\square_{a+a}+\square$ and notice that $\spos(C)=1$ while $\spos(\unr(C)) = \spos(a+a+\square)=2$. 
This address change is a nuisance that we encounter in certain proofs, and must account for.
\end{remark}

\label{def:lifted morph}
We now lift the notion of morphism to handle $\tup{\tilH,\tilV}$ as follows.
For a forest morphism $\tup{\alpha,\beta}$ where $\alpha\colon H_A\to H$ and $\beta\colon V_A\to V$ (for some forest algebra $\tup{H,V}$) we define $\tup{\wHole{\alpha},\wHole{\beta}}=\tup{\alpha\circ \unr,\beta\circ \unr}$. Since $\unr$ is a monoid morphism, this is a monoid morphism from $\tup{\tilH,\tilV}$ to $\tup{H,V}$.

The following now readily follows from the definition of morphism (since $\wHole{\beta}$ already unravels the given input).
\begin{proposition}
    \label{prop:inverse image is unravelling stable}
    Consider a monoid morphism $\wHole{\beta}\colon\tilV\to V$ and a sub-semigroup $S\le V$, then $\wHole{\beta}^{-1}(S)$ is a stable context set.
\end{proposition}

\section{Binary Decomposition}
\label{sec:bin decomp}
We can now start describing the first phase of our decomposition scheme, whereby we take a forest and decompose it via a binary tree. 
At a glance, this is a tree analogue of taking a word $w$ and simply decomposing it recursively into subwords in some way. Crucially, however, in trees there are many more ways to decompose, and more importantly -- in order to combine two children into a parent node we need to keep more information than just the two parts: we need to also know \emph{how} they fit together (as discussed in \cref{fig:problem-decomp}).

We remark that this phase of the decomposition is purely syntactic: we take a forest $h\in \tilH$ and produce a binary tree that corresponds to its decomposition, regardless of any morphisms involved in the setting. In \cref{def:general decomposition} we generalize the construction to achieve bounded depth, and there we use the morphism. 

This section is organized as follows. In \cref{sec:bin decomp definitions} we define a binary decomposition. In \cref{subsec:references and ancestor embedding} we enrich the decomposition with additional structural information that is crucial for our reasoning. Then, in \cref{subsec:stable and basis} we parameterize the definition by sets of contexts and of leaves, which is used later to apply inductive arguments.

\subsection{Binary Decomposition Definitions}
\label{sec:bin decomp definitions}
Before presenting the precise definition, we illustrate a binary decomposition.
\begin{example}
\label{xmp:binary decomposition}
Consider the forest $h=b(a+a)+a(b+a(a+b))$. We wish to decompose it as \colorbox{blue!30}{$b(a+a)+a($}\colorbox{red!30}{$b+a(a+b)$}\colorbox{blue!30}{$)$}.
To this end, we write $h=C\cdot h'$ where $C=b(a+a)+a(\square)$ and $h'=b+a(a+b)$. Then, the corresponding binary decomposition is depicted in \cref{fig:bin decomp example}.
To distinguish the depictions of decomposition trees from forests in the algebra, we depict decompositions with horizontal/vertical edges, with \emph{right-child} for the \emph{factor}, and \emph{down-child} for the \emph{residue}. 

The decomposition tree is constructed as follows: we start with $h$ at the root of the decomposition, which is given address {\color{green!50!black} $\epsilon$}. 
We then ``pluck'' from $h$ the sub-forest factor $h'$ and place it as the right-child, with address {\color{green!50!black} $1$}. The node $1$ is also assigned a \emph{context}, which is the residue $C$ (so that $C\cdot h'=h$). Then, in the down-child (addressed {\color{green!50!black} $0$}) we place the tree that remains to be factored, which is $C$ only with $\square_{h'}$ instead of $\square$, to signify that this remaining hole is actually $h'$, but $h'$ itself is already being factored elsewhere (namely from node $1$).
Thus, the forest at node $0$ is $C\cdot \square_{h'}$. We remark that for uniformity (i.e., so that each node is labeled with a forest and a context), we also assign a context to node $0$, but this is just $\square$ and is immaterial.

\begin{figure}[ht]
    \centering
    \includegraphics[width=0.3\linewidth]{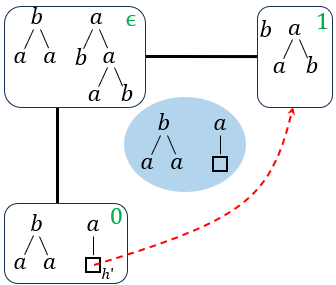}
    \caption{A (partial) binary decomposition corresponding to \cref{xmp:binary decomposition}. The plucked factor $h'$ appears in Node $1$. The context $C$ appears highlighted below the $\epsilon\to 1$ edge. The contexts associated with nodes $\epsilon$ and $2$ are omitted (they are both $\square$, see \cref{def:bin decomp}).}
    \label{fig:bin decomp example}
\end{figure}

\end{example}

We are now ready for the precise definition of binary decomposition.
\begin{definition}[Binary Decomposition Tree]
    \label{def:bin decomp}
    A \emph{binary decomposition tree} is an $\tilH\times \tilV$-labeled tree domain $\tau=\tup{\tau.\nds,\tau.l}\in \lltds {(\tilH\times \tilV)}$ equipped with three functions 
    \begin{itemize}
        \item $\tau.\frst\colon\tau.\nds\rightarrow \tilH$
        \item $\tau.\ctx\colon\tau.\nds\rightarrow \tilV$
        \item $\tau.\type\colon\tau.\nds\rightarrow \{\B,\lf\}$
    \end{itemize}
    that satisfy the following conditions.
    \begin{enumerate}
        \item 
        \label{itm:decomp:full binary}
        $\tau.\nds$ is a \emph{full, binary} tree domain. That is, every  $x\in \tau.\nds$ is either a leaf, or it has \emph{exactly} two children: $x0$ and $x1$. In the latter case, we refer to these children as the \emph{down} and the \emph{right} child, respectively.
        \item 
        \label{itm:decomp:label pair}
        For every $x\in\tau.\nds$ we have $\tau.l(x)=(\tau.\frst(x),\tau.\ctx(x))$.
        \item 
        \label{itm:decomp:node types}
        $\tau.\type$ identifies each node as a \emph{leaf} ($\lf$) or a \emph{binary node} ($\B$). Thus, the set of leaves is $\tau.\type^{-1}[\lf]$ and the set of binary nodes is $\tau.\type^{-1}[\B]$.
        \item 
        \label{itm:decomp:pluck nonempty}
        For every $x\in \tau.\nds$ and $i\in \{0,1\}$ we have $\tau.\frst(xi)\neq \tau.\frst(x)$. Intuitively, this means that we pluck a nontrivial forest in each node.
        \item 
        \label{itm:decomp:down child is square context}
        For every $x\in \tau.\nds\setminus \bns1$ (i.e., $x$ a down-child or the root) we have $\tau.ctx(x)=\square$.
        \item 
        \label{itm:decomp:correct context composition}
        For every binary node $x\in\tau.\type^{-1}[\B]$ we have $\tau.\frst(x)=\tau.\ctx(x1)(\tau.\frst(x1))$ and $\tau.\frst(x0)=\tau.\ctx(x1)(\square_{\tau.\frst(x1)})$.
    \end{enumerate}
\end{definition}
For brevity, in the remainder of the section we just write ``decomposition'' when referring to a binary decomposition. Also, since decompositions are in particular labeled tree domains, all the relevant definitions of \cref{sec:labeling intuitive} apply. 
We say that a decomposition $\tau$ is \emph{of forest $h$} if $h$ is the forest of the root of $\tau$, i.e., $\tau.\frst(\epsilon)=h$. 
For a node $x\in \tau.\nds$ we sometimes refer to $x1$ as the \emph{factor} and to $x0$ as the \emph{residue}, corresponding to the intuition that we split $\tau.\frst(x)$ at some subforest, and factor out this subforest into $x1$, with the rest (the residue) going to $x0$. We often refer to a factor as being \emph{plucked} from its parent node.

\begin{example}
\label{exmp:decomposition}
    Recall the decomposition in \cref{xmp:binary decomposition,fig:bin decomp example}. Formally, it has the following components:
    \begin{itemize}
        \item $\tau.\nds = \{\epsilon,0,1\}$
        \item $\tau.\frst(\epsilon) = h = b(a+a)+a(b+a(a+b))$
        \item $\tau.\frst(1)=b+a(a+b)$
        \item $\tau.\frst(0)= h''=b(a+a)+a(\square_{b+a(a+b)})$ 
        \item $\tau.\ctx(\epsilon)=\tau.\ctx(0)=\square$
        \item $\tau.\ctx(1)=C=b(a+a)+a(\square)$ (notice that $h=C\cdot h'$)
    \end{itemize}
    Notice that $h=C\cdot h'$ and $h''=C\cdot \square_{h'}$.
\end{example}

We remark that despite being part of the definition, $\tau.\ctx$ can actually be fully inferred from $\tau.\frst$. Indeed, the requirement that $\tau.\frst(xi)\neq \tau.\frst(x)$ (Condition~\ref{itm:decomp:pluck nonempty}) intuitively means that each decomposition step actually decomposes a nontrivial part of the forest, i.e., splits not just $\square$ and not the whole forest. Specifically, $\tau.\frst(x0)\neq \tau.\frst(x)$, which means that the part of the tree decomposed into $x1$ is not $\square_h$ for some $h$. Therefore, 
$\tau.\frst(x)\neq \square_{\tau.\frst(x)}$ for every $x$. Then, by \cref{lem:equal contexts characterization} we have that $\tau.\ctx(x1)$ is fully determined by the two equations of Condition~\ref{itm:decomp:correct context composition}. 

In addition, since each node plucks a nontrivial sub-forest (Condition~\ref{itm:decomp:pluck nonempty}), we show that every decomposition has bounded depth with respect to the forest at its root.
\begin{lemma}[Maximal Decomposition Depth]
\label{lem:max dec}
Consider a forest $t\in \tilH$ and a binary decomposition $\tau$ with $\tau.\frst(\epsilon)=t$, then $\depth(\tau)\leq |\tofd(t)|-|\lab(t)^{-1}[\square_\hf]|$ (i.e., the depth of $\tau$ is at most the number of nodes in $t$ that are not labeled by some $\square_h$).
\end{lemma}
\begin{proof}
    Denote $\Phi(t)=|\tofd(t)|-|\lab(t)^{-1}[\square_\hf]|$ and observe that $\Phi(t)\geq 0$ since $\lab(t)^{-1}[\square_\hf]\subseteq \tofd(t)$.
    We prove the lemma by induction on $\Phi(t)$:
    \begin{itemize}
        \item Base case: $\Phi(t)=0$.
        In this case we have $t\in \square_\hf$, so the only possible decomposition of $\tau$ is a single leaf at the root, whose forest is $t$. This has depth $0$, as required.
        \item Step: assume that for every $t'$ such that $\Phi(t')<\Phi(t)$ the lemma holds. Consider a decomposition $\tau$ such that $\tau.\frst(\epsilon)=t$, and assume $0,1\in \tau.\nds$ (otherwise $\depth(\tau)=0$ and we are done).
        Denote $t_0=\tau.\frst(0)$ and $t_1 = \tau.\frst(1)$. 
        Notice that $\Phi(t_1)<\Phi(t)$. Indeed, $|\tofd(t_1)|<|\tofd(t)|$, and any node in $t$ that is not in $t_1$ either contributes $1$ to $\Phi(t)$ (if it is not in $\square_\hf$) or $0$ if it is in $\square_\hf$.

        We claim that $\Phi(t_0)<\Phi(t)$ as well. 
        Observe that $|\tofd(t_0)|\le |\tofd(t)|$, but $t_0$ has an new $\square_{h}$ node, with $h=\tau.\frst(t_1)$ that is not in $t$. If $|\tofd(t_0)| = |\tofd(t)|$, then since we introduce a new $\square_h$ node in $t_0$, we have $\Phi(t_0)=\Phi(t)-1$ (notice that we cannot have $h\in \square_\hf$ as it would violate Condition~\ref{itm:decomp:pluck nonempty}, moreover $h\in A$). 
        Otherwise, we have $|\tofd(t_0)|< |\tofd(t)|$, and the same reasoning as above applies: every node in $t$ and not in $t_0$ contributes either $1$ or $0$ to $\Phi(t)$, depending on whether it is or is not in $\square_\hf$. 

        By the induction hypothesis, we can now conclude that 
        \[\depth(\tau)\leq 1+\max\{\depth(\tau[0]),\depth(\tau[1])\}\leq 1+\max\{\Phi(t_0),\Phi(t_1)\}\leq \Phi(t)\]
    \end{itemize}
\end{proof}

\subsection{References and Ancestor Embeddings}
\label{subsec:references and ancestor embedding}
Our next step is to associate with a decomposition $\tau$ a \emph{reference} relation. Intuitively, this relation describes how different nodes in the decomposition are ``linked'' to each other as factors of the same forest. More precisely, consider some node $z\in \tau.\nds$. When we factorize $\tau.\frst(z)$ to $\tau.\frst(z)=C\cdot h$, then $C=\tau.\ctx(z1)$, and $\tau.\frst(z0)=C(\square_{h})$. However, there could other nodes in the decomposition tree whose forest is $\square_h$, stemming from other parts of the tree where another $h$ subtree was plucked.
For certain algorithmic purposes, it is important that we not only remember $\square_h$, but also exactly which node in $\tau$ generates this \emph{specific} $h$. We thus ``link'' $z0$ with $z1$ in this case.

However, things get complicated as the decomposition continues: this $\square_h$ propagates through the tree in further decompositions. For example, when decomposing $\tau.\frst(z0)=C'\cdot h'$, we either have that $\square_h$ is part of $C'$, or part of $h'$. In the former case, we need to track the reference from $z0$, and in the latter from $z1$. We also need to carefully maintain the address of this link.

The formal definition is somewhat technical, since it requires tracking addresses. We therefore present an intuitive version, followed by the precise version and an example.

\begin{definition}[Decomposition References -- Intuitive Version]
    \label{def:decomp references intuitive}
    Consider a decomposition $\tau$. The \emph{reference relation} is a relation $\tau.\rdef\subseteq \tau.\nds\times \bnp\times \tau.\nds$, where $\tau.\rdef(x,u,y)$ is read as ``node $x$ at position $u$ points to node $y$'', and we denote it is $\tau.x@u\to y$. 
    When $\tau$ is clear from context we use the abbreviation $x@u\to y$. We write $x@*\to y$ when the exact position $u$ is not important.

    Then, $x@u\to y$ if $\tau.\frst(x)$ at address $u$ is $\square_h$ and $\tau.\frst(y)=h$ such that at some ancestor of $x$ the decomposition plucks this particular subtree $h$.

    Note that the only way to have a reference of the form $x1@u\to y$ is if the plucked subforest at $x$ contains the reference to a subforest that is already plucked above $x$. These represent a ``bad'' way of plucking (as it is ``better'' to first pluck the larger sub-forest). We refer to such references as \emph{inherited}.
\end{definition}
The precise version of the references relation is as follows.
\begin{definition}[Decomposition References]
\label{def:decomp references}
    Consider a decomposition $\tau$. The \emph{reference relation} is a relation $\tau.\rdef\subseteq \tau.\nds\times \bnp\times \tau.\nds$, where $\tau.\rdef(x,u,y)$ is read as ``node $x$ at position $u$ points to node $y$''. We typically use the abbreviation $x@u\to y$, and write $x@*\to y$ when the exact position $u$ is not important.
    
    We define $\tau.\rdef$ inductively on $x\in \tau.\nds$ as follows.
    Denote $s=\spos(\tau.\ctx(x1))$. 
    \begin{description}
    \item[Base case:] we always add the reference
    $x0@s\mapsto x1$ 
    \item[Step:] assume we already have $x@u\to y$, we define references of $x0$ and $x1$ that inductively stem from this reference.
    Let $M=\width(\tau.\frst(x1))$.
    \begin{itemize}
        \item If $s\not\lcrs u$: add reference $x0@u\to y$. This corresponds to the case where the reference from $x$ is within the residue, and therefore carried to $x0$ as is.
        \item If $s\spa (M+1)\lcrs u$: denote $u=s\spa (M+1)\spa v$ and add reference $x0@(s\spa 1\spa v)\to y$. This corresponds to the case where the reference from $x$ is within the residue but its address require adjustment due to plucking out the factor.
        \item Else (meaning $s\lcrs u$ and $s\spa(M+1)\not\lcrs u$): Denote $u=s\spa v$, and add reference $x1@v\to y$.  This corresponds to the case where the reference from $x$ appears in the sub-forest which currently is plucked out. 
        We distinguish this latter type of references, and refer to them as \emph{inherited references}.
        \label{def:inherited refs}
    \end{itemize}
    \end{description}
\end{definition}

\begin{example}[References]
\label{xmp:references}
Consider the decomposition depicted in \cref{fig:bin decomp example 2}, extending \cref{xmp:binary decomposition}. The dashed arrows now represent references. Specifically, observe that the address of $\square_{h'}$ in the forest of Node $1$, namely $b(a+a)+a(\square_{h'})$ is $10$, and therefore we have the reference $0@10\to 1$.
Next, observe that the subforest plucked at Node $0$ contains $\square_{h'}$, and therefore we also have the reference $01@*\to 1$ (specifically, $01@00\to 1$), and this reference is \emph{inherited}. 
We also have a reference $00@*\to 01$.

\begin{figure}[ht]
    \centering
    \includegraphics[width=0.4\linewidth]{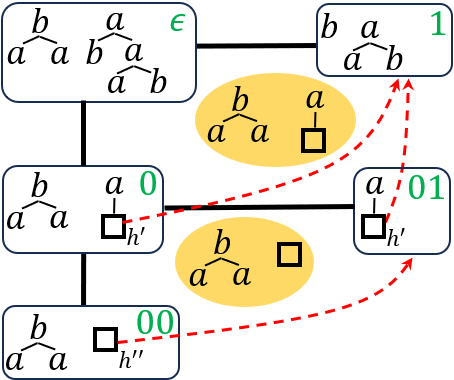}
    \caption{A binary decomposition with depicted references (dashed arrows).}
    \label{fig:bin decomp example 2}
\end{figure}
\end{example}

Note that only \emph{inherited} references are from right-children (i.e., $x1$). Intuitively, an inherited references means that we pluck out a factor, but later decide to pluck a predecessor of that factor. This raises a red flag, as we could have started by plucking out the larger sub-forest. Inherited references help us keep track of this, and we often pay special attention to cases where they do not occur.

A consequence of \cref{def:decomp references} is that for every $x,y\in \tau.\nds$ there is at most one address $u$ such that $x@u\to y$. Intuitively, $\tau.\frst(y)$ is only plucked from $x$ once using the node $y$. Identical forests may be plucked from $x$, but not to node $y$. 
We can therefore represent the link between $x$ and $y$ by replacing the referenced $\square_{\tau.\frst(y)}$ in $x$ by $\square$. 
\begin{definition}[Linking Context]
    \label{def:linking context}
    For every $x,y$ such that $x@u\to y$ we define the \emph{linking context between $x$ and $y$} as $\clnk(\tau,x,y)\coloneqq \tau.\frst(x)[u\mapsto \square]$.
\end{definition}
For example, in \cref{fig:bin decomp example 2} the reference $01@*\to 1$ induces the linking context $\clnk(\tau,01,1)=a(\square)$. Note that this may seem independent of the node $1$, but in forests with several default holes, it is crucial to know which default hole is replaced with $\square$.

When constructing a decomposition tree of forest $h$, we repeatedly pluck a factor from the current forest of node $x$, and place the factor at $x\cdot 1$ and the residue at $x\cdot 0$. Thus, the forests of the factor and the residue can be embedded into the forest of their ancestor nodes. This embedding, dubbed the \emph{Ancestor Embedding}
is useful for reasoning about decompositions. We demonstrate this in \cref{fig:ancestor embedding} and define it both intuitively and formally in the following.
\begin{definition}[Ancestor Embedding -- Intuitive Definition]
\label{def:ancestor embedding intuitive}
Consider a decomposition $\tau$ and two nodes $x,y\in \tau.\nds$ such that $x\le y$. Let $h_x=\tau.\frst(x)$ and $h_y=\tau.\frst(y)$. The \emph{Ancestor Embedding} is a function $\AncEmb(\tau,y,x)\colon\dom(h_y)\to \dom(h_x)$ which assigns for each address $u$ in $h_y$ the corresponding address of $u$ in $h_x$, when $h_y$ is naturally embedded in $h_x$.
\end{definition}
\begin{definition}[Ancestor Embedding -- Formal Definition]
\label{def:ancestor embedding formal}
Consider a decomposition $\tau$ and two nodes $x,y\in \tau.\nds$ such that $x\le y$. Let $h_x=\tau.\frst(x)$ and $h_y=\tau.\frst(y)$. Denote $y=z\cdot \sigma$ and let $M_z=\width(\tau.\frst(z1))$ and $u_z = \spos(\tau.\ctx(z1))$.
The \emph{Ancestor Embedding} $\AncEmb(\tau,y,x)\colon\dom(h_y)\to \dom(h_x)$ is defined inductively as follows.
\[\begin{cases}
    \AncEmb(\tau, y, x)= \id_{\dom(h_x)}   & y = x\\
    \AncEmb(\tau, y, x)(u)=u_z\spa u      & y = x\cdot 1\\
    \AncEmb(\tau, y, x)(u)= 
    \begin{cases}
        u                       & u=u_z\vee u_z\not\lcrs u\\
        u_z\spa M_z \spa v      & u\neq u_z \wedge u = u_z\spa v
    \end{cases}                                 & y = x\cdot 0\\
    \AncEmb(\tau, y, x)= \AncEmb(\tau, {xw},x)\circ \AncEmb(\tau,  y, {xw})                   & y = x\cdot w \cdot \sigma\\
\end{cases}\]
\end{definition}
\begin{figure}[ht]
    \centering
    \includegraphics[width=0.3\linewidth]{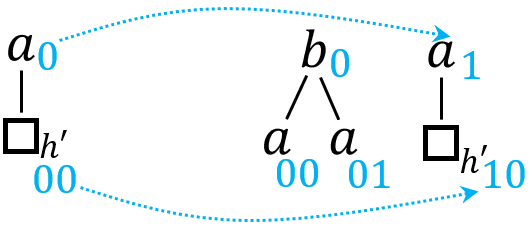}
    \caption{A depiction of the Ancestor Embedding $\AncEmb(\tau,01,0)$ for the nodes $01$ and $0$ in \cref{fig:bin decomp example 2}. Specifically, the map is $0\mapsto 1$ and $00\mapsto 10$.}
    \label{fig:ancestor embedding}
\end{figure}

In the following lemma we point out some natural properties satisfied by the ancestor embedding.
The first is that the root of $x\cdot 1$ is mapped to the position of $\square$ in the context of $x\cdot 1$. 
The second is that in a reference of the form $y@u\to x1$, the position $u$ is mapped to the position of $\square$ in the context of $x1$. 
The third is the converse of the second: if address $u$ in $y$ is mapped to the position of $\square$ for some node $x1$, then there is a reference $y@u\to x1$.
\begin{lemma}
\label{lem:ancestor embedding respects holes and references}
Consider a node $y$ in a decomposition $\tau$. Then the following hold:
\begin{enumerate}
    \item $\AncEmb(\tau,y1,y)(0)=\spos(\tau.\ctx(y1))$.
    \item If $y@u\to x1$ then $x<y$ and $\AncEmb(\tau,y,x)(u)=\spos(\tau.\ctx(x1))$.
    \item For two nodes $x<y$ such that $y\neq x1$, if $\AncEmb(\tau, y,x)(u)=\spos(\tau.\ctx(x1))$ for some $u\in \tau.\frst(y)$, then $y@u\to x1$.
\end{enumerate}
\end{lemma}
\begin{proof}
    Let $\tau$ be a decomposition.
    \begin{enumerate}
        \item Consider some binary node $y$ (so that $y1 \in \tau.\nds$). By \cref{def:ancestor embedding formal,def:spa} we have that $\AncEmb(\tau, y1, y)(0) = \spos(\tau.\ctx(y1)) \spa 0 = \spos(\tau.\ctx(y1))$.
        \item We prove this part by induction on the structure of the references, as per \cref{def:decomp references}.
        The base case is $y=x0$. In this case we indeed we have $x<y$ and by \cref{def:decomp references} we also have $u=\spos(\tau.\ctx(x1))$. Therefore,  by \cref{def:ancestor embedding formal} we get $\AncEmb(\tau,x0,x)(u)=u=\spos(\tau.\ctx(x1))$.
        
 For the induction step, assume $y@u\to x1$ and that $y=z\sigma$ for some $z$. In particular, by \cref{def:decomp references} it holds that $z@u_z\to x1$ (since this is not a base-case reference). 
 By the induction hypothesis we then have $x<z<y$, so $x<y$. 
 We now split to cases:
 \begin{itemize}
     \item If $y=z0$ and $u=u_z$, then by \cref{def:ancestor embedding formal} we have we have $\AncEmb(\tau, y,x)(u)=\AncEmb(\tau,z,x)\circ \AncEmb(\tau,y,z)(u)=\AncEmb(\tau,z,x)(u_z)$, and by the induction hypothesis this equals $\spos(\tau.\ctx(x1))$.
     \item If $y=z0$ and $u\neq u_z$, then by \cref{def:decomp references} we have $u=s\spa 1 \spa v$ where $s=\spos(\tau.\ctx(z1))$,  $M=\width(\tau.\frst(z1))$, and $u_z=s\spa (M+1)\spa v$. Again we have $\AncEmb(\tau,y,z)(u)=u_z$ by \cref{def:ancestor embedding formal}, so the same reasoning as the previous item proves the claim.
     \item If $y=z1$, we have (by \cref{def:decomp references}) that $u_z=s\spa u$. Then, by \cref{def:ancestor embedding formal} we get $\AncEmb(\tau,y,z)=u_z$ and again the claim follows.
 \end{itemize}
 \item We prove the claim by induction on $y$.
 \begin{itemize}
     \item If $y=x0$, then by \cref{def:decomp references} we have $y@u\to x1$ for $u=\spos(\tau.\ctx(x1))$. Note that by \cref{def:ancestor embedding formal} we have $\AncEmb(\tau, y,x)(u)=u$ (since $u=u_z$ in the definition by cases). Thus, we have $u=\spos(\tau.\ctx(x1))$, as needed.
     \item If $y=xz\sigma$ and $\AncEmb(\tau,y,x)(u)=\spos(\tau.\ctx(x1))$, by the inductive assumption we have $xz@\AncEmb(\tau,xz\sigma,xz)(u)\to x1$.
     Denote $s=\spos(\tau.\ctx(xz1))$ and $M=\width(\tau.\frst(xz1))$.
     We split to cases.
     \begin{itemize}
         \item If $\sigma=1$: we have $\AncEmb(\tau,xz1,xz)(u)=s\spa u$, so by \cref{def:decomp references} we get $xz1@u\to x1$.
         \item If $\sigma=0$ and $s\not\lcrs \AncEmb(\tau,xz0,xz)(u)$ then by  \cref{def:ancestor embedding formal} we have $u=\AncEmb(\tau,xz0,xz)(u)$.  Together with \cref{def:decomp references} this implies $y@u\to x1$.
         \item If $\sigma=0$ and $s\lcrs\AncEmb(\tau,xz0,xz)(u)$, notice that $u\neq s$. Indeed, otherwise we would have that $xz@s\to x1$, and therefore $\tau.\frst(xz0)=\tau.\frst(xz)$ which contradicts $\tau$ to be a valid binary decomposition (specifically, Condition~\ref{itm:decomp:pluck nonempty} of \cref{def:bin decomp}, stating that plucked forests are non-empty).
         
         Then, by \cref{def:ancestor embedding formal} we are at the case of$\AncEmb(\tau,xz0,xz)(u)=s\spa M\spa v$ where $u=s\spa v$ ($1\lcrs v$), and in particular $s\spa (M_1)\lcrs \AncEmb(\tau,xz0,xz)(u)$. Thus, by \cref{def:decomp references} we have that $xz0@u\to x1$ (as $u=s\spa v$), so $y@u\to x1$ as required.
     \end{itemize}
 \end{itemize}
\end{enumerate}
\end{proof}
We conclude this section by formalizing the (intuitive) notion that $\AncEmb(\tau,y,x)$ is determined only by the path in $\tau$ from $x$ to $y$, and not by anything outside it. We start by showing that $\AncEmb$ is not affected by ancestors of $x$.
\begin{lemma}
\label{lem:ancestor embedding not affected by ancestor}
    Consider a decomposition $\tau$ and nodes $x\le y$. For every $z\le x$ denote $x=z\cdot x_z$ and $y=z\cdot y_z$, then $\AncEmb(\tau,y,x)\equiv \AncEmb(\tau[z],y_z,x_z)$.
\end{lemma}
\begin{proof}
  Denote $y=x\cdot y'$, we prove the claim by induction on $y'$.
    \begin{itemize}
        \item If $y'=\epsilon$, i.e. $x=y$, then $\AncEmb(\tau,y,x)$ is the identity on $\dom(\tau.\frst(y))$ (by \cref{def:ancestor embedding formal}). We also have $x_z=y_z$ and therefore $\AncEmb({\tau[z]},{y_z},ex_z)$ is the identity as well.
        \item If $y' = y''\cdot \sigma_y$ then we have $y_z = x_zy'$ and the claim readily follows from the induction hypothesis and \cref{def:ancestor embedding formal}:
        \[
        \begin{split}
            \AncEmb(\tau,y,x)=&\AncEmb(\tau, xy'',x)\circ \AncEmb(\tau, y, xy'')=\\
            &\AncEmb(\tau[z],x_z y'',x_z)\circ \AncEmb(\tau[z],y_z,x_zy'') =\AncEmb(\tau[z],y_z,x_z)
        \end{split}
        \]
    \end{itemize}  
\end{proof}
We now show that $\AncEmb$ is not affected by nodes outside the relevant subtree.
\begin{lemma}
    \label{lem:ancestor embedding not affected by incomparable}
    Consider two decompositions $\tau,\tau'$ and node $z$ such that $\equpto{\tau} z {\tau'}$ and let $x,y$ be nodes such that $z\not<x,y$ and $x< y$. Then $\AncEmb(\tau,y,x) = \AncEmb({\tau'}, y,x)$.
\end{lemma}
\begin{proof}
    By the definition of $\equpto \tau {z} \tau'$ (see notations in \cref{sec:labeling intuitive}), we have $x,y\in \tau.\nds \iff x,y\in \tau'.\nds$ and  $\tau.\frst(x)=\tau'.\frst(x)$ and $\tau.\frst(y)=\tau'.\frst(y)$.

    We now proceed by induction on the distance between $x$ and $y$.
    \begin{itemize}
        \item If $y$ is an immediate child of $x$, notice that $\AncEmb(\tau, y,x)$ and $\AncEmb(\tau', y,x)$ only depend on $\tau.\ctx(x1)$ and $\tau'.\ctx(x1)$, respectively. However, we have $\tau.\ctx(x1)=\tau'.\ctx(x1)$ (since $z\not<y$ and therefore $z\not< x1$ when $y$ is an immediate child of $x$). We thus conclude the equality of $\AncEmb$.
        \item If $y=x\cdot w\cdot \sigma$, then by the induction hypothesis we have $\AncEmb(\tau,xw,x)=\AncEmb(\tau',xw,x)$. As above,  $\AncEmb(\tau,y,xw)$ and $\AncEmb(\tau',y,xw)$ are determined by $\tau.\ctx(xw1)$ and $\tau'.\ctx(xw1)$, respectively, and are
        therefore equal (because $z\not<xw1$ and $\equpto{\tau}{z}{\tau'}$).
        We therefore conclude with 
        \[ 
        \begin{split} \AncEmb(\tau,y,x)=\AncEmb(\tau,xw,x)\circ\AncEmb(\tau,y,xw)=\\ \AncEmb(\tau',xw,x)\circ\AncEmb(\tau',y,xw)=\AncEmb(\tau',y,x)
        \end{split}\] 
    \end{itemize}
\end{proof}

\subsection{Stable Decomposition, Basis, and Minimal Decomposition}
\label{subsec:stable and basis}
Recall from \cref{def:linking context} that when $x@u\to y$, a \emph{linking context} from $x$ to $y$ is induced. 
In \cref{sec:rotations} we introduce certain manipulations of decompositions which cause these linking contexts to appear as contexts in the decomposition itself.  
In order to guarantee a certain ``stability'' of the decompositions under such manipulations, we place emphasis on decompositions that allow a pre-specified stable set of contexts (as per \cref{def:stable context set}), as follows.
\begin{definition}[Stable Decomposition]
    \label{def:stable dec}
    Consider a stable context set $\wHole{V}\subseteq \tilV$ (as per \cref{def:stable context set}) and a decomposition $\tau$.
    We say that $\tau$ is \emph{stable with respect to $\wHole{V}$} if for every $x,y$ such that $x@u\to y$ we have that $\clnk(\tau,x,y) \in \wHole{V}$
\end{definition}
Notice that in a stable decomposition $\tau$ with respect to stable context set $\wHole{V}$ we also have that for every $x1\in \tau.\nds$ it holds that $\tau.x0@\spos(\tau.\ctx(x1))\to x1$ and thus, by stability we have $\tau.\ctx(x1)=\clnk(\tau,x0,x1)\in \wHole{V}$.

A particular case where decompositions are stable ``for free'' is when they do not have inherited references, as we now show.

\begin{lemma}
\label{lem:no inherited decomposition is stable}
  Every decomposition $\tau$ with no inherited reference is stable.
\end{lemma} 
\begin{proof}
    Consider a decomposition $\tau$ with respect to a stable context set $\wHole{V}$ with no inherited references (\cref{def:decomp references}). That is, every reference is of the form $\tau.x0@u\to y1$ for some $x,y$. In order to prove that $\tau$ is stable, we need to show that for every such reference it holds that $\clnk(\tau,x0,y1)\in \wHole{V}$ (see \cref{def:stable dec}). Fix $y$, we prove this by induction on $x$.
    The base is $x=y$, in which case $\clnk(\tau,x0,y1)=\tau.\ctx(y1)\in \wHole{V}$ by \cref{def:linking context}.

    We proceed with the inductive case. Since $x0@u\to y1$, then $x@u'\to y1$ for some $u'$ (by \cref{def:decomp references}). Denote $f=\tau.\frst(x),f_1=\tau.\frst(x1),f_0=\tau.\frst(x0)$ (see \cref{fig:linking context stability}).
    \begin{figure}[ht]
        \centering
        \includegraphics[width=0.5\linewidth]{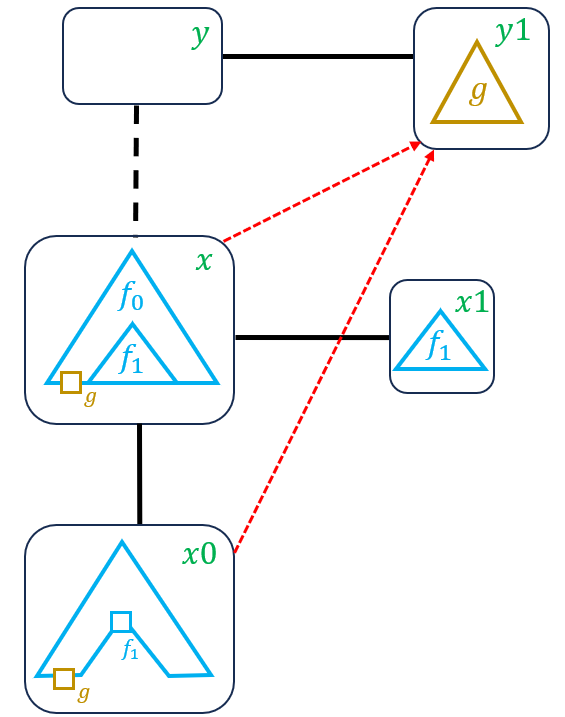}
        \caption{The decomposition in the proof of \cref{lem:no inherited decomposition is stable}. The linking contexts are obtained by replacing $\square_g$ with $\square$.}
        \label{fig:linking context stability}
    \end{figure}
    By \cref{def:linking context,def:bin decomp} we have:
        \begin{align*}
            \clnk(\tau,x0,y1)=f_0[u\mapsto \square]\\
            \clnk(\tau,x,y1)=f[u'\mapsto \square]\\
        \end{align*}
    We claim that $\clnk(\tau,x,y1)\eqUnr \clnk(\tau,x0,y1)$ (see \cref{def:equiv unravel}). Indeed, the difference between the two contexts is that $\clnk(\tau,x0,y1)$ has $\square_{f_1}$ where $\clnk(\tau,x,y1)$ has $f_1$ itself, but this difference disappears under unraveling. 

    By the induction hypothesis, we have $\clnk(\tau,x,y1)\in \wHole{V}$. Since $\wHole{V}$ is a stable context set, it is closed under unraveling (\cref{def:stable context set}) and therefore $\clnk(\tau,x0,y1)\in \wHole{V}$, and we are done.
\end{proof}

Another restriction we place on decompositions is which forests can appear in their leaves. Technically, this restriction arises due to the inductive nature of our proof: we refine a decomposition by going through a sequence of quotients of algebraic structures. This means that at certain points, we decompose to a set of leaves, but then further decompose the leaves themselves. To this end, we draw the leaves from a specific  \emph{basis} set $B\subseteq \tilH$. 

Furthermore, similarly to the case of words, in the following we construct decompositions by extending existing decompositions. Therefore, we need a notion for decompositions that are ``prefixes'' of full decompositions.

\begin{definition}[Full and Partial Decompositions]
    \label{def:full decomposition}
    \label{def:partial decomposition}
    Consider a stable context set $\wHole{V}$ and a basis $B\subseteq \tilH$. 
    \begin{itemize}
        \item We say that a decomposition $\tau$ is \emph{full} with respect to  $\wHole{V}$ and $B$ if $\tau$ is stable with respect to $\wHole{V}$ and for every leaf $x\in \tau.\nds$ we have $\tau.\frst(x)\in B$.

        We denote the set of full decompositions by $\fd(\wHole{V},B)$. 
        We say that a forest $f\in \tilH$ is \emph{decomposable by $\wHole{V}$ to the basis $B$} if there is some $\tau\in \fd(\wHole{V},B)$ such that $\tau.\frst(\epsilon)=f$.

        \item We call a decomposition $\tau$ \emph{partial with respect to $\wHole{V}$ and $B$} if there exists $\tau'\in \fd(\wHole{V}, B)$ such that $\tau\leq \tau'$ 
        (meaning $\tau$ is a substructure of $\tau'$ as per \cref{def:dl order}).
        We denote the set of partial decompositions with respect to $\wHole{V}$ and $B$ by $\pd(\wHole{V}, B)$.
    \end{itemize}
\end{definition}
Observe that while \cref{def:partial decomposition} does not explicitly require a partial decomposition to be stable, it is in fact implicitly required, since a partial decomposition must be extendable to a full decomposition (and in particular, a stable decomposition). Since all the linking contexts remain in the extension, we get in particular that the partial decomposition must also be stable.

In order to bootstrap our construction, we observe that there is always a trivial partial decomposition, which contains a single node, as follows.
We define for every forest $f\in \hdef$ the \emph{singleton decomposition} $\sing(f)\coloneqq \tup{\{\epsilon\}, l_f}$ where $l_f(\epsilon)\coloneqq (f,\square)$.

We also use singleton decomposition for a more convenient notion of decomposition substitution:
Consider a decomposition $\tau$ and some $x\in \tau.\nds$, and a forest $h\in \hdef$.
Denote by $\tau[x\mapsto h]$ the substitution $\tau[x\mapsto \sing(h)]$.

\section{Decomposition Rotations}
\label{sec:rotations}
A central tool in our proof, dubbed \emph{rotations}, provides a method for changing the shape of a binary decomposition tree in a way that maintains certain aspects of the decomposition, and guarantees some desirable properties.
Fundamentally, rotations can be thought of as changing the order at which we pluck out factors. Crucially, the difficult part is maintaining the stability of the trees, and reasoning about their references structure.

We remark that this framework is a novelty of the setting of forests. Indeed, for the setting of words the order at which we pluck subwords is (to some extent) irrelevant, and can be easily changed (see \cref{fig:rotation for words}).
\begin{figure}[ht]
    \centering
    \includegraphics[width=0.5\linewidth]{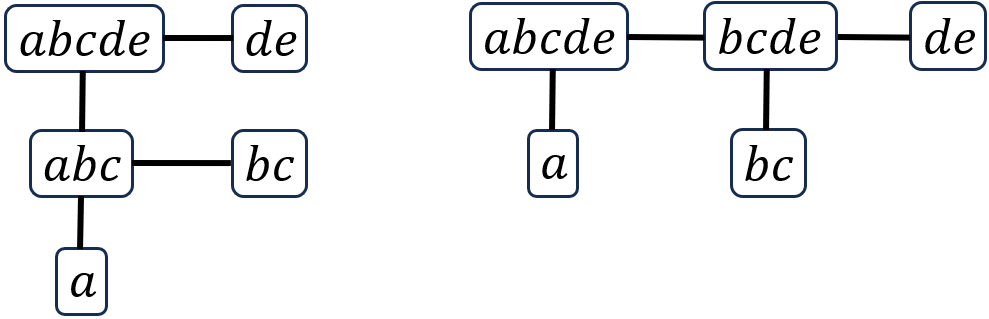}
    \caption{A tree rotation in the setting of words: two binary decomposition trees for $abcde$ with the same leaves. To perform this rotation, we ``violently'' move the subtree that corresponds to the infix $bc$ to the right child of the root instead of the down child. For decomposition (and their stability requirement) this is much more involved.}
    \label{fig:rotation for words}
\end{figure}

\subsection{Sequential Rotation}
\label{sec:sequential rotation}
A \emph{sequential rotation}, depicted in~\cref{fig:factorizations for TB BT rotation lemma}, refers to the following setting. We are decomposing a forest of the form $C_1 C_2 h$ where $C_1,C_2\in \tilV$ and $h\in \tilH$. Intuitively, there are two approaches to such a decomposition: in the \emph{Bottom-Top (BT)} approach, we first pluck the ``bottom'' factor $h$, and then pluck $C_2$ from the residue. In the \emph{Top-Bottom (TB)} approach, we first pluck the factor $C_2 h$, and then pluck $h$ from this factor.

BT and TB yield two different decomposition trees. However, there is a natural correspondence between them, which is captured by sequential rotations. More precisely, we show that the leaves of the two decompositions have a bijection between them that preserves the way they embed in the root, via $\AncEmb$.
    \begin{figure}[ht]
    \centering
    \begin{subfigure}[t]{0.35\textwidth}
        \centering
        \includegraphics[width=\linewidth]{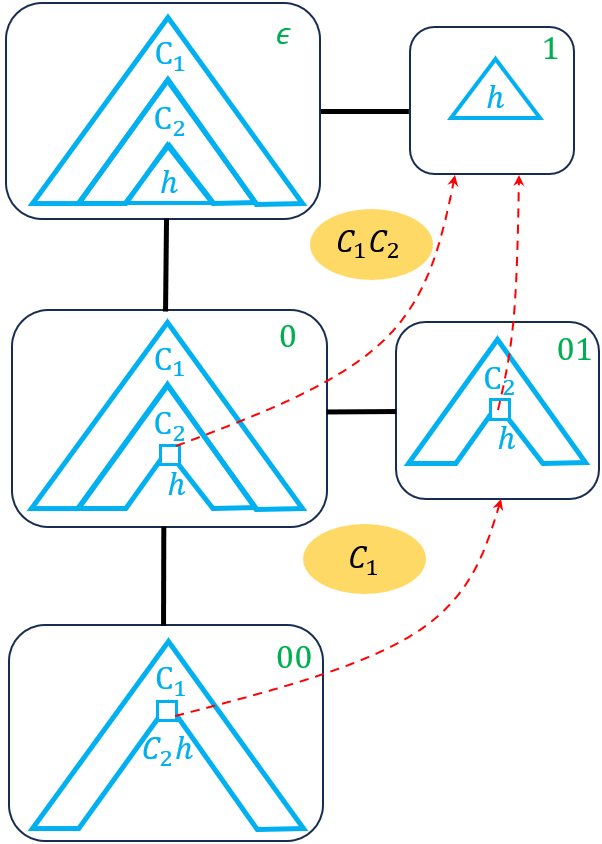}
        \caption{BT factorization.}
    \end{subfigure}
    \hfill
    \begin{subfigure}[t]{0.6\textwidth}
        \centering
        \includegraphics[width=\linewidth]{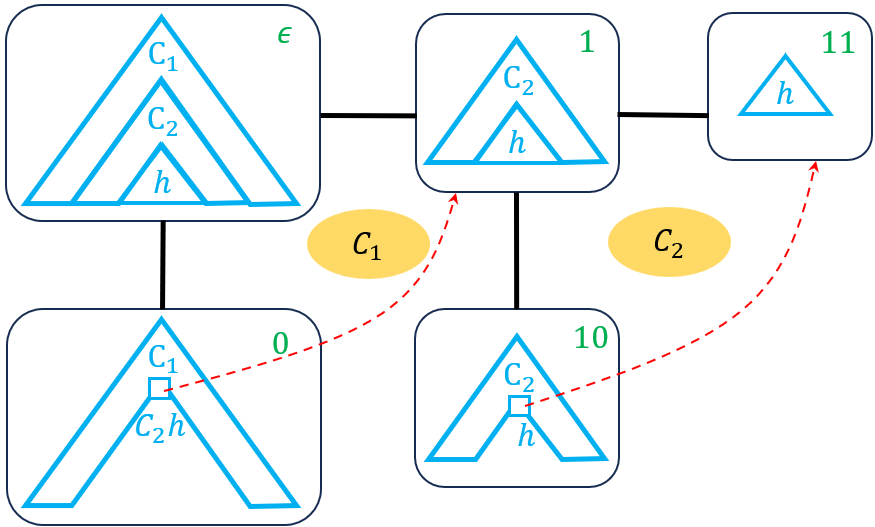}
        \caption{TB factorization}
    \end{subfigure}
    \caption{Two factorizations (BT and TB) for $C_1C_2h$. The translation between them is a \emph{sequential rotation}.}
    \label{fig:factorizations for TB BT rotation lemma}
\end{figure}
\begin{lemma}[Sequential Factorization Ancestor Embedding Equivalence]
\label{lem:sequential factor map}
\label{lem:sequential rotation AncEmb equivalence}
    Consider a binary decomposition $\tau$ such that $\tau.\nds = \{\epsilon,0,1,00,01\}$.
    Such that there exists contexts $C_1,C_2\in \tilV$ for which $\tau.\ctx(1)=C_1\cdot C_2$ and $\tau.\ctx(01)=C_1$.
    Denote $h=\tau.\frst(1)$.
    Now consider a decomposition $\tau'$ which is defined by $\tau'.\nds=\{\epsilon,0,1,10,11\}$ with $\tau'.\frst(\epsilon)=\tau.\frst(\epsilon)$ and $\tau'.\ctx(1)=C_1,\tau'.\ctx(11)=C_2$.
    Then for the bijection $\pi=\left(\begin{array}{c}
    00 \mapsto 0\\
    01\mapsto 10\\
    1 \mapsto 11
    \end{array}\right)$ between the leaves of $\tau$ and $\tau'$ it holds that for every leaf $y$ of $\tau$ we have $\AncEmb(\tau,y,\epsilon)=\AncEmb(\tau',\pi(y),\epsilon)$.
\end{lemma}
\begin{proof}
Denote $s_1=\spos(C_1),\ s_2=\spos(C_2)$ and let 
\[M_1=\width(h),\ M_{01}=\width(C_2\square_h),\ M'_1=\width(C_2h)\] 
We consider each leaf separately. Most of the equalities are immediate from \cref{def:ancestor embedding formal}, but some require additional justification.
\begin{itemize}
    \item For the leaf $1\in \tau.\nds$, consider an address $u$. By \cref{prop: address spos morphism} we have $\AncEmb(\tau,1,\epsilon)(u)= s_1\spa s_2 \spa u$.
    On the other hand, we have
    \[\AncEmb(\tau',11,\epsilon)(u)=\AncEmb(\tau',1,\epsilon)\circ \AncEmb(\tau',11,1)(u)\] and then \[\AncEmb(\tau',11,1)(u)=\spos(\tau'.\ctx(11))\spa u=s_2\spa u\] so 
    \[\AncEmb(\tau',1,\epsilon)(s_2\spa u)=s_1\spa s_2\spa u\]
    giving the equivalence $\AncEmb(\tau,1,\epsilon)=\AncEmb(\tau',11,\epsilon)$.

    \item For the leaf $01\in \tau.\nds$, again consider an address $u$. We have
    \[\AncEmb(\tau,01,\epsilon)(u)=\AncEmb(\tau,0,\epsilon)\circ \AncEmb(\tau,01,0)(u)=\AncEmb(\tau,0,\epsilon)(s_1\spa u)\]
    By instantiating \cref{def:ancestor embedding formal} to our setting, we have
    \[\AncEmb(\tau,0,\epsilon)(s_1\spa u)=\begin{cases}
        s_1\spa u & s_2\not \lcrs u\\
        s_1\spa s_2\spa M_1 \spa v & u = s_2\spa v
    \end{cases}\]
    We compare this to the leaf $10$ of $\tau'$, where the definition yields
    \[\AncEmb(\tau',10,1)(u)=\begin{cases}
        u & s_2\not \lcrs u\\
        s_2\spa M_1 \spa v & u = s_2\spa v
    \end{cases}\]
    and $\AncEmb(\tau',1,\epsilon)(u')=s_1\spa u'$ for every $u'$. Since \[\AncEmb(\tau',10,\epsilon)=\AncEmb(\tau',1,\epsilon)\circ \AncEmb(\tau',10,1)\] we plug in the above and get
    \[
    \AncEmb(\tau',10,\epsilon)(u)=\begin{cases}
        s_1\spa u & s_2\not \lcrs u\\
        s_1\spa s_2\spa M_1 \spa v & u = s_2\spa v
    \end{cases}
    \]
    again yielding the desired equality $\AncEmb(\tau',10,\epsilon)=\AncEmb(\tau,01,\epsilon)$.

    \item Finally, for the leaf $00\in \tau.\nds$, consider an address $u$.  We separate to cases depending on the relative position of $s_1,s_2$ and $u$.
    \begin{itemize}
        \item If $s_1\not \lcrs u$, then in particular $s_1\spa s_2 \not \lcrs u$. Then, we simply have
        \[\AncEmb(\tau,00,\epsilon)(u)=u=\AncEmb(\tau',0,\epsilon)(u)\] 
        and we are done.
        \item Otherwise, $s_1\lcrs u$, i.e., $u$ is a descendant of $s_1$ in the left-child/right-sibling order. By \cref{prop:lcrs spa} we can write $u=s_1\spa v$ for some unique $v$.         
        Note that $u$ cannot be a descendant of $s_1$ itself, since $s_1$ is a leaf in $\tau.\frst(00)$, and therefore has no descendants.
        Thus, we can write $v=\sigma_v v'$ with $\sigma_v>0$ (since if $\sigma_v=0$ we would get $u=s_1\cdot v'$, contradicting $s_1$ being a leaf).

        Write $s_2=\sigma_2 s'_2$ and observe that $\sigma_2\le M_{01}$ by definition. We split to two further cases.
        \begin{itemize}
            \item If $s_2\notin \bbN$ (i.e., if $s'_2\neq \epsilon$) then $C_2$ does not have its $\square$ as a root, and therefore $M_{01}=\width(C_2\square_h)=\width(C_2h)=M'_1$.
            By \cref{prop:spa and concat} we have $s_2=(\sigma_2\cdot 0)\spa s'_2$. Now, since $\sigma_2\le M_{01}$ and $\sigma_v>0$ we have
            \[(\sigma_2\cdot 0)\spa s'_2\not \lcrs ((M_{01}+\sigma_v)\cdot 0)\spa v'\]
            Indeed, otherwise by \cref{prop:lcrs spa} we would have that $((M_{01}+\sigma_v)\cdot 0)\spa v'=(\sigma_2\cdot 0)\spa s'_2 \spa v''$ for some $v''$, but this is a contradiction since the right-hand side starts with $\sigma_2<M_{01}+\sigma_v$.
            Then, by \cref{def:ancestor embedding formal} we again have
            \[
            \begin{split}
            &\AncEmb(\tau,00,\epsilon)(u)=\AncEmb(\tau,0,\epsilon)(s_1\spa M_{01}\spa v)=\\
            &s_1\spa M_{01}\spa v = s_1\spa M'_1\spa v = \AncEmb(\tau',0,\epsilon)    
            \end{split}
            \]

            \item If $s_2\in \bbN$ then $s_2=\sigma_2$ and $s'_2=\epsilon$. Notice that in this case we have $M'_1=M_{01}+M_1$. Since $s_2\le M_{01}$, we have $M_{01}=s_2\spa (M_{01}-s_2)$. Plugging this into \cref{def:ancestor embedding formal} and using the fact that for $x,y\in \bbN$ we have $x\spa y=x+y$, we have
            \[
            \begin{split}
                \AncEmb(\tau,00,\epsilon)=&\AncEmb(\tau,0,\epsilon)(s_1\spa M_{01}\spa v)\\
                =&\AncEmb(\tau,0,\epsilon)(s_1\spa s_2 \spa (M_{01}-s_2)\spa v)\\
                =&s_1\spa s_2\spa M_1 \spa (M_{01}-s_2)\spa v\\
                =& s_1\spa (s_2+M_1+M_{01}-s_2)\spa v\\
                =& s_1 \spa M'_1\spa v \\
                =& \AncEmb(\tau',0,\epsilon)(u)
            \end{split}
            \]       
        \end{itemize}
    \end{itemize}
\end{itemize}
\end{proof}

Having established the equivalence in the leaves of $\AncEmb$ between the a given BT and TB decomposition, we now show how to obtain a TB decomposition by ``rotating'' a BT decomposition.
\begin{lemma}[BT to TB Rotation]
    \label{lem:dec rot A}
    \label{lem:rotation BT to TB}
    Consider a partial decomposition $\tau$ (over $\wHole{V}, B$) and nodes $x,x0,x1,x00,x01\in \tau.\nds$ such that $x01@u\to x1$
    (as depicted by the dashed red arrow in \cref{fig:factorizations for TB BT rotation lemma}).
    Then $\tau.\ctx(x1) = \tau.\ctx(x01)\cdot \clnk(\tau,x01,x1)$,
    and there exist a partial decomposition $\equpto{\tau'}x\tau$ 
    such that:
    \begin{itemize}
        \item $\tau'.\ctx(x1) = \tau.\ctx(x01)$
        \item $\tau'.\ctx(x11) = \clnk(\tau,x01,x1)$
        \item $\tau'[x0]= \tau[x00]$, $\tau'[x10]=\tau[x01]$, and $\tau'[x11]=\tau[x1]$.
    \end{itemize}
\end{lemma}
\begin{proof}
    We start the proof by focusing on the first two items. To this end, we ignore the suffix of $\tau$ below $x00,x01,x1$. Intuitively, we can think of these nodes as leaves, which they might be. Later on, we replace them with their actual subtrees in order to obtain the last requirement of the lemma.
    Denote the following
    \begin{align*}
        C_1&\coloneqq \tau.\ctx(x1)\\
        C_{01}&\coloneqq \tau.\ctx(x01)\\
        h_1&\coloneqq \tau.\frst(x1)\\
        h_0&\coloneqq \tau.\frst(x0)\\
        h_{01}&\coloneqq \tau.\frst(x01)
    \end{align*}
    Since by \cref{def:decomp references} we have $x0@\spos(C_1)\to x_1$ and we assume $x01@u\to x1$, then by \cref{def:decomp references} we have  $\spos(C_1)=\spos(C_{01})\spa u$.

    We first show that $C_1 = C_{01}\cdot \clnk(\tau, x01,x1)$, using the context-equality criteria of \cref{lem:equal contexts characterization}:
    \begin{itemize}
        \item By the above and by \cref{def:linking context,prop: address spos morphism} we get
        \[
        \begin{split}
        \spos(C1)&=\spos(C_{01})\spa u=\spos(C_{01})\spa \spos(\clnk(\tau,x01,x1))\\
        &=\spos(C_{01}\cdot \clnk(\tau,x01,x_1))
        \end{split}
        \]
        \item We now plug $\square_{h_1}$ in $C_1$ (and recall $h_1= \tau.\frst(x1)$):
        \[C_1(\square_{h_1})=h_0=C_{01}(h_{01})=C_{01}\cdot \clnk(\tau,x01,x1)(\square_{h_1})\]
    \end{itemize}
    By \cref{lem:equal contexts characterization}, the two criteria above imply $C_1=C_{01}\cdot \clnk(\tau,x01,x1)$.
    
    We can now define $\tau''=\tup{d''=\{\epsilon,0,1,10,11\}, l''}$ such that:
    \begin{align*}
        l''(\epsilon) &= \tup{\tau.\frst(x),\square}\\
        l''(0) &= \tup{\tau.\frst(00), \square}\\
        l''(1) &= \tup{\clnk(\tau, x01,x1)(\tau.\frst(x1)),C_{01} }\\
        l''(10) &= \tup{\tau.\frst(x01), \square}\\
        l''(11) &= \tup{\tau.\frst(x1), \clnk(\tau,x01,x1)}
    \end{align*}
    Since $\tau$ is a partial decomposition, then it is in particular stable (as per \cref{def:stable dec}). It then follows that $\tau''$ is also a stable decomposition. Indeed, it has a single right child and $\wHole{V}$ is a stable context set. 
    
    Additionally, observe that $\tau[x],\tau''$ satisfy the conditions of \cref{lem:sequential factor map}.
    Define $\tau' \coloneqq \tau[x\mapsto \tau'']$ (i.e., we plug $\tau''$ in $\tau$ under $x$). We then have that $\tau'$ is a binary decomposition over $\wHole V$, $\equpto{\tau'}x\tau$ and it holds that $\tau'.\ctx(x1) = \tau''.\ctx(1)=\tau.\ctx(x01)$ and $\tau'.\ctx(x11) = \tau''.\ctx(11) = \clnk(\tau.x01,x1)$.

    It remains to show that $\tau'\in \pd(\wHole V, B)$, i.e., that it can be completed to a full decomposition. Intuitively, this follows by using roughly the same extension of $\tau$, correctly adapted to the rotation. We proceed with the details.
    
    Since $\tau\in \pd(\wHole V, B)$, there is a full decomposition $\tau_1\in \fd(\wHole{V},B)$ such that $\tau \leq \tau_1$. We assume without loss of generality that every leaf $y\in \tau.\type^{-1}[\lf]$ which is not below $x$ ($x\not<y$) satisfies $\tau.\frst(y)\in B$ (otherwise we can first complete $\tau$ outside the $x$ subtree).
    Define $\tau_2\coloneqq \tau'\left[\begin{array}{c}
        x0\mapsto \tau_1[x00] \\
        x10\mapsto \tau_1[x01]\\
        x11\mapsto \tau_1[x1]
    \end{array}\right]$. Since $\tau_1$ is full, then $\tau_2$ has all its leaves in $B$, and has contexts in $\wHole V$. We now need to show that $\tau_2$ is stable.
    Since $\equpto \tau x {\tau'}$ we also get $\equpto {\tau_1} x {\tau_2}$.
    Consider $y,z\in \tau_2.\nds$ and $u$ such that 
    $\tau_2.y@u\to z$, and denote denote $z=\hat z1$. We split to cases and prove that $\clnk(\tau_2,y,z)\in \wHole{V}$.
    \begin{itemize}
        \item If $x\not\leq y,z$:
        by \cref{lem:ancestor embedding respects holes and references,lem:ancestor embedding not affected by incomparable} we get that 
        $\tau_1.y@u\to z$.
        Since $\equpto{\tau_1} x {\tau_2}$ we also have $\tau_1.\frst(y)=\tau_2.\frst(y)$.
        We conclude that $\clnk(\tau_2,y,z)=\clnk(\tau_1,y,z)\in \wHole{V}$.
        \item If there is some $w\in \{x0,x10,x11\}$ (in $\tau'$) such $w\leq y,z$,
        recall (as in \cref{lem:sequential factor map}) that $\pi=\left(\begin{array}{c}
             00\mapsto 0\\
             01\mapsto 10\\
             1\mapsto 11
        \end{array}\right)$ maps between the leaves of $\tau[x]$ and $\tau''$.
        Denote $w = x\pi(v)$ and $y=w\cdot y',z=w\cdot z'1$.
        By \cref{lem:ancestor embedding not affected by ancestor} and because $\tau_2[w]=\tau_1[xv]$ we have $\AncEmb(\tau_2,y,\hat z)=\AncEmb(\tau_2[w],y',z')=\AncEmb(\tau_1[xv],y',z')=\AncEmb(\tau_1,xvy',xvz')$.
        Hence (by \cref{lem:ancestor embedding respects holes and references}) 
        $\tau_1.xvy'@u\to xvz'$,
        and because $\tau_1.\frst(xvy')=\tau_1[xv].\frst(y')=\tau_2[w].\frst(y')=\tau_2.\frst(y)$,
        we conclude that $\clnk(\tau_2,y,z)=\clnk(\tau_1,xvy',xvz'1)\in \wHole{V}$.
        
        \item If $x\not\leq z$ and there is a leaf $w$ of $\tau'$ such that $w\leq y$ (denote $y=w\cdot y'$), again recall $\pi$ as per \cref{lem:sequential factor map}.
        Notice that (by the assumptions and \cref{lem:ancestor embedding not affected by ancestor,lem:ancestor embedding not affected by incomparable,def:ancestor embedding formal}) we have $\AncEmb(\tau_2,x,\hat z)=\AncEmb(\tau_1,x,\hat z)$,$\AncEmb(\tau_2,w,x)=\AncEmb(\tau_1,xv,x)$ and $\AncEmb(\tau_2,y,w)=\AncEmb(\tau_1,xvy',xv)$.
        Overall, together with \cref{lem:ancestor embedding respects holes and references} we conclude 
        $\tau_1.xvy'@u\to z$, and since $\tau_1.\frst(xvy')=\tau_2.\frst(y)$ we have $\clnk(\tau_2,y,z)=\clnk(\tau_1,xvy',z)\in \wHole{V}$.
        
        \item If $z = x11$ and $x10\leq y$ (denote $y=x10y'$), we have \[\AncEmb(\tau_2,y,x1)(u)=\AncEmb(\tau_2,x10,x1)(\AncEmb(\tau_2,y,x10)(u))=\spos(\tau_2.\ctx(x11))\]
        By \cref{def:ancestor embedding formal} we obtain $\AncEmb(\tau_2,y,x10)(u)=\spos(\tau_2.\ctx(x11))$.
        Denote $u'=\spos(\tau_2.\ctx(x11))=\spos(\clnk(\tau_1,x01,x1))$.
        By \cref{lem:ancestor embedding not affected by incomparable,lem:ancestor embedding not affected by ancestor} we have that $\AncEmb(\tau_1,x01y',x01)(u)=u'$.
        By \cref{def:linking context} we get that 
        $\tau_1.x01@u'\to x1$, so by \cref{lem:ancestor embedding respects holes and references}, we have $\AncEmb(\tau_1,x01,x)(u')=\spos(\tau_1.\ctx(x1))$.
        
        Overall we have $\AncEmb(\tau_1,x01y',x)(u)=\spos(\tau_1.\ctx(x1))$, so 
        $\tau_1.x01y'@u\to x1$.
        Since $\tau_2.\frst(y)=\tau_1.\frst(x01y')$ we conclude $\clnk(\tau_2,y,x11)=\clnk(\tau_1,x01y',x1)\in \wHole{V}$.
        
        \item If $z=x1$ and $x0 \leq y$ (denote $y=x0y'$), then similarly to the previous case, we have 
        \[\AncEmb(\tau_2,y,x)(u)=\AncEmb(\tau_2,x0,x)\circ \AncEmb(\tau_2,y,x0)(u)=\spos(\tau_2.\ctx(x1))\]
        Thus, 
        \[\AncEmb(\tau_1,x00y',x00)(u)=\AncEmb(\tau_2,y,x0)(u)=\spos(\tau_2.\ctx(x1))=\spos(\tau_1.\ctx(x01))\]
        Together with \cref{def:ancestor embedding formal,lem:ancestor embedding respects holes and references} we conclude $\AncEmb(\tau_1,x00y',x0)(u)=\spos(\tau_1.\ctx(x01))$, so 
        $\tau_1.x00y'@u\to x01$.
        Overall $\clnk(\tau_2,y,z)=\clnk(\tau_1,x00y',x01)\in \wHole{V}$.
        
    \end{itemize}
    We conclude that $\tau_2$ is stable, and therefore $\tau'$ is a partial decomposition.

    Note, however, that we have yet to consider the third requirement of the lemma (namely the equivalence of the subforests rooted at the leaves $x00,x01$, and $x1$ with their $\tau$ analogues). 
    We now easily obtain it: recall that $\tau_2$ is obtained by plugging in subtrees of $\tau_1$ into the rotated $\tau''$, and that $\tau_1$ extends $\tau$. We can therefore modify $\tau'$ so that instead of plugging singletons in the rotated tree, we plug the subtrees of $\tau$. Thus, we essentially ``prune'' $\tau_2$ so that its subtrees come from $\tau$ instead of $\tau_1$.
    Then $\tau'$ remains a partial decomposition (since it is still a prefix of the stable $\tau_2$), and now it satisfies the last requirement, namely that 
    $\tau'[x0]= \tau[x00]$, $\tau'[x10]=\tau[x01]$, and $\tau'[x11]=\tau[x1]$.
\end{proof}

A dual argument to \cref{lem:rotation BT to TB} shows that we can obtain a BT decomposition by rotating a TB decomposition. We therefore have the following.
\begin{lemma}[TB to BT Rotation]
    \label{lem:rotation TB to BT}
    Consider a partial decomposition $\tau$ (over $\wHole{V}, B$) and nodes $x,x0,x1,x11,x10\in \tau.\nds$,
    Then there exists a partial decomposition $\equpto{\tau'}{x}{\tau}$ such that:
    \begin{itemize}
        \item $\tau'.\ctx(x1)=\tau.\ctx(x1)\cdot \tau.\ctx(x11)$
        \item $\tau'.\ctx(x01)=\tau.\ctx(x1)$
        \item $\tau'[x00]= \tau[x0]$, $\tau'[x01]=\tau[x10]$, and $\tau'[x1]=\tau[x11]$.
    \end{itemize}
\end{lemma}

\subsection{Parallel Rotation}
\label{sec:parallel rotation}
Our second type of rotation concerns the setting where we pluck two subforests that are not nested, i.e., they are ``parallel'', as depicted in \cref{fig:LR and RL decompositions}.
We then have a choice of plucking the left sub-forest before the right one (LR), or vice-versa (RL), yielding two different, but highly-related, decompositions. 

\begin{figure}[ht]
    \centering
    \begin{subfigure}[t]{0.4\textwidth}
        \centering
        \includegraphics[width=\linewidth]{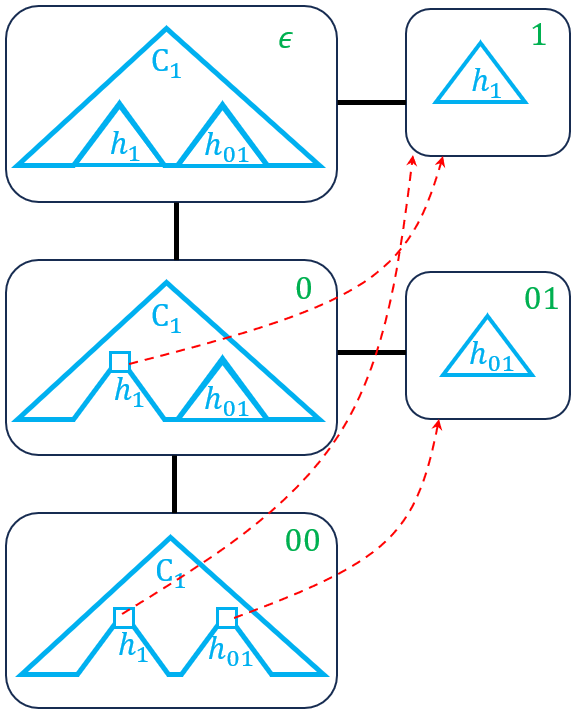}
        \caption{LR factorization.}
    \end{subfigure}
    \hfill
    \begin{subfigure}[t]{0.4\textwidth}
        \centering
        \includegraphics[width=\linewidth]{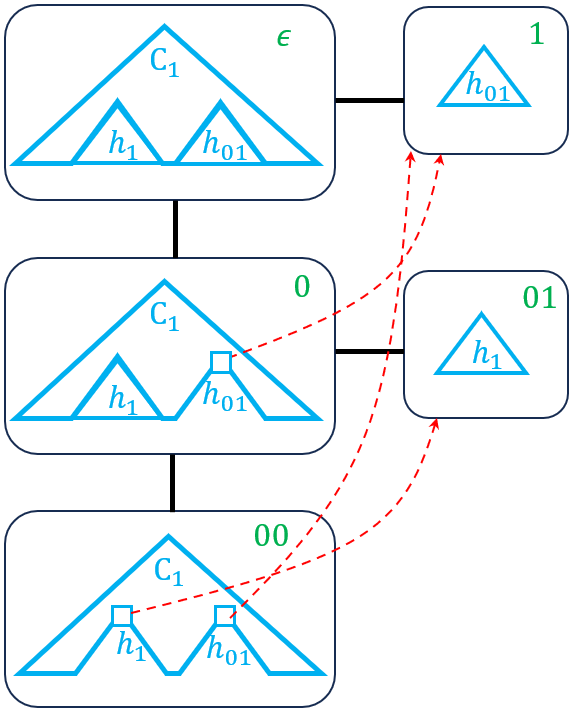}
        \caption{RL factorization}
    \end{subfigure}
    \caption{Two factorizations (LR and RL). The translation between them is a \emph{parallel rotation}.}
    \label{fig:LR and RL decompositions}
\end{figure}

Similarly to \cref{sec:sequential rotation}, we start by identifying a bijection between the nodes of LR and RL decompositions. Then, we proceed to show how we obtain an RL decomposition from an LR one, and vice versa.
\begin{lemma}[Parallel Factorization Ancestor Embedding Equivalence]
\label{lem:parallel factor map}
Consider a binary decomposition $\tau$ with $\tau.\nds = \{\epsilon, 0, 1, 00, 01\}$ 
such that exist $u_1,u_{01}$ with 
$00@u_1\to 1$ and $00@u_{01}\to 01$. 
Let \begin{align*}
    &C_1 = \tau.\ctx(1)\\
    &C_{01} = \tau.\ctx(01)\\
    &h_1=\tau.\frst(1)\\
    &h_{01} = \tau.\frst(01)\\
    &h_{00} = \tau.\frst(00)
\end{align*}
We obtain $\tau'$ from $\tau$ by setting: $\tau'.\nds = \tau.\nds$, $\tau'.\ctx(1) = C_{01}[u_1\mapsto h_1]$ (denoted $C'_{01}$) and $\tau'.\ctx(01) = \clnk(\tau, 00, 1)$ (denoted $C'_1$).  Then $\tau'$ is a decomposition which for the mapping $\pi = \left(\begin{array}{c}
     1\mapsto 01  \\
     01 \mapsto 1\\
     00 \mapsto 00
\end{array}\right)$ satisfies $\AncEmb(\tau,y,\epsilon)=\AncEmb(\tau',\pi(y),\epsilon)$ for every $y\in \{1,01,00\}$.
\end{lemma}
\begin{proof}
    Notice that $\tau'$ is a valid binary decomposition, with $\tau'.\frst(1)=h_{01}$ and $\tau'.\frst(01)=h_1$, and since $\tau$ is stable (over $\wHole{V}$) we have $C'_1=\tau'.\ctx(01)=\clnk(\tau,00,1)\in\wHole{V}$ and  $C'_{01}=\tau'.\ctx(1)=C_{01}[u_1\mapsto h_1]\eqUnr C_{01}\in \wHole{V}$. 
    Additionally, $C'_{01}(h_{01})=\tau.\frst(\epsilon)$, $C'_1(h_1)=\tau'.\frst(0)=C'_{01}(\square_{h_{01}})$,
    and $\tau'.\frst(00)=C'_1(\square_{h_1})=h_{00}$.
    Notice that $\pi$ matches between $\tau.\frst$ and $\tau'.\frst$.
    
    The remainder of the proof is painful accounting, complicated by the fact that $h_{1}$ and $h_{01}$ might have several roots (i.e., have width larger than $1$). Technically, we consider each possible configuration of $h_1$ and $h_{01}$ (i.e., whether they are nested or incomparable), and for each such configuration consider each possible address $u$ in their domain. Moreover, this is done for each leaf, leading to a bunch of cases.
    
    Denote by $s_1,s_{01},s'_1,s'_{01}$ the respective $\spos$ of $C_1,C_{01},C'_1,C'_{01}$, and let $M_1=\width(h_1)$ and $M_{01}=\width(h_{01})$.
    Notice that $\tau.0@s_1\to1$, thus by \cref{def:decomp references} if $s_{01}\lcrs s_1 \wedge s_{01}\spa (M_{01}+1)\not\lcrs s_1$ then $\tau.01@*\to 1$. However, the latter is not the case, since we already have $\tau.00@u_1\to 1$ by assumption. Therefore, we have that 
    \[s_{01}\lcrs s_1\rightarrow s_{01}\spa (M_{01}+1)\lcrs s_1\]
    Indeed, this intuitively means that if $s_{01}$ is ``left'' of $s_1$, then it remains left of it also after adding the width of the plucked forest $h_{01}$ (otherwise there is no reference back to $1$ from $00$).
    We proceed by expressing $s'_1,s'_{01}$ in terms of $s_1,s_{01}$:
    \begin{itemize}
        \item For $s'_1$:
        Recall that $s'_1=\spos(\clnk(\tau,00,1))=u_1$, and therefore  $\tau.00@s'_1\to 1$ (by the assumption $\tau.00@u_1\to 1$). 
        Since  $\tau.0@s_1\to 1$, then by \cref{def:decomp references} we have that \[s'_1 = \begin{cases}
        s_1 & s_{01} \not\lcrs s_1\\
        s_{01} \spa w_1 & s_1 = s_{01} \spa M_{01} \spa w_1
    \end{cases}\]
    \item For $s'_{01}$:
    again by \cref{def:decomp references} we have 
    \[u_1=\begin{cases}
        s_1 & s_{01}\not\lcrs s_1\\
        s_{01}\spa v_1 & s_1=s_{01}\spa M_{01}\spa v_1
    \end{cases}\]
    Recall that $s'_{01}=\spos(C_{01}[u_1\mapsto h_1])$, and therefore (by direct computation) we have 
    \[s'_{01}=\begin{cases}
        s_{01} & u_1\not\lcrs s_{01}\\
        u_1 \spa M_1 \spa v_{01} & s_{01}=u_1\spa v_{01}
    \end{cases}\]
    Plugging in $u_1$ in the latter we conclude that 
    \[s'_{01} = \begin{cases}
        s_{01} & s_1\not\lcrs s_{01}\\
        s_1 \spa M_1 \spa w_{01} & s_{01} = s_1 \spa w_{01} 
    \end{cases}\]
    \end{itemize}
    We now show the equality of $\AncEmb$ for each leaf $y$, based on $s'_1$ and $s'_{01}$ above.
    \begin{itemize}
        \item For $y=00$:
        \begin{itemize}
            \item If $s_1,s_{01}$ are $\lcrs$ incomparable, then we have $s_1 = s'_1,s_{01} = s'_{01}$ and we get the desired $\AncEmb(\tau,00,\epsilon) = \AncEmb(\tau',00,\epsilon)$.
            \item If $s_1\lcrs s_{01}$ (denote $s_{01} = s_1\spa w_{01}$) then recall we have $s'_{01} = s_1\spa M_1 \spa w_{01}$, and $s'_1 = s_1$. Note that this case is not depicted in \cref{fig:LR and RL decompositions} -- it reflects the case where $h_{01}$ is nested within $h_1$.
            We now compute $\AncEmb(\tau,00,\epsilon)(u)$ by first computing $\AncEmb(\tau, 00, 0)(u)$, as per \cref{def:ancestor embedding formal}.
            We therefore have 
            \[\AncEmb(\tau, 00, 0)(u) = \begin{cases}
                u & s_{01}\not\lcrs u\\
                s_1\spa w_{01}\spa M_{01}\spa v & u = s_{01}\spa v = s_1 \spa w_{01}\spa v
            \end{cases}\]
            We can now split to cases according to whether $u$ is below $s_1$ and $s_{01}$ or not. This yields
            \[
            \begin{split}
            &\AncEmb(\tau, 00, \epsilon)(u) = 
            \AncEmb(\tau, 0, \epsilon)\circ \AncEmb(\tau, 00, 0)(u)=\\
            &\begin{cases}
                u & s_1\not\lcrs u\\
                s_1 \spa M_1 \spa v & u = s_1\spa v \wedge s_{01}\not\lcrs u\\
                s_1 \spa M_1 \spa w_{01} \spa M_{01} \spa v & u = s_{01} \spa v
            \end{cases}
            \end{split}
            \]
            A similar analysis for $\tau'$ yields the following.
            \[\AncEmb(\tau', 00, 0)(u) = \begin{cases}
                u & s_1 \not\lcrs u\\
                s_1 \spa M_1 \spa v & u = s_1 \spa v
            \end{cases}\] 
            and again introducing a split on $s_{01}$ and $u$ yields
            \[\AncEmb(\tau',00,\epsilon)(u) = \begin{cases}
            u     & s_1\not\lcrs u\\
            s_1 \spa M_1 \spa v     & u = s_1 \spa v \wedge s_{01}\not\lcrs u \\
            s'_{01} \spa M_{01} \spa v & u = s_{01}\spa v = s_1\spa w_{01}\spa v
            \end{cases}\]
            and the latter is equal to $\AncEmb(\tau,00,\epsilon)(u)$, so we are done.
            
            \item If $s_{01} \lcrs s_1$ denote $s_1 = s_{01} \spa M_{01} \spa w_1$. This is again a case not depicted in \cref{fig:LR and RL decompositions}.
            Recall that $s'_{01} = s_{01}$ and $s'_1 = s_{01}\spa w_1$.
            We use \cref{def:ancestor embedding formal} to compute
            \[
                \begin{split}
                &\AncEmb(\tau,00, \epsilon)(u) = \begin{cases}
                \AncEmb(\tau,0,\epsilon)(s_{01}\spa M_{01}\spa v) & u = s_{01}\spa v\\
                \AncEmb(\tau,0,\epsilon)(u) & s_{01} \not\lcrs u
            \end{cases} \\
            &=\begin{cases}
                u & s_{01} \not\lcrs u\\
                s_{01}\spa M_{01} \spa v   & u = s_{01}\spa v \wedge s'_1 \not\lcrs u\\
                s_{01} \spa M_{01} \spa w_1 \spa M_1 \spa v   & u = s'_1 \spa v = s_{01}\spa w_1 \spa v
            \end{cases}\\
            &=\begin{cases}
                \AncEmb(\tau',0,\epsilon)(u) & s'_{01} \not\lcrs u\\
                s'_{01}\spa M_{01} \spa v   & u = s'_{01}\spa v \wedge s'_1 \not\lcrs u\\
                \AncEmb(\tau',0,\epsilon)(s'_{01} \spa w_1 \spa M_1 \spa v) & u = s'_1 \spa v
            \end{cases}
            \end{split}
            \]
            and the latter is exactly $\AncEmb(\tau',00,\epsilon)(u)$ as required.
        \end{itemize}
        \item For $y=01$:
        \begin{itemize}
            \item If $s_1 \lcrs s_{01}$ denote $s_{01} = s_1\spa w$. 
            Notice that in this case $s'_{01} = s_1\spa M_1\spa w$, so
            \[\AncEmb(\tau,01,\epsilon)(u) = \AncEmb(\tau,0, \epsilon)(s_{01}\spa u) = s_1\spa M_1 \spa w \spa u = s'_{01}\spa u = \AncEmb(\tau',1, \epsilon)(u)\]
            \item If $s_1\not \lcrs s_{01}$, 
            we show that $s_1\not\lcrs s_{01}\spa u$. Indeed, if $s_{01}\lcrs s_1$, then recall the implication $s_{01}\lcrs s_1\rightarrow s_{01}\spa (M_{01}+1)\lcrs s_1$, which precludes $s_1\lcrs s_{01}\spa u$ (for every $u\in\dom(h_{01})$), since $M_{01}=\width(h_{01})$.
            Otherwise, if $s_{01}\not \lcrs s_1$, then clearly $s_1\not\lcrs s_{01}\spa u$, since no suffix appended to $s_{01}$ can go ``below'' $s_1$.

            Thus, we have $s'_{01}=s_{01}$, and we conclude that $\AncEmb(\tau,01,\epsilon)(u)=\AncEmb(\tau,0,\epsilon)(s_{01}\spa u)=s_{01}\spa u = s'_{01}\spa u$.
            
        \end{itemize}
        \item For $y=1$, the setting is is symmetric to the case above, as $\pi(1) = 01$.
    \end{itemize}
\end{proof}

We can now show the parallel rotation lemma.
\begin{lemma}[LR to RL Rotation]
\label{lem:rotation LR to RL}
    Consider a partial decomposition $\tau$ (over $\wHole V, B$) and $x,x0,x1,x00,x01\in \tau.\nds$ such that 
    $\tau.x00@u_1\to x 1$.
    Then there exist a partial decomposition $\equpto{\tau'}x\tau$ such that:
    \begin{itemize}
        \item $\tau'.\ctx(x1) = \tau.\ctx(x01)[u_1\mapsto \tau.\frst(x1)] \eqUnr \tau.\ctx(x01)$
        \item $\tau'.\ctx(x01)=\clnk(\tau,x00,x1)\eqUnr \tau.\ctx(x1)$
        \item $\tau'[x1]= \tau[x01]$, $\tau'[x01]=\tau[x1]$, and $\tau'[x00]=\tau[x00]$.
    \end{itemize}
    as depicted in \cref{fig:LR and RL decompositions}.
\end{lemma}
\begin{proof}
    Similarly to \cref{lem:dec rot A}, denote by $h_1,h_{01},h_{00}$ the forests $\tau.\frst(x1)$, $\tau.\frst(x01)$, $\tau.\frst(x00)$, respectively, and by $C_1,C_{01}$ the contexts $\tau.\ctx(x1),\tau.\ctx(x01)$, respectively.
    As per \cref{lem:parallel factor map}, denote $C'_{01}\coloneqq C_{01}[u_1\mapsto h_1]$ and $C'_1\coloneqq \clnk(\tau,x00,x1)$,
    and define $\tau''=\tup{d'' = \{\epsilon, 0, 1, 00, 01\},l''}$ such that:
    \begin{align*}
        l''(\epsilon) &= \tup{\tau.\frst(x),\square}\\
        l''(0) &= \tup{C'_{01}(\sqr(h_{01}),\square}\\
        l''(1) &= \tup{h_{01},C'_{01}}\\
        l''(00) &= \tup{h_{00},\square}\\
        l''(01) &= \tup{h_1,C'_1}\\
    \end{align*}
    By \cref{lem:parallel factor map} we have that $\pi=\left(\begin{array}{c}
         1\mapsto 01  \\
         01 \mapsto 1 \\
         00 \mapsto 00
    \end{array}\right)$ preserves $\AncEmb$ between $\tau''$ and $\tau[x]$.
    Define $\tau'\coloneqq \tau[x\mapsto \tau'']$.
    Notice that we already fulfill the requirements $\tau'.\frst(x1)=\tau.\frst(x01)$ and $\tau'.\frst(x01)=\tau.\frst(x1)$.
    We now show that $\tau'\in \pd(\wHole{V}, B)$.
    First consider $\tau_1\in \fd(\wHole{ V}, B)$ which extends $\tau$, and assume without loss of generality that for every leaf of $y\in \tau.\type^{-1}[\lf]$ which is not below $x$ ($x\not<y$) we have $\tau.\frst(y)\in B$.
    Define $\tau_2\coloneqq \tau'\left[\begin{array}{c}
         x1\mapsto \tau_1[x01]  \\
         x01\mapsto \tau_1[x1] \\
         x00\mapsto \tau_1[x00]
    \end{array}\right]$.
    Notice that all the leaves of $\tau_2$ are in $B$, so it remains to prove that $\tau_2$ is a stable decomposition.
    Similarly to \cref{lem:dec rot A} we rely on \cref{lem:parallel factor map} to prove that whenever $\tau_2.y@u\to z$ it holds that $\clnk(\tau_2,y,z)\in \wHole{V}$.
    The only two cases which that differ from the proof of \cref{lem:rotation BT to TB} are the following:
    \begin{itemize}
        \item If $z=x1$: in this case we have $x00\leq y$ (because $\tau_2.x00@u_{01}\to x1$).
        By \cref{lem:ancestor embedding respects holes and references} we have that $\AncEmb(\tau_2,y,x)(u) = \AncEmb(\tau_2,x00,x)\circ \AncEmb(\tau_2,y,x00)(u) = \spos(C'_{01})$.
        Since $\AncEmb(\tau_2,x00,x)(u_{01}) = \spos(C'_{01})$ we get $\AncEmb(\tau_2,y,x00)(u) = u_{01} = \AncEmb(\tau_1,y,x00)(u)$.
        Recall that $\tau.x00@u_{01}\to x01$,
        and therefore $\AncEmb(\tau_1,x00,x0)(u_{01})=\spos(C_{01})=u_{01}$.
        
        We conclude (by \cref{def:ancestor embedding formal}) that 
        $\tau_1.y@u\to x01$, yielding $\clnk(\tau_2,y,x1) = \clnk(\tau_1,y,x01)\in \wHole V$.
        \item If $z=x01$: similarly to the previous case, we obtain $\clnk(\tau_2,y,x01) = \clnk(\tau_1,y, x1) \in \wHole{V}$.
    \end{itemize}
    We therefore conclude that $\tau_2\in \fd(\wHole{V}, B)$ so $\tau'$ satisfies the first two requirements of the lemma.

    We can now extends $\tau'$ by plugging in the respective subforests of $\tau$ in its leaves. This decomposition is still partial, since $\tau_2$ still extends it (as $\tau_1$ extends $\tau$). We thus conclude the third requirement.
\end{proof}

\section{Semantic Properties of Binary Decompositions}
\label{sec:dec lemmas}
Equipped with the rotation toolbox, we can now show that, in a sense, the contexts along a decomposition get ``smaller''. 
As with many of our results -- in the setting of words it is a trivial observation, but becomes intricate for forests.


Consider a partial decomposition $\tau$ and some $x\in \tau.\nds$. There may be multiple ways to extend $\tau$ and decompose $x$ to achieve a \emph{full} decomposition (see \cref{def:partial decomposition,def:full decomposition}).
In order to reason about such extensions, we define a ``local'' version of these possibilities, pertaining to which factors can immediately follow a node.
\begin{definition}[Next Factors]
\label{def:next factorizations}
    Consider a stable context set $\wHole{V}$, a basis $B\subseteq \tilH$ and a node $x\in \tau.\nds$. 
    \begin{itemize}
        \item 
    The \emph{factorization step} of $x$ is $\factor(\tau,x)\coloneqq (\tau.\ctx(x1),\tau.\frst(x1))\in \tilV\times \tilH$ (if $\type(x)=\B$, and undefined for leaves).
    \item The \emph{next factors} of $x$ are:
    \[\nf(\tau,\wHole{V},B,x)\coloneqq \{\factor(\tau',x): \equpto{\tau}{x}{\tau'}, \tau'\in \pd(\wHole{V},B), \tau'.\type(x)=\B\}\]
    \end{itemize}
    When clear from context, denote $\nf(\tau,x)$.
\end{definition}
We often restrict attention only to the next context of a node, rather than the entire factor. We capture this as follows.
\begin{definition}[Next Contexts]
\label{def:next contexts}
Consider a stable context set $\wHole{V}$, a basis $B\subseteq \tilH$, a binary decomposition tree $\tau$ and some $x\in \tau.\nds$.
We define
\[
\ns(\tau,\wHole V, B, x)= \{\tau'.\ctx(x\cdot1): \equpto{\tau}{x}{\tau'},\ \tau'\in \pd(\wHole V, B),\ \tau'.\type(x)=\B\}
\]
When clear from context, denote $\ns(\tau,x)$.

Note that $\ns(\tau,x)=\{C: (C,f)\in \exists f. \nf(\tau,x)\}$.
\end{definition}

An important observation is that \emph{computing} $\nf(\tau,\wHole V, B, x)$ (and $\ns(\tau,\wHole V, B, x)$) might be difficult, as it relies on \emph{every} full decomposition $\tau'$. Since our main result is not algorithmic, this does not affect our results.

The following simple observation is that if a node in a partial decomposition has no next contexts, then its forest is already in $B$.
\begin{lemma}
    \label{lem:no next contexts implies B}
    Consider a binary decomposition $\tau\in \pd(\wHole{V}, B)$ and $x\in \tau.\nds$ such that $\ns(\tau, \wHole{V}, B, x) = \emptyset$, then $\tau.\frst(x)\in B$.
\end{lemma}
\begin{proof}
    Since $\tau\in \pd(\wHole{V}, B)$, by \cref{def:partial decomposition} there exists  $\tau'\in \fd(\wHole{V}, B)$ that extends $\tau$.
    By \cref{def:next contexts} and since $\ns(\tau, \wHole{V}, B, x)=\emptyset$, we conclude that $\tau'.\type(x)=\lf$ (otherwise there would be a next context in $\tau'$). 
    By \cref{def:full decomposition} we conclude that $\tau.\frst(x)=\tau'.\frst(x)\in B$.
\end{proof}

In the following we place further restrictions on the ``allowed'' set of next factors (or contexts), namely ``minimality''. In order to allow enough flexibility in such restriction, we only consider restrictions that are closed under unraveling. Thus, we lift \cref{def:stable context set} to factors by defining a \emph{stable factorization set} as $S\subseteq \tilV\times \tilH$ such that for every $(C,f)\in S$, if $C\eqUnr C'$ and $f\eqUnr f'$ then $(C',f')\in S$.
In order to define minimality, we first equip $\tilV$ with a natural ``prefix'' order\footnote{Readers familiar with Green's relation~\cite{green1951structure} may notice that this is exactly Green's $\le_R$ relation.} (and lift it to factors). We demonstrate this in \cref{fig:order on contexts}.
\begin{definition}[The prefix order $\cleq$.]
\label{def:context prefix order}
    Let $C_1,C_2\in \tilV$, then $C_1\cleq C_2$ if there exists $C\in \tilV$ such that $C_1\cdot C=C_2$.

    For factors $(C,f),(C',f')$ we define $(C,f)\cleq (C',f')$ when $C\cleq C'$ (ignoring the forest).
\end{definition}  

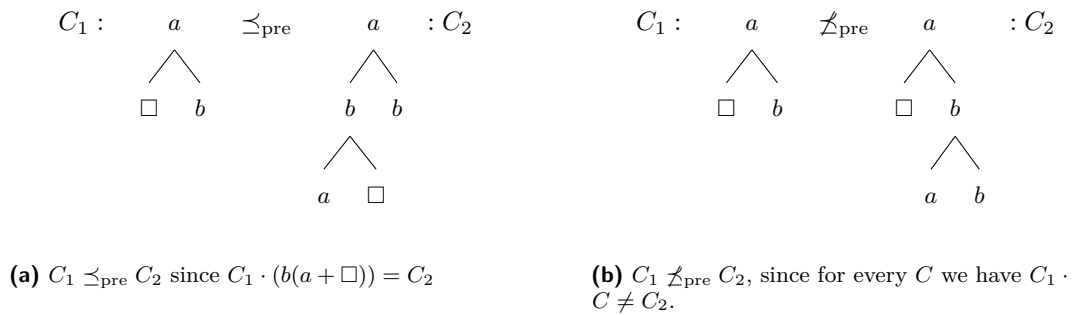
\begin{figure}[ht]
\centering
\begin{subfigure}[t]{0.45\textwidth}
\centering
\[
\begin{array}{c@{\quad}c@{\quad}c}
C_1 : &
\begin{forest}
    [$a$
        [$\square$]
        [$b$]
    ]
\end{forest}
&
\cleq\
\begin{forest}
    [$a$
        [$b$
            [$a$]
            [$\square$]
        ]
        [$b$]
    ]
\end{forest}
\ : C_2
\end{array}
\]
\caption{\(C_1 \cleq C_2\) since $C_1\cdot (b(a+\square))=C_2$}
\label{fig:order-example-yes}
\end{subfigure}
\hfill
\begin{subfigure}[t]{0.45\textwidth}
\centering
\[
\begin{array}{c@{\quad}c@{\quad}c}
C_1: &
\begin{forest}
    [$a$
        [$\square$]
        [$b$]
    ]
\end{forest}
&
\not\cleq\
\begin{forest}
    [$a$
        [$\square$]
        [$b$
            [$a$]
            [$b$]
        ]
    ]
\end{forest}
\ : C_2
\end{array}
\]
\caption{\(C_1 \not\cleq C_2\), since for every $C$ we have $C_1\cdot C\neq C_2$.}
\label{fig:order-example-no}
\end{subfigure}

\caption{Examples illustrating the order $\cleq$ between contexts.}
\label{fig:order on contexts}
\end{figure}

We can now capture the notion of \emph{minimal} decompositions, in  the sense that from a certain node we pluck the ``first'' context available from the context set.
More precisely, a decomposition is minimal below a certain node $x$ if the contexts plucked below $x$ are $\cleq$-minimal among the contexts in $\ns$ (given $\wHole{V}$, basis, and position). That is, we cannot leave a ``smaller'' residue and still extend to a full decomposition. 
\begin{definition}[Minimal Decomposition Below Node]
    \label{def:min dec}
    Consider a partial decomposition $\tau\in \pd(\wHole{V}, B)$, a stable factorization set $S\subseteq \tilH\times \tilV$ (respectively, stable context set $S\subseteq \tilV$), and a node $x\in \tau.\nds$,
    
    We say that $\tau$ \emph{is minimal below $x$ with respect to $\wHole V, B$ and $S$}, if for every $y\in \bns$ such that $\tau.\type(xy)=\B$ we have that $\tau.\factor(xy)$ (respectively, $\tau.\ctx(xy1)$) is $\cleq$-minimal among $\nf(\tau,\wHole V,B,xy)\cap S$ (respectively, $\ns(\tau,\wHole V,B,xy)\cap S$).

    When $S=\tilH\times \tilV$ (respectively, $S=\tilV$) we omit it and say that that $\tau$ is minimal with respect to $\wHole{V},B$.
    When $\wHole{V},B$ (and $S$) are clear, we simply say that $\tau$ \emph{is minimal below $x$}. 
    
    Notice that a minimal decomposition is not necessarily unique.
\end{definition}

A minimal partial decomposition is not particularly restrictive. However, what we are actually after is a way of completing minimal partial decompositions to full ones, while maintaining minimality. The latter is stricter, since we need to maintain minimality while ``striving'' toward the basis $B$. In the following we show something even slightly stronger: if a partial decomposition is minimal below some node $x$, then we can complete the subtree below $x$ to a full decomposition while retaining minimality.
\begin{lemma}[Minimal Full Decomposition]
    \label{lem:min full dec}
    For every $\tau\in \pd(\wHole{V},B)$ and $x\in \tau.\nds$, there is $\tau'\in \pd(\wHole{V},B)$ such that $\equpto{\tau}{x}{\tau'}$ and $\tau'$ is minimal below $x$ (for both contexts and factors), and $\tau'[x]\in \fd(\wHole{V},B)$.
\end{lemma}
\begin{proof}
    We remark that proving the lemma for either contexts or factors implies the other, since a minimal decomposition with respect to the factorization set $\tilH\times \tilV$ implies a minimal decomposition for the context set $\tilV$ (by simple projection), and vice-versa, a minimal decomposition with respect to the context set $\tilV$ induces a decomposition with respect to $\tilV\times \tilH$, by using the factors that actually occur in the tree. We therefore prove the lemma for contexts. 

    Observe that the decomposition $\tau''=\tau[x\mapsto \tau.\frst(x)]$ obtained by ``folding'' the subtree under $x$ in $\tau$ is trivially minimal below $x$, and moreover $\equpto{\tau''}{x}{\tau}$. In particular, the set of minimal decompositions below $x$ is nonempty. 

    Let $\tau'$ be a decomposition that is minimal below $x$, satisfies $\equpto{\tau'}{x}{\tau}$, and is maximal in the sense that no leaf can be further expanded while maintaining these properties. Such $\tau'$ exists by \cref{lem:max dec}, which shows that the depth of all decompositions is bounded.

    We claim that $\tau'[x]$ is a full decomposition. Indeed, assume by way of contradiction that this is not the case, then there exists some leaf $y\in \tau'.\type^{-1}[\lf]$ such that $\tau'.\frst(y)\notin B$.
    By \cref{def:partial decomposition,def:next contexts} we have that $\ns(\tau',\wHole{V},B,y)\neq \emptyset$ and it therefore has some minimal element $C$. We can thus extend $\tau'$ by plucking out a forest from $y$ so that $\tau'.\ctx(y1)=C$. This contradicts the maximality of $\tau'$.
\end{proof}


The following result is our first important step toward semantic reasoning about decompositions. Intuitively, we show that the set of contexts that are available at every node is ``decreasing'', in a sense that becomes precise in \cref{cor:decomp semantically decrease}.
\begin{lemma}[Decompositions are Decreasing]
\label{lem:prefix set}
\label{lem:decomposition decreasing}
Consider a partial decomposition $\tau\in \pd(\wHole{V}, B)$ 
and a node $x\in \tau.\type^{-1}[\B]$, then the following hold.
    \begin{enumerate}
        \item For every $C\in \ns(\tau, \wHole{V}, B, x0)$ there exists $C'\in \ns(\tau, \wHole{V}, B, x)$ such that $C' \eqUnr C$ (i.e., they are \emph{unravel equivalent} as per \cref{def:equiv unravel}).
        \item For every $C\in \ns(\tau,\wHole{V}, B, x1)$ we have $\tau.\ctx(x1)\cdot C\in \ns(\tau,\wHole{V}, B, x)$
    \end{enumerate}
\end{lemma}
\begin{proof}
    \begin{enumerate}
        \item Let $C\in \ns(\tau,\wHole{V}, B, x0)$. By \cref{def:next contexts} there exists a partial decomposition $\tau'\in \pd(\wHole{V},B)$ with $\equpto{\tau}{x0}{\tau'}$ such that $\tau'.\ctx(x01)=C$ (and in particular $\equpto{\tau}{x}{\tau'}$). There are now two possibilities: either the decomposition is BT or LR. More precisely, either $\tau'.x01@u\to x1$ (BT) or $\tau'.x00@u\to x1$ (LR). In the former, we can rotate to TB by \cref{lem:rotation BT to TB} to get $\tau''$ with $\tau''.\ctx(x1)=C$, and in the latter we can rotate to RL with \cref{lem:rotation LR to RL} to get $\tau''$ with $\tau''.\ctx(x1)=C[u\mapsto \tau'.\frst(x1)]$, and notice that $\lab(\tau'.\frst(x00))(u)=\square_{\tau'.\frst(x1)}$, so $\tau''.\ctx(x1)\eqUnr C$. Either way, we obtain a partial decomposition $\tau''$ such that $\equpto{\tau''}{x}{\equpto{\tau'}{x}{\tau}}$ and for $C'=\tau''.\ctx(x1)$ we have $C'\eqUnr C$ and $C'\in \ns(\tau,\wHole{V},B,x)$.

        \item Assume $\ns(\tau, \wHole{V},B,x1)\neq \emptyset$. Let $C\in \ns(\tau, \wHole{V}, B, x1)$. By \cref{def:next contexts} there exists a partial decomposition $\tau'\in \pd(\wHole{V},B)$ with $\equpto{\tau}{x1}{\tau'}$ such that $\tau'.\ctx(x11)=C$. We now have only one possible setting, namely TB. By \cref{lem:rotation TB to BT} there is a partial decomposition $\tau''$ such that $\equpto{\tau''}{x}{\equpto{\tau'}{x}{\tau}}$ and $\tau''.\ctx(x1)=\tau'.\ctx(x1)\cdot C=\tau.\ctx(x1)\cdot C$. In particular, $\tau.\ctx(x1)\cdot C\in\ns(\tau,\wHole{V},B,x)$.
        
    \end{enumerate}
\end{proof}

The decrease along the decomposition directly translates to the images under a morphism. Indeed, recall that for a morphism $\phi$, if $C\eqUnr C'$ then $\phi(C)=\phi(C')$ (as defined in \cref{sec:toolbox for free algebras}). We thus have the following.
\begin{corollary}
\label{cor:decomp semantically decrease}
Consider a partial decomposition $\tau\in \pd(\wHole{V},B)$, a node $x\in \tau.\type^{-1}[\B]$ and a morphism $\varphi\colon\wHole{V}\to V$, then the following hold.
    \begin{enumerate}
        \item $\varphi[\ns(\tau, \wHole{V}, B, x0)] \subseteq \varphi[\ns(\tau, \wHole{V}, B, x)]$
        \item $\varphi(\tau.\ctx(x1))\cdot \varphi[\ns(\tau,\wHole{V}, B, x1)] \subseteq \varphi[\ns(\tau, \wHole{V}, B, x)]$
    \end{enumerate}
\end{corollary}

A particularly useful application of \cref{cor:decomp semantically decrease} is when $V$ has a zero, as follows.
\begin{corollary}
\label{cor:top zeros}
    Consider a semigroup $V$ with a zero $\zeroel\in V$, a stable context set $\wHole{V}$ and a morphism $\varphi\colon\wHole{V}\rightarrow V$,
    Then for every $x\leq y\in \tau.\nds$ we have
    \[\zeroel\in \varphi[\ns(\tau,\wHole{V},B,y)]\implies \zeroel\in \varphi[\ns(\tau,\wHole{V}, B, x)]\]
\end{corollary}
Intuitively, we use this in \cref{sec:decomposition theorems,sec:main theorem bound} to ``decompose away'' all the $\zeroel$ contexts from a forest, and then we are guaranteed that $\zeroel$ does not occur in the next available contexts.

Recall from \cref{lem:no inherited decomposition is stable} that decompositions without inherited references are stable, which is a desirable property.
In the following we show that minimality in fact entails that there are no inherited references. 
Intuitively, this corresponds to the intuition that inherited references correspond to ``bad'' factorization: plucking out a sub-forest, and later on plucking a parent of the sub-forest -- this kind of factorization is not minimal. 
Importantly, this proof relies on our framework of tree rotations (\cref{sec:rotations}).
\begin{lemma}[No Inherited References in Factorization-Minimal Decomposition]
\label{lem:min factorization dec no references}
    Consider a decomposition $\tau$ that is minimal below some node $x$ (with respect to $\wHole{V},B$ and a factorization set $S$, as per \cref{def:min dec}), then there are no inherited references below $x$, as per \cref{def:decomp references}.
\end{lemma}
\begin{proof}
    Notice that if $\tau$ is minimal below $x$ then it is also minimal below $z$ for every $x\leq z$.
    We proceed by way of contradiction: assume that $\tau$ is minimal below $x$ and that there is an inherited reference $y1@*\to z1$ below $x$. In particular, we can assume without loss of generality that $y1@*\to x1$ (since $\tau$ is minimal below $z$ as well, so we can assume $x=z$).

    We now take $\tau,y,x$ to be minimal with this property. That is, there is no smaller $y$ for which this holds, where ``smaller'' is by the standard order on $\bns$, in \emph{any} decomposition $\tau$.

    By the minimality of $y$, we can write $y=x0^n$ for some $n>0$. Indeed, if $y$ is obtained by a path with a right-child, then this right child also has a reference to $x1$, contradicting the minimality of $y$.

    Denote $C_x=\tau.\ctx(x1)$, and $C_y=\tau.\ctx(y1)$, we now split to cases.
    \begin{itemize}
        \item If $n=1$: this puts us under the conditions of \cref{lem:rotation BT to TB} (i.e., a BT decomposition). We therefore have $\tau.\ctx(x1)=C_x=\tau.\ctx(x01)\cdot \clnk(\tau,x01,x1)$ (in particular $\tau.\ctx(x01)\cleq C_x$) and there exists $\equpto{\tau'}{x}{\tau}$ with $\tau'.\ctx(x1)=\tau.\ctx(x01)$ and 
        \[\begin{split}
        &\tau'.\frst(x1)=\clnk(\tau,x01,x1)(\tau.\frst(x1))\\
        &\eqUnr \clnk(\tau,x01,x1)(\square_{\tau.\frst(x1)})=\tau.\frst(x01)
        \end{split}
        \]
        By the definition of a minimal decomposition (\cref{def:min dec}), we have
        $\factor(\tau,x0)=(\tau.\ctx(x01),\tau.\frst(x01))\in S$.
        Overall we have that 
        \[\factor(\tau',x)=(\tau.\ctx(x01),\clnk(\tau,x01,x1)(\tau.\frst(x1)))\in \nf(\tau,f)\cap S\]
        

        This, however, contradicts the minimality of $\tau$: we could have replaced $C_x$ with the smaller $\tau.\ctx(x01)$. So we are done.

        \item 
        If $n>1$: let $m=n-1$ and consider $C'=\tau.\ctx(x0^m1)$ (i.e., the context just above $y$, see \cref{fig:LR to RL in MinLemma}). We first claim that $x0^n0@*\to x0^m1$. Indeed, by \cref{def:decomp references} we have $x0^n@*\to x0^m1$, and there are now two options, either $x0^n1@*\to x0^m1$ (again contradicting the minimality of $n$ since we can use $x0^m$ as $x$, reducing back to the $n=1$ case), or we do not have such a reference, giving the desired $x0^n0@*\to x0^m1$.

        This now puts us in the premise of \cref{lem:rotation LR to RL}, so we can apply an LR to RL rotation (see \cref{fig:LR to RL in MinLemma}). Intuitively, following the rotation we obtain an inherited reference $x0^m1@*\to x1$, contradicting the minimality of $n$ (since $m<n$). 
                
        \begin{figure}[ht]
            \centering
            \begin{subfigure}[t]{0.3\textwidth}
                \centering
                \includegraphics[width=\linewidth]{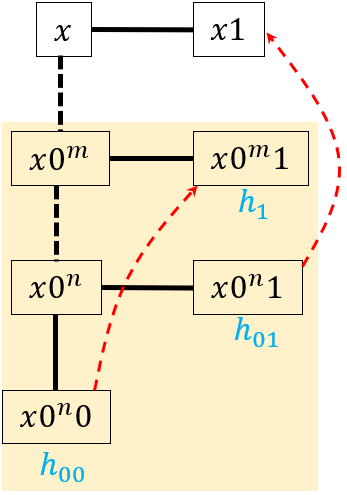}
            \end{subfigure}
            \qquad
            \begin{subfigure}[t]{0.28\textwidth}
                \centering
                \includegraphics[width=\linewidth]{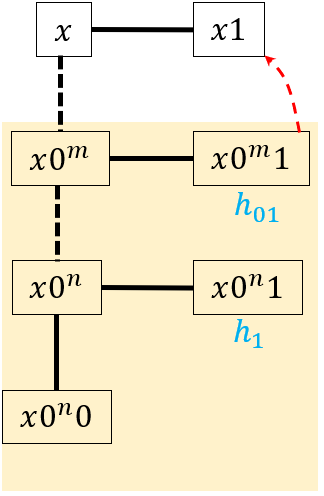}
            \end{subfigure}
            \caption{LR to RL rotation for \cref{lem:min factorization dec no references}. The dashed lines represent references. The rotation is at the subtree rooted at $x0^m$ (the highlighted region).}
            \label{fig:LR to RL in MinLemma}
        \end{figure}
        
        Importantly, however, it is not clear that the rotated tree is minimal. We therefore need to dive into some more details. Denote $h_{00}=\tau.\frst(x0^n0)$,$h_1 = \tau.\frst(x0^m1)$ and $h_{01}=\tau.\frst(x0^n1)$. By \cref{lem:rotation LR to RL} there is a partial decomposition $\equpto{\tau'}{x0^m}{\tau}$ with $\tau'.\frst(x0^m1)=h_{01}$ and $\tau'.\ctx(x0^m1)\eqUnr \tau.\ctx(x0^n1)$, thus by the stability of $S$ we get $\factor(\tau',x0^m)\in \nf(\tau,x0^m)\cap S$.
        In order to reason about the trees in their most ``generic way'' (in particular, so they can be expanded to a minimal decomposition), we ``fold'' everything below $h_1,h_{00}$ and $h_{01}$ as follows. Denote
        \[\tau_1=\tau[x0^m]\left[\begin{array}{c}
             1 \to \sing(h_1)\\
             00 \to \sing(h_{00})\\
             01 \to \sing(h_{01})
        \end{array}\right]
        \text{ and } \tau_2=\tau'[x0^m]\left[\begin{array}{c}
             1 \to \sing(h_{01})\\
             00 \to \sing(h_{00})\\
             01 \to \sing(h_1)
        \end{array}\right]\]
        The definition of $\tau',\tau_1$ and $\tau_2$ puts us in the premise of \cref{lem:parallel factor map} on $\tau_1,\tau_2$, so we get $\AncEmb(\tau_1,01,\epsilon)=\AncEmb(\tau_2,1,\epsilon)$. Since $\AncEmb$ is ``well behaved'' by \cref{lem:ancestor embedding not affected by ancestor,lem:ancestor embedding not affected by incomparable,lem:ancestor embedding respects holes and references} and since $\tau.x0^n1@u\to x1$ for some $u$, we have that $\spos(C_x)=\AncEmb(\tau,x0^n1,x)(u)=\AncEmb(\tau',x0^m1,x)(u)$. 

        We are now ready to consider a minimal decomposition: let $\tau''$ be a minimal decomposition such that $\equpto{\tau''}{x0^m}{\tau}$ and such that $\tau''.\ctx(x0^m1)\cleq \tau'.\ctx(x0^m1)$ (which exists since $\tau'$ is a partial decomposition). By the ordering of the contexts, it follows that $\tau''.\frst(x0^m1)$ (which is plucked out at $x0^m$ leaving $\tau''.\ctx(x0^m1)$) includes the sub-forest $\tau'.\frst(x0^m1)$, as the latter leaves the larger $\tau'.\ctx(x0^m1)$. In particular, $\spos(\tau.\ctx(C_x))$ is in the image of $\AncEmb(\tau'',x0^m1,x0^m)$. Again by \cref{lem:ancestor embedding respects holes and references} this implies that $\tau''.x0^m1@*\to x1$, which is a contradiction to the minimality of $n$ (since $\tau''$ is a minimal decomposition, and $m<n$).
    \end{itemize}
    Since both cases reach a contradiction, we conclude that there are no inherited references below $x$.
\end{proof}

We now wish to obtain an analogue of \cref{lem:min factorization dec no references} to stable context sets. Fortunately, this follows immediately: if a decomposition is minimal below $x$ with respect to a stable context set $S$, then it is also minimal below $x$ with respect to the stable factorization set $S\times \tilH$, and therefore we can apply \cref{lem:min factorization dec no references}. 
\begin{corollary}[No Inherited References in Context-Minimal Decomposition]
    \label{lem:min dec}
    \label{lem:min dec no references contexts}
    Consider a decomposition $\tau$ that is minimal below some node $x$ (with respect to $\wHole{V},B$ and a stable context set $S$, as per \cref{def:min dec}), then there are no inherited references below $x$, as per \cref{def:decomp references}.
\end{corollary}

\section{General Decomposition}
\label{sec:general decomp}
We are now equipped with a robust notion of binary decomposition for forests, and have a versatile toolbox for manipulating such decompositions. For comparison, in the classical setting of words (rather than forests), everything we develop thus far are trivial observations about concatenation.

At the heart of Simon's decomposition lie the \emph{Idempotent nodes} (denoted $\I$-nodes), which we are yet to define. These capture the plucking of several sub-forests simultaneously, where each is mapped to the same idempotent element.
In the setting of trees, defining such decomposition is (unsurprisingly) involved. Moreover, we introduce another type of nodes: \emph{Centipede nodes} (denoted $\C$-node). The latter are a convenient technicality, and we elaborate on them in the following.

Recall that in the setting of words (with respect to a morphism $\phi$), an idempotent node labeled $x$ with children labeled $x_1,\ldots,x_k$ means that $\phi(x)=\phi(x_1)=\ldots=\phi(x_k)=e$ for some idempotent element $e$. Intuitively, we can think of this node as ``hiding'' a binary decomposition of $x$ to the leaves $x_1,\ldots,x_k$, such that at each node in the binary decomposition we have the image $e$ under $\phi$. The latter is the intuition for our $\I$-nodes: they hide a binary decomposition whose contexts are all mapped to an idempotent element. 
Crucially, however, we also require that this decomposition is \emph{stable} (\cref{def:stable dec}). This guarantees that all the linking contexts within the tree are also mapped to the same idempotent. In turn, this allows us to use developed for decompositions in \cref{sec:dec lemmas} as well as tree rotations.

The Centipede nodes are a novel technicality, which can be adopted to words (see \cref{fig:centipede node}). They are prompted by the following scenario (demonstrated over words, for simplicity). Consider a node labeled with a word $x$, and suppose we aim to initially decompose $x=uvw$. Using binary decompositions, we need to decompose $x=u\cdot (vw)$, and then $vw= v\cdot w$, in two ``layers''. It is natural to ask why we do not directly allow a ternary node, to save this split. 
The reason is that in the extreme case, one would decompose $x$ directly to letters in a depth-1 tree, voiding the point of the decomposition theorem. More importantly, it does not really make sense to allow such a node, since given a single child of a multi-child node, we do not directly know where it ``fits'' within its parent node.

However, if allow multiple children and record an ``index'' where each child fits in its parent, then the problem above disappears: we cannot gain too much depth by decomposing to many children, as it would require a large index. Technically, this ``index'' is simply another word with placeholders for its children, i.e., a context over $\tilV$, and so the index itself needs to be further decomposed.
We remark that $\C$-nodes can be avoided entirely (at a certain cost, see \cref{lem:elimination of C nodes}), but they are convenient, as we show in later proofs.

 \begin{figure}[ht]
            \centering
            \begin{subfigure}[t]{0.2\textwidth}
                \centering
                \includegraphics[width=\linewidth]{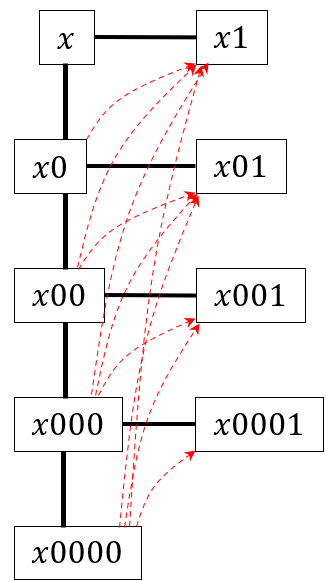}
            \end{subfigure}
            \qquad
            \begin{subfigure}[t]{0.5\textwidth}
                \centering
                \includegraphics[width=\linewidth]{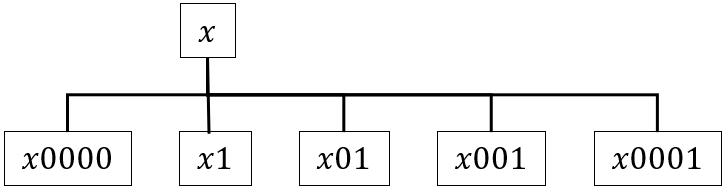}
            \end{subfigure}
            \caption{A Centipede decomposition, and the corresponding $\C$-node. Note the bijection between the leaves. Also note there are no inherited references in the centipede. Finally, the node $x0000$ is the ``index'' -- it has holes corresponding to all other leaves.}
            \label{fig:centipede node}
        \end{figure}

We proceed with the formal definition of general decompositions, as an extension of Binary decompositions (\cref{def:bin decomp}).

\begin{definition}[General Decomposition Tree]
    \label{def:general decomposition}
    Consider a stable context set $\wHole{V}$, a finite semigroup $V$, and a morphism $\varphi\colon\wHole{V}\to V$.
    A \emph{general decomposition tree} with respect to $\phi$ is an $\tilH\times (\wHole{V}\cup \{\square\})$-labeled tree domain $\tau=\tup{\tau.\nds,\tau.\lab}\in \lltds{(\tilH\times (\wHole{V}\cup \{\square\}))}$ equipped with three functions
    \begin{itemize}
        \item $\tau.\frst\colon\tau.\nds\rightarrow \tilH$
        \item $\tau.\ctx\colon\tau.\nds\rightarrow (\wHole{V}\cup \{\square\})$
        \item $\tau.\type\colon\tau.\nds\rightarrow \{\B,\lf,\I,\C\}$
    \end{itemize}
    that satisfies the following conditions for every node $x\in \tau.\nds$.
    \begin{enumerate}
        \item If $\tau.\type(x)=\lf$ then $x$ is a \emph{$\lf$eaf}.
        \item If $\tau.\type(x)=\B$, then $x$ is a \emph{$\B$inary node} satisfying the conditions of \cref{def:bin decomp}.
        \item If $\tau.\type(x)=\I$ then $x$ is an \emph{$\I$dempotent node}. Its set of children is unbounded (but finite), denoted $S=\tau.\nds\cap (x\cdot \bbN)$, and there is a binary decomposition $\tau'$ such that the following hold. 
        \begin{itemize}
            \item $\tau'.\frst(\epsilon) = \tau.\frst(x)$.
            \item There exists an idempotent element $v\in V$ (i.e., $v^2=v$) such that for the stable context set\footnote{by the properties of idempotent elements, monoid morphism, and the artificial inclusion of $\{\square\}$, we get that $\wHole{V_x}$ is indeed a sub-monoid of $\tilV$.} $\wHole{ V_x}=\varphi^{-1}[v]\cup \{\square\}$ and the basis $\tau.\frst[S]$ (i.e., the set of forests in the children of $x$), we have $\tau'\in \fd(\wHole{V_x},\tau.\frst[S])$. 
            \item There is a bijection $g\colon S\rightarrow \tau'.\type^{-1}[\lf]$ such that for every $y\in S$ we have $\tau.\frst(y)=\tau'.\frst(g(y))$. 
        \end{itemize} 
        For an $\I$-node, we set $\tau.\ctx(y)=\square$ for every $y\in S$ (intuitively, the context of $y$ is of no consequence).
        
        We denote the idempotent element $v$ by $\val(\tau,x,\varphi)$.
        
        \item If $\tau.\type(x)=\C$ then $x$ is an \emph{$\C$entipede node}. Its set of children is unbounded (but finite), denoted $S=\tau.\nds\cap (x\cdot \bbN)$, and there is a binary decomposition $\tau'$ such that the following hold:
        \begin{itemize}
            \item $\tau'.\frst(\epsilon)=\tau.\frst(x)$.
            \item $\tau'.\nds\subseteq 0^*(\epsilon+1)$ (see \cref{fig:centipede node}).
            \item $\tau'$ does not have inherited references (as per \cref{def:decomp references}).
            \item There is a bijection $g\colon S\to \tau'.\type^{-1}[\lf]$ such that for every $y\in S$ we have $\tau.\frst(y)=\tau'.\frst(g(y))$
        \end{itemize}
        For a $\C$-node, we also set $\tau.\ctx(y)=\square$ for every $y\in S$.
\end{enumerate}
\end{definition}

\begin{remark}[Full vs Stable Decomposition for $\I$-Nodes]
\label{rmk:I node full decomp vs stable}
In \cref{def:general decomposition} we require that an $\I$-node has a corresponding full decomposition $\tau'$ (\cref{def:full decomposition}). In fact, we only need $\tau'$ to be a stable decomposition (\cref{def:stable dec}). Indeed, since we match the leaves of the decomposition, then we can essentially choose the basis to be exactly the forests in the leaves, rendering a given stable decomposition full.
\end{remark}

Recall that a minimal decomposition does not have inherited references (\cref{lem:min dec no references contexts}) and that a decomposition without inherited references is stable (\cref{lem:no inherited decomposition is stable}). Thus, in light of \cref{rmk:I node full decomp vs stable}, we have the following.
\begin{corollary}
\label{cor:min to I}
    Consider a semigroup morphism $\varphi\colon\wHole{V}\to V$ and an idempotent element $e\in V$.
    In every decomposition $\tau$ that is minimal below $x$ with respect to $\wHole{V}, B$ and the set $S=\varphi^{-1}[e]$ (\cref{def:min dec}) we can be collapse the subtree $\tau[x]$ to obtain a decomposition $\equpto{\tau'}{x}{\tau}$ such that $\tau'[x]$ is a single $\I$-node (and in particular a subtree of depth $1$). 
\end{corollary}

We are finally ready to define our main object of study, namely Full General Decomposition trees.

\begin{definition}[Full General Decomposition]
    \label{def:fdec}
    \label{def:full general decomposition}
    A \emph{full general decomposition} of a forest $f\in \tilH$ over a stable context set $\wHole{V}$, basis $B$, and morphism $\varphi\colon\wHole{V}\rightarrow V$ is a general decomposition $\tau$ such that $\tau.\frst(\epsilon)=f$, for every node $x\in \tau.\nds$ we have $\tau.\ctx(x)\in \wHole{V}\cup\{\square\}$, and $\tau.\frst(x)\in B$ if $\tau.\type(x)=\lf$.

    The set of all full general decompositions over $\wHole{V},B,f$ and $\varphi$ is denoted $\dec(\wHole{V},B,f,\varphi)$.
\end{definition}

Notice that the requirements of a general decomposition in \cref{def:general decomposition,def:full general decomposition} are \emph{local}, in the sense that they place restrictions for each node only on itself and its children. In particular, full general decompositions are closed under subtrees and under subtree-substitution, as follows.

\begin{corollary}
    \label{cor:general decomp closed substitution}
    Consider a general decomposition $\dec(\wHole{V}, B, f, \varphi)$ and $x\in \tau.\nds$, then the following hold.
    \begin{enumerate}
        \item $\tau[x]\in \dec(\wHole{V},B,\tau.\frst(x),\varphi)$.
        \item For a decomposition $\tau'\in \dec(\wHole{V}, B, \tau.\frst(x),\varphi)$ we have $\tau[x\mapsto \tau']\in \dec(\wHole{V}, B, f, \varphi)$.
    \end{enumerate}
\end{corollary}
In contrast to \cref{cor:general decomp closed substitution}, if we consider \emph{stability} (\cref{def:stable dec}) as a property, then it is not local, and the corollary does not hold for it (specifically, the second item fails).

As mentioned above, the parameterization of the decomposition by a basis is to facilitate the inductive proof. However, there is a canonical basis that is ``elementary'', in some sense.
\begin{definition}[Standard Basis]
\label{def:standard basis}
we define the standard basis $\bs \subseteq \tilH$ to be:
\[\bs = A\cup \{\sigma\square_h:\sigma\in A,h\in \hf\}\cup \{\square_{h_1}+\square_{h_2}:h_1,h_2 \in \hf\}\]
\end{definition}
\begin{remark}[$\bs$ does not contain $\square_h$]
\label{rmk:standard basis no squares}
The observant reader may notice that elements of the form $\square_h$ are not present in $\bs$. While these elements are ``more elementary'' than those in $\bs$, we observe that they (almost) never occur as leaves in a decomposition, since obtaining $\square_h$ as a leaf means that the plucked subforest is the entire forest, which is precluded by \cref{def:bin decomp} (Condition~\ref{itm:decomp:pluck nonempty}).

The only exception to this is when the subforest at the root is already $\square_h$. Since this is a degenerate case (and never occurs in the induction), we can safely disregard $\square_h$.
\end{remark}
As the intuition suggests, every forest can be decomposed to the standard basis, regardless of the morphism. Moreover, such a decomposition can be achieved using binary decompositions (i.e., no $\I$ or $\C$ nodes). Intuitively, this is simple: inductively pluck subtrees according to the construction of the forest, until reaching a basis element. 
\begin{lemma}
\label{lem:every forest is decomposable to standard basis}
    For every $f\in \hf$ \footnote{We could show this for $f\in \tilH$, with the exception of $f=\square_h$, see \cref{rmk:standard basis no squares}} there exists a decomposition $\tau\in \fd(\tilV, \bs)$ such that $\tau.\frst(\epsilon)=f$.
    
    In particular, for every morphism $\varphi\colon\wHole{V}\to V$, we have $\tau\in \dec(\wHole{V},\bs,f,\varphi)$.
\end{lemma}

\begin{proof}
    By induction on the generating sequence of $f$. Note that since we only have an inductive assumption on $\hf$, we take care not to introduce subtrees in $\tilH$, unless they are already basis elements.
    \begin{description}
        \item[Base:] $f\in A$. Then, the singleton decomposition $\sing(f)$ is such a decomposition.
        \item[Composition step:] assume $f=\sigma(f')$. 
        Consider the decomposition $\tau$ with $\tau.\nds=\{\epsilon,0,1\}$ such that $\tau.\frst(\epsilon)=f$ and $\tau.\ctx(1)=\sigma(\square)$.
        Notice that $\tau.\frst(0)=\sigma(\square_{f'})\in \bs$ and $\tau.\frst(1)=f'$.
        By the inductive assumption, there is a corresponding decomposition $\tau'\in \fd(\tilV,\bs)$ for $f'$. Then, the decomposition $\tau''=\tau[1\mapsto \tau']\in \fd(\tilV,\bs)$ satisfies the requirements for $f$.
        \item[Addition step:] assume $f=f_1+f_2$.
        Consider the decomposition $\tau$ with $\tau.\nds=\{\epsilon,0,1,00,01\}$ such that $\tau.\frst(\epsilon)=f$ and $\tau.\ctx(1)=f_1+\square$ and $\tau.\ctx(01)=\square+\square_{f_2}$.
        Notice that $\tau$ has no inherited references, and that $\tau.\frst(1)=f_2$ and $\tau.\frst(01)=f_1$ and $\tau.\frst(00)=\square_{f_1}+\square_{f_2}\in \bs$.
        By the inductive assumption on $f_1,f_2\in \hf$, there are corresponding decompositions $\tau_1,\tau_2\in \fd(\tilV,\bs)$. Then, we construct $\tau'=\tau\left[\begin{array}{c}
             1\mapsto \tau_2  \\
             01\mapsto \tau_1 
        \end{array}\right] \in \fd(\tilV,\bs)$ such that $\tau'.\frst(\epsilon)=f$, as required.
    \end{description}
\end{proof}

Recall that Simon's factorization theorem places a bound on the depth of decomposition of any word over a given monoid (and morphism). We now introduce the notion of a depth bound for forests. For technical reasons that become apparent in the proof, we define this notion first for semigroups, and then specialize it to monoids and to the standard basis.
\begin{definition}[Semigroup Decomposition Bound]
\label{def:semigroup decomp bound}
Consider a finite semigroup $S$, a stable context set $\wHole{V}$, a morphism $\varphi\colon\wHole{V}\rightarrow S$, and a basis $B$.
The \emph{decomposition bound} of $\phi$ over $B$ is
\[\norm{\varphi, B}\coloneqq \sup_{f\text{ is decomposable by }\wHole{V}\text{ to }B}\min_{\tau\in \dec(\wHole{V},B,f,\varphi)}\depth(\tau)\]
(note that $\phi$ implicitly determines $\wHole{V}$ and $S$).
We further define:
\begin{itemize}
    \item $\norm{\varphi}=\sup_{B}\norm{\varphi,B}$.
    \item $\norm{S}=\sup_{\varphi:\wHole{V}\to S}\norm{\varphi}$ (where $\varphi$ is taken over any context set $\wHole{V}$).
\end{itemize}
\end{definition}
Thus, the decomposition depth of $\varphi$ over basis $B$ is the maximal depth of a minimal decomposition on any forest, restricted to forests that are actually decomposable. Note that this maximum (rather, supremum) might be infinite. 
\begin{remark}[Decomposability]
\label{rmk:decomposability}
Note that the requirement of $f$ to be decomposable by $\wHole{V}$ to $B$ is independent of the morphism. Indeed, a forest is decomposable if and only if there is a binary decomposition (by e.g., replacing $\C$-nodes and $\I$-nodes with their respective binary sub-decompositions). Finally, note that binary decompositions are independent of the morphism. 
\end{remark}
Also note that for every decomposable forest we have $\dec(\wHole{V},B,f,\varphi)\neq \emptyset$, since binary decompositions are in particular general.

Then, the decomposition depth of a semigroup $S$ is defined by the worst-case morphism and the standard basis.

As discussed above, while $\C$-nodes are a convenient technicality, they are (in a sense) not necessary. In the following result we formalize this, and make explicit the blowup incurred by removing them. Note that our main result (\cref{thm:main}) states that there is a constant bound on the depth of decompositions. This therefore also holds without $\C$-nodes, only with a higher (constant) bound.
\begin{lemma}
\label{lem:elimination of C nodes}
    For every $\tau\in \dec(\wHole{V},\bs,f,\varphi)$ there exists $\tau'\in \dec(\wHole{V},\bs,f,\varphi)$ with no centipede nodes, and such $\depth(\tau')\leq 2^{\depth(\tau)+1}-2$.
\end{lemma}
\begin{proof}
    We start with a general observation about decompositions: consider some decomposition $\tau\in \pd(\wHole{V},B)$ for some $\wHole{V},B$, and let $f=\tau.\frst(\epsilon)$. Suppose that $f$ has a leaf $x$ labeled $\square_h$ for some $h\in \tilH$. We can modify $\tau$ by replacing $\square_h$ with $h$ in all occurrences that stem from $x$ throughout $\tau$, thus obtaining a new decomposition $\tau'$ with $\tau'.\nds=\tau.\nds$ and the property that for every node $w$ it holds that $\tau'.\frst(w)\eqUnr \tau.\frst(w)$, and similarly $\tau'.\ctx(w)\eqUnr \tau.\ctx(w)$ and the reference structure is maintained. 
    Thus, $\tau'$ is ``equivalent'' to $\tau$ up to unraveling. 
    The main difference is that the leaves of $\tau'$ are not not necessarily in $B$, but may contain instances of $h$. 
    
    This is illustrated in \cref{fig:remove C node} -- the node $x000$ has leaves labeled $\square_{h_0},\square_{h_1},\square_{h_2}$. In turn, the relevant subtrees are decomposed in the remaining subtrees. We can therefore replace each $\square_{h_i}$ with the respective $h_i$, and concatenate the decomposition (as in \cref{fig:remove C node 2}).
    
    Next, we recall the structure of the decomposition corresponding to a $\C$-node (see \cref{fig:centipede node}): its root is $x$, and the nodes are $x0^*$ and $x0^*1$ up to some $x0^{k-1}1$ and $x0^k$. Specifically, the bottom-most node (i.e., the ``index'' node) is $x0^k$. By the reference structure, we observe that $\frst(x0^k)$ has among its leaves $\square_{h_0},\ldots, \square_{h_{k-1}}$ where $h_i=\frst(x0^i1)$. 
    In particular, if we replace each $\square_{h_i}$ with $h_i$ in $x0^k$, as described above, we would obtain exactly $\frst(x)$.
    We henceforth assume that in the corresponding $\C$-node, the child corresponding to $x0^k$ is $x0$ (recall that this $\C$-node has $k+1$ children).

    \begin{figure}[ht]
        \begin{subfigure}{0.4\textwidth}
            \centering
            \includegraphics[width=1\linewidth]{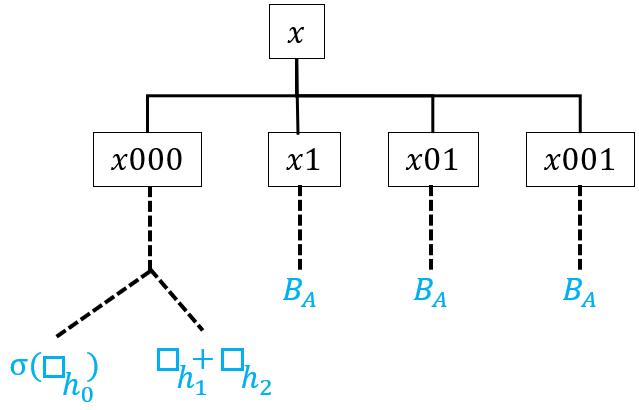}
            \caption{$\C$-node.}
            \label{fig:remove C node 0}
        \end{subfigure}
        \hfill
        \begin{subfigure}{0.3\textwidth}
            \centering
            \includegraphics[width=1\linewidth]{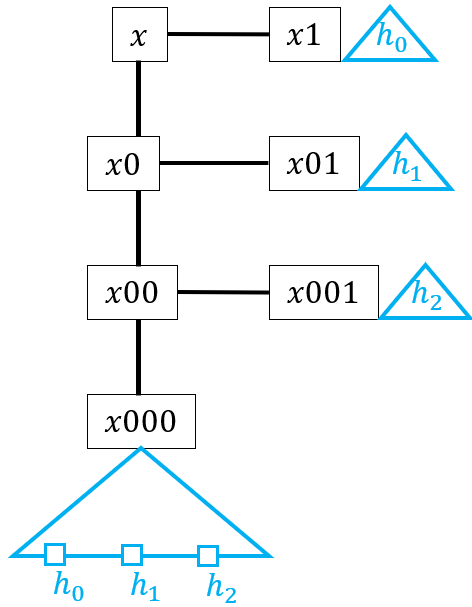}
            \caption{Centipede.}
            \label{fig:remove C node 1}
        \end{subfigure}
        \hfill
        \begin{subfigure}{0.2\textwidth}
            \centering
            \includegraphics[width=1\linewidth]{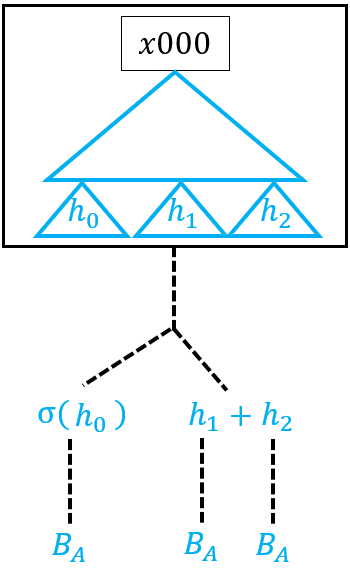}
            \caption{Decomposition}
            \label{fig:remove C node 2}
        \end{subfigure}
        \caption{A $\C$-node with index $x000$ (\cref{fig:remove C node 0}), and its corresponding Centipede (\cref{fig:remove C node 1}). In \cref{fig:remove C node 2} we compose the decompositions and eliminate the $\C$-node.}
        \label{fig:remove C node}
    \end{figure}

    Our approach for eliminating $\C$-nodes is now the following: consider a $\C$ node $x$ in $\tau$. We consider the decomposition subtree $\tau_x$ of $x0$ (i.e., of the index node). 
    We then obtain from $\tau_x$ an equivalent tree as described above: we replace each $\square_{h_i}$ with $h_i$, and follow the same decomposition as $\tau_x$ (which we can do by \cref{cor:general decomp closed substitution}). 
    Note the the new tree $\tau'_x$ has one fewer $\C$-node than $\tau[x]$, and depth \emph{lower} by $1$ (since we remove the $\C$-node and replace it with the child $x0$).
    However, this decomposition no longer ends in $B_A$, but rather has in its leaves the trees $h_i$. We therefore need to continue decomposing these trees to reach $B_A$.
    Note that for each $h_i$ we already have a decomposition of depth at most $\depth(\tau[x])-1$, from the corresponding child of the $\C$-node. Moreover, this decomposition has fewer $\C$-nodes than $\tau$, so we can inductively obtain for it a decomposition without $\C$-nodes (with a corresponding depth blowup).

    Therefore, we now consider each leaf $y$ of $\tau'_x$, and decompose it to $B_A$ (according to the cases of \cref{def:standard basis}), while accounting for the added depth. We denote by $D(n)$ the depth of a decomposition without $\C$ nodes equivalent to a given a decomposition of depth $n$. 
    \begin{itemize}
            \item If $y\in \bs$, then no change is necessary.
            \item If $y=\sigma(h_i)$ (for $\sigma\in A$), this case corresponds to a leaf labeled $\sigma(\square_{h_i})$ of $\tau$. 
            We then factor it using with the context $\sigma(\square)$ obtaining a bottom-child labeled $\sigma(\square_{h_i})\in \bs$ and a right-child labeled $h_i$. We then plug the decomposition of $h_i$ without $\C$-nodes, obtaining a full general decomposition. The depth of this decomposition adds at most $D(n-1)+1$.
            
            \item If $y=\tau.\frst(h_i)+\square_{h}$ for some $h$, this case corresponds to a leaf labeled $\square_{h_i}+\square_{h}$ of $\tau$. 
            We then factor it using the context $\square+\square_{h_i}$, yielding a bottom-child leaf labeled $\square_{h_i}+\square_h$ and a right-child labeled $h_i$, and again we plug the decomposition of $h_i$, adding depth at most $D(n-1)+1$.
            
            \item If $y=\square_h+\tau.\frst(h_i)$, this is analogous to the previous case.
            \item If $y=\square_{\tau.\frst(h_i)}+\square_{\tau.\frst(h_j)}$, this is again analogous to the above, only now we need two factors: first use $\square+\tau.\frst(h_i)$ to pluck out $\tau.\frst(h_i)$ and leave residue $\square_{h_i}+\tau.\frst(h_j)$, and then use $\square_{h_i}+\square$ to pluck out $\tau.\frst(h_j)$. 
            Overall this adds depth at most $D(n-1)+2$, due to the two binary nodes.
    \end{itemize}
    We then replace the original decomposition $\tau$, now with one fewer $\C$-node, by a decomposition of depth at most $D(n-1)$.
    Overall we obtain that a decomposition of depth $n$ can be transformed to a decomposition without $\C$ nodes of depth at most $D(n-1)+D(n-1)+2$. Thus, $D(n)\le 2D(n-1)+2$, and together with $D(0)=0$ we get $D(n)\le 2^{n+1}-2$, concluding the proof.
\end{proof}

\section{Preliminaries Part II -- Green's Relations, Ideals and Congruences}
\label{sec:green prelim}
In the remainder of the proof we make heavy use of tools from semigroup theory, specifically -- Green's relations and related concepts. We bring here several definitions and results. We remark that the results presented in this section are well known, but have several presentations. We include proofs to make the work self contained. We refer to e.g.,~\cite{howie1995fundamentals,colcombet2011green} for a more structured introductions.

For the remainder of this section, fix a finite semigroup $S$. We denote by $S^1$ the monoid obtained by adding a unit element to $S$, if there is not one already.
\subsection{Green's Relations}
\label{sec:Green relations}
A standard approach to studying a semigroup is by characterizing its \emph{Green's relations}, which we now recap.
\begin{definition}[Green's Relations]
    \label{def:green relations}
    Consider two elements $u,v\in S$.
    \emph{Green's pre-order relations} are the following: 
    \begin{itemize}
        \item $u\leq_\mj v$ if $u\in S^1vS^1$. 
        \item $u\leq_\ml v$ if $u\in S^1v$. 
        \item $u\leq_\mr v$ if $u\in vS^1$.
        \item $u\leq_\mh v$ if $u\leq_\ml v$ and $u\leq_\mr v$.
    \end{itemize}
    Every Green relation $G\in \{\mj,\ml,\mr,\mh\}$ induces an equivalence relation $\sim_G$ by $u\sim_G v\iff u\leq_Gv\wedge v\leq_G u$, and we denote by $[u]_G$ the equivalence class of $u$, dubbed the \emph{$G$-class of $u$}.
    We abbreviate and write $u\mj v$ for $u\sim_\mj v$ (and similarly for $\ml,\mr,\mh$).

    When the semigroup $S$ is not clear from context (in particular when several semigroups are involved) we write e.g., $\mj_S$ to emphasize which semigroup the relation is over.
\end{definition}

The results we present are corollaries and formulations of \emph{Green's lemma, the location lemma and the egg-box lemma}~\cite{colcombet2011green,howie1995fundamentals,clifford1961algebraic}. 
We bring the proofs for completeness, and since we tailor the formulations for our needs.
\begin{lemma}
\label{lem:J equivalence to one sided equivalence}
 For every $u,v\in S$ the following hold.
    \begin{enumerate}
        \item $u\mj uv \implies u \mr uv$ (i.e. if $S^1uS^1=S^1uvS^1$ then $uS^1=uvS^1$).
        \item $v \mj uv\implies uv \ml v$ (i.e. if $S^1uS^1=S^1vuS^1$ then $S^1u=S^1vu$).
    \end{enumerate}
\end{lemma}
\begin{proof}
    We prove the first implication, the second is analogous.
    Assume $S^1uS^1=S^1uvS^1$, then there exist $x,y\in S^1$ such that $x(uv)y=u$. 
    By \cref{lem:idm power} there exists an idempotent power $n\in \bbN$ such that $(vy)^n=(vy)^{2n}$.
    We now replace $u$ in the expression $x(uv)y$ by $x(uv)y$, and repeat this $n$ times. That is, we have
    \[u=xuvy=x(xuvy)vy=\ldots =x^nu(vy)^n=x^nu(vy)^{2n}=x^nu(vy)^{n}(vy)^n=u(vy)^n\]
    We therefore get that $u=u(vy)^n=uv(y(vy)^{n-1})\in uv S^1$, so $uS^1\subseteq uvS^1$ (indeed, we now have $u=uvw$ for some $w\in S^1$, so for every $uz\in uS^1$ we have $uz=uvwz\in uvS^1$).

    The direction $uvS^1\subseteq uS^1$ is trivial, giving us $uS^1=uvS^1$, as required.
\end{proof}

\begin{lemma}
\label{lem:H class is a group}
    Every $\mh$-class that is closed under multiplication is a group.
\end{lemma}
\begin{proof}
    Consider an $\mh$-class $H$ of $S$ that is closed under multiplication, i.e., $H$ is a semigroup. 
    By \cref{lem:idm power} there is an idempotent element $e\in H$. 
    We start by claiming that it is unique. Indeed, assume $e_1,e_2\in H$ are idempotent. By the definition of $\mh$ equivalence, we have $e_2\leq_\mr e_1$ and $e_1\leq_\ml e_2$. 
    Thus, there exist $r_1,l_2\in S^1$ such that $e_2=e_1r_1$ and $e_1=l_2e_2$. We then have
    \[
    e_1=l_2e_2=l_2e_2e_2=e_1e_2=e_1e_1r_1=e_1r_1=e_2
    \]
    We can therefore denote this unique idempotent by $e$.

    We now claim that $e$ is the unit of $H$. Consider some $x\in H$, then $e\mh x$, so there are $y,z$ such that $ey=x=ze$. It follows that
    \[ex=eey=ey=x=ze=zee=xe\]

    Finally, we show that every element $x\in H$ has an inverse. Since every element has an idempotent power (\cref{lem:idm power}), there exists $i>1$ such that $x^i=e$. Denote $x^{-1}=x^{i-1}$, then $x^{-1}x=xx^{-1}=x^i=e$.

    We conclude that $H$ is a group.
\end{proof}

A common practice when reasoning about semigroups is to look at an ``atomic'' kind of semigroups~\cite{howie1995fundamentals}, as follows.
\begin{definition}[Simple, \zsimple and null semigroups]
    \label{def:simple zsimple null semigroup}
    A semigroup $S$ is:
    \begin{itemize}
        \item \emph{null} if there is a zero element $\zeroel\in S$ such that for every $a,b\in S$ we have $ab=\zeroel$.
        \item \emph{simple}\footnote{Despite the same terminology, this notion is different from simple \emph{groups}.} if it is a single $\mj$-class (i.e., $u\mj v$ for all $u,v\in S$).
        \item \zsimple if it is not null\footnote{Note that the non-null requirement merely excludes the degenerate case of the semigroup $\{\zeroel,a\}$ with $a^2=\zeroel$.}, and it has a zero element $\zeroel\in S$ such that $S\setminus \{\zeroel\}$ is a single $\mj$-class (note that $\{\zeroel\}$ is also a $\mj$-class).
    \end{itemize}    
\end{definition}

The following result conceptually states that in a simple or \zsimple semigroup $S$, the $\mh$ class of $s\cdot t$ is characterized by the $\mr$-class of $s$ and the $\ml$-class of $t$, for every $s,t\in S$.
\begin{lemma}
\label{lem:green lemma}
    Consider simple or \zsimple semigroup $S$ and elements $s,t\in S$.
    \begin{enumerate}
        \item If $st=\zeroel$ then $s't'=\zeroel$ for every $s',t'$ such that $s'\ml s$ and $t'\mr t$.
        \item If $st \neq \zeroel$ then $s\mr st$ and $st \ml t$.
    \end{enumerate}
\end{lemma}
\begin{proof}
    \begin{enumerate}
        \item By \cref{def:green relations} there are $l_s,r_t\in S$ such that $s'=l_s\cdot s$ and $t'=t\cdot r_t$. Therefore,
        \[s't' = (l_ss)(tr_t)=l_s\zeroel r_t=\zeroel\]
        \item Since $st\neq \zeroel$, then by simplicity we have $s\mj st \mj t$. By \cref{lem:J equivalence to one sided equivalence} we obtain
        $s\mr st$ and $st \ml t$.
    \end{enumerate}
\end{proof}
Intuitively, \cref{lem:green lemma} states that if $st=\zeroel$, then $[st]_\mh=\zeroel$. Otherwise, we have
\[[st]_\mh=\{v\mid st\mr v \wedge st\ml v\}=\{v\mid s\mr v \wedge t\ml v\}=[s]_{\mr}\cap [t]_\ml\]

\subsection{Congruences and Ideals}
\label{sec:congruences and ideals}
A \emph{congruence} on $S$ is an \emph{equivalence} relation $\cong\subseteq S\times S$ that preserves the semigroup action. 
That is, if $u_1\cong u_2$ and $v_1\cong v_2$ then $u_1\cdot v_1 \cong u_2\cdot v_2$.
For a congruence, the \emph{quotient} set $S/{\cong}$ is a semigroup, with the operation $[u]_{\cong}\cdot [v]_{\cong} = [u\cdot v]_{\cong}$. 
Then, the quotient map $\quo_{\cong}$ defined by 
$\quo_{\cong}(u)= [u]_{\cong}$ is a surjective semigroup morphism.

An \emph{ideal} of $S$ is a subset $I\subseteq S$ that satisfies 
$S^1IS^1=I$. 
Notice that every ideal $I\subseteq S$ is in particular a sub-semigroup of $S$.

We present several results and definitions pertaining to ideals and congruences, used throughout the proof.

\begin{definition}[$\zeroel$-minimal ideal]
\label{def:0min ideal}
An ideal $I$ is \emph{$\zeroel$-minimal} if for every ideal $I'$ with $I'\subsetneq I$ it holds that $I'=\{\zeroel\}$.
\end{definition}

\begin{lemma}
\label{lem:min ideal is simple}
Consider a semigroup $S$ and a $\zeroel$-minimal ideal $I\subseteq S$ such that $I^2=I$.
Then $I$ is a simple or \zsimple semigroup
(see \cite[3.1.3]{howie1995fundamentals}).
\end{lemma}
\begin{proof}
    We show that all nonzero elements in $I$ are $\mj_I$ equivalent.
    Let $a\in I\setminus \{\zeroel\}$ (if $S$ has a $\zeroel$).
    Notice that $S^1aS^1$ is a non-zero ideal of $S$ and $S^1aS^1\subseteq I$ (since $a\in I$ and $I$ is an ideal).
    By the minimality of $I$ we have $S^1aS^1=I$. In addition, since $I^2=I$ we also have $I^3=I$, and note that $IS^1=S^1I=I$ since $I$ is an ideal.
    We therefore have:
    \[I^1aI^1\subseteq S^1aS^1= I = I^3 =I(S^1aS^1)I=(IS^1)a(S^1I)\subseteq I^1aI^1\]
    and therefore $IaI=I$. Since this holds for any non-zero $a\in I$, we conclude that for any non-zero $a,b\in I$ we have $I^1aI^1=I=I^1bI^1$, so $a\mj_I b$, meaning $I$ is simple or \zsimple.
\end{proof}

\begin{lemma}
\label{lem:0I}
    Consider a semigroup $S$ and an ideal $I\subseteq S$.
    Then $S$ has a zero $\zeroel_S$ if and only if $I$ has a zero $\zeroel_I$. If these exist, then $\zeroel_I=\zeroel_S$.
\end{lemma}
\begin{proof}
Assume $S$ has a zero $\zeroel_S$, then in particular $\zeroel_S\in S^1IS^1=I$. Then, for every $a\in I$ we have $a\zeroel_S=\zeroel_Sa=\zeroel_S$, so $\zeroel_S$ is also a zero of $I$.

Conversely, assume $I$ has a zero $\zeroel_I$, and let some $a\in S$.
Since $I$ is an ideal, we have $a\zeroel_I,\zeroel_Ia\in I$. Since $\zeroel_I$ is a zero of $I$ we have $a\zeroel_I= \zeroel_Ia=\zeroel_I$. But this holds for every $a\in S$, so  $\zeroel_I$ is also a zero of $S$.
\end{proof}

An ideal $I\subseteq S$ naturally induces an equivalence relation $\sim_I$ by collapsing $I$ to a single element, as follows.
\begin{definition}
    \label{def:rees}
    Consider an ideal $I\subseteq S$, the \emph{Rees Congruence} induced by $I$ is $\sim_I$ where for every $u,v\in S$ we have $u\sim_I v$ if $u=v$ or $u,v\in I$.

    In particular, $S/\sim_I=\{I\}\cup \{\{v\}\mid v\in S\setminus I\}$ (see \cref{rmk:identify singletons with elements Rees}).
\end{definition}
The following lemma justifies calling this relation a congruence.
\begin{lemma}[Rees Congruence (\cite{Rees_1940,howie1995fundamentals})]
\label{lem:rees cong}
    Consider an ideal $I\subseteq S$, then $\sim_I$ is a congruence. 
\end{lemma}
\begin{proof}
Clearly $\sim_I$ is an equivalence relation, as it is induced by the partition $\{\{v\}:v\notin I\}\cup \{I\}$ of $S$.
We turn to show that it is a congruence. Let $u_1\sim_I u_2$ and $v_1\sim_I v_2$.
We split to cases.
\begin{itemize}
    \item If $u_1,v_1\notin I$ then $u_2=u_1$ and $v_2=v_1$, so $u_1\cdot v_1=u_2\cdot v_2$ and in particular $u_1\cdot v_1\sim_I u_2\cdot v_2$.
    \item If $u_1,v_1 \in I$ then $u_2,v_2 \in I$. Since $I$ is an ideal, we get that $u_1\cdot v_1,u_2\cdot v_2\in I$, so $u_1\cdot v_1\sim_I u_2\cdot v_2$.
    \item Otherwise, without loss of generality $u_1\in I,v_1\notin I$ then $u_2\in I$ and $v_2=v_1$. Again since $I$ is an ideal we get that $u_1\cdot v_1,u_2\cdot v_2 \in I$ and $u_1\cdot v_1\sim_I u_2\cdot v_2$.
\end{itemize}
\end{proof}

\begin{remark}
\label{rmk:identify singletons with elements Rees}
When considering the quotient map $\quo_{\sim I}$, we often identify the quotient set with $\{I\}\cup S\setminus I$. That is, we identify singletons $\{v\}$ with $v$ for convenience.
We abbreviate $S/\sim_I$ to $S/I$ and $\quo_{\sim I}$ to $\quo_I$.
\end{remark}

\begin{lemma}
\label{lem:rees J}
    Consider an ideal $I\subseteq S$ and $a,b\in S\setminus I$, then (as per \cref{rmk:identify singletons with elements Rees})
    $a\mj_Sb \iff a\mj_{S/I}b$.
\end{lemma}
\begin{proof}
    We show that $a\leq_{\mj_S}b \iff a\leq_{\mj_{S/I}}b$ and obtain the result by symmetry.
    \begin{itemize}
        \item If $a\leq_{\mj_S}b$, then there exist $x,y\in S$ such that $a=xby$ (as elements of $S$).
        Since $I$ is an ideal and $a\notin I$ we have that $x,y\notin I$. Thus, we still have $a=xby$ as elements of $S/I$, so $a\leq_{\mj_{S/I}}b$.
        \item If $a\leq_{\mj_{S/I}} b$: There exists $x,y\in S/I$ such that $a=xby\in S/I$.
        Because $a\notin I$ (and in particular $a\neq \zeroel_{S/I}$) we get that $x,y\neq \zeroel_{S/I}$.
        Thus we have $x,y\in S\setminus I$, and we have $a=xby$ as elements in $S$, so  $a\leq_{\mj_S}b$.
    \end{itemize}
\end{proof}

\section{Decomposition Theorems}
\label{sec:decomposition theorems}
This section focuses on constructing full general decompositions for certain building-block cases, that are used for the general setting in \cref{sec:main theorem bound}. 
Specifically, in \cref{sec:unambiguity} we define our central restriction on morphisms, namely \emph{$\mr$-alignment}. 
We then turn to establish concrete bounds on decomposition depth for the cases of groups (\cref{sec:group decomposition}), simple and \zsimple semigroups (\cref{sec:simple semigroups}, which is our most elaborate result), and null groups (\cref{sec:null semigroups}).

\subsection{$\mr$-Alignment}
\label{sec:unambiguity}
We start by presenting the technical restriction on morphisms that enables our main result. We follow the precise definition with some discussion on its intuitive meaning.
\begin{definition}[$\mr$-aligned Forests and Morphisms]
\label{def:unamb property}
Consider a semigroup $S$, a stable context set $\wHole{V}$ and a morphism $\varphi\colon\wHole{V}\rightarrow S$.
A forest $f\in \tilH$ is \emph{$\mr$-aligned} over $\varphi$ if
for every $f_1,f_2\in\tilH$ and $C_1\neq C_2\in \wHole{V}$ such that $f=C_1(f_1)=C_2(f_2)$ and $\varphi(C_1)\mj \varphi(C_2)$ there exist $r\in S$ such that $r\mj \varphi(C_1)$ and $\varphi(C_2)=\varphi(C_1)\cdot r$. 

A morphism $\varphi$ is said to be \emph{$\mr$-aligned} if all forests $f\in \tilH$ are $\mr$-aligned.
\end{definition}
Note that since the requirement in the definition holds for every $C_1,C_2$, it is in fact symmetric: there is also $r'\mj \varphi(C_1)$ such that $\varphi(C_1)=\varphi(C_2)\cdot r'$. 

Intuitively, the definition states the following: suppose we decompose $f$ in two different ways (by plucking out different sub-forests), such that the remaining contexts are $\mj$-equivalent, then not only are these contexts also $\mr$-equivalent, but they are $\mr$-equivalent using multiplicands that are within the same $\mj$-class (hence the term $\mr$-aligned). 

The semigroup-savvy reader may treat this condition as $\varphi(C_1)\mj_S \varphi(C_2)\implies \varphi(C_1)\mr_{S_J}\varphi(C_2)$ where $S_J$ is the \emph{principal factor}~\cite{howie1995fundamentals} of $J$.

The applications of $\mr$-alignment become clear in \cref{sec:simple semigroups}. However, in order to apply it we need to ensure that it survives through some manipulations of morphisms.
The first such result shows that composition with a quotient morphism retains $\mr$-alignment.
\begin{lemma}
\label{lem:rees unambiguous}
    Consider an $\mr$-aligned morphism $\varphi\colon\wHole{V}\rightarrow S$ and an ideal $I\subseteq S$ with the quotient morphism $\quo_I\colon S\rightarrow S/I$.
    Then $\quo_I\circ \varphi\colon\wHole{V}\rightarrow S/I$ is $\mr$-aligned.
\end{lemma}
\begin{proof}
    Following \cref{def:unamb property}, let $f\in \tilH$ and $C_1,C_2\in \wHole{V}$ such that $f=C_1(f_1)=C_2(f_2)$ for some $f_1,f_2$. 
    Let $u_1=\quo_I(\varphi(C_1)),u_2=\quo_I(\varphi(C_2))$ be the respective images under $\quo_I\circ \varphi$, and assume that $u_1\mj_{S/I}u_2$. 
    We prove that the condition holds by cases.
    \begin{itemize}
        \item If $u_1=u_2=\zeroel_{S/I}$ (i.e., the zero element of $S/I$), take $r=\zeroel$ then clearly $r\mj_{S/I}$ and we have $u_2=u_1\cdot r=\zeroel_{S/I}$).
        \item Otherwise, $u_1,u_2\neq \zeroel_{S/I}$ (since $[\zeroel]_\mj=\{\zeroel\}$ for every semigroup). Recall by \cref{def:rees} that $\zeroel_{S/I}=I$, and in light of \cref{rmk:identify singletons with elements Rees} we can therefore (abusively\footnote{Strictly speaking, we have $u_1=\{v_1\}, u_2=\{v_2\}$ with $v_1,v_2\in S\setminus I$ and $v_1\mj_S v_2$.}) write $u_1,u_2\in S\setminus I$ and by \cref{lem:rees J} we have $u_1\mj_S u_2$.
        Since $\varphi$ is $\mr$-aligned, there exists $x\mj_S u_1$ such that $u_2=u_1x$.
        Since $I$ is an ideal it holds that $x\notin I$ (otherwise we would have $u_2\in I$). Thus, we can write $u_2=u_1x$ as elements of $S/I$. 
        Similarly, since $x\notin I$ then by \cref{lem:rees J} we have $x\mj_{S/I}u_1$.
    \end{itemize}
    We conclude that $\quo_I\circ \varphi\colon\wHole{V}\rightarrow S/I$ is $\mr$-aligned.
\end{proof}

The second result shows that restriction to the inverse image of a ($\zeroel$-minimal) ideal retains $\mr$-alignment.
\begin{lemma}
    \label{lem:min ideal unambiguous}
    Consider an $\mr$-aligned morphism $\varphi\colon\wHole{V}\rightarrow S$, and some $\zeroel$-minimal ideal $I\subseteq S$ (\cref{def:0min ideal}).
    Denote $\wHole{V'}\coloneqq \varphi^{-1}[I]$ and $\varphi'\coloneqq \varphi|_{\wHole{V'}}\colon\wHole{V'}\rightarrow I$, then $\varphi'$ is $\mr$-aligned.
\end{lemma}
\begin{proof}
    Let $f\in \tilH$ and $C'_1,C'_2\in \wHole{V'}$ such that $f=C_1(f_1)=C_2(f_2)$ for some $f_1,f_2$. 
    Let $u_1=\varphi'(C_1),u_2=\varphi'(C_2)\in I$ and assume $u_1\mj_I u_2$.
    Since $I\subseteq S$ then in particular $u_1\mj_S u_2$, so we can use the $\mr$-alignment of $\varphi$ and obtain $r\in S$ with $r\mj_S u_1$ such that $u_2=u_1r$. 
    Moreover, we have $r\in S^1 u_1 S^1$, but since $u_1\in I$ and $I$ is an ideal, it follows that $r\in I$.

    It remains to show that $r\mj_I u_1$. Recall that $[\zeroel_S]_\mj =\{\zeroel_S\}$ and therefore if $\zeroel_S\in \{u_1,u_2,r\}$ then $r=u_1=u_2=\zeroel$ and we are done. Thus, assume they are not $\zeroel_S$, so by \cref{lem:0I} we have that $r,u_1,u_2\in I\setminus \{\zeroel_I\}$ (if $I$ has a zero).
   
    Observe that $I$ is not a null semigroup: indeed, if $I^2=\{\zeroel_I\}$, then since $u_2=u_1r$ it follows that $u_2=\zeroel_I$, contradicting our assumption above.
    It follows that $I^2\subseteq I$ is a non-zero ideal, and by the $\zeroel$-minimality of $I$ we have $I^2=I$. Thus, by \cref{lem:min ideal is simple} we have that $I$ is simple or \zsimple. In either case, since $u_1,u_2,r\neq \zeroel_I$ we have that $u_1\mj_I u_2\mj_I r$ (as $I\setminus \{\zeroel_I\}$ is a single $\mj$-class, by \cref{def:simple zsimple null semigroup}). 
    We conclude that $\varphi'$ is $\mr$-aligned.
\end{proof}

\subsection{Decomposition over a Group}
\label{sec:group decomposition}
Our first base case is when the semigroup under consideration is actually a group. The overall approach to obtain a bound on the decomposition depth is to consider a node $x$ in a decomposition $\tau$, and reason about $|\phi[\ns(\tau,\wHole{V},B,x)]|$, i.e., the number of possible images of the next decomposition step. 
Recall that in \cref{lem:decomposition decreasing,cor:decomp semantically decrease} we show that this set is, in a sense, decreasing. We now specialize the arguments of \cref{lem:decomposition decreasing} to groups to obtain a sharper result.
\begin{lemma}
\label{lem:group next ctx cardinality decreasing}
    Consider a group $V$, a \emph{binary} decomposition $\tau\in \pd(\wHole{V},B)$, and a morphism $\varphi\colon\wHole{V}\to V$.
    For every $x\leq y\in \tau.\nds$ we have:
    \[|\varphi[\ns(\tau,\wHole{V}, B, y)] \leq |\varphi[\ns(\tau,\wHole{V},B,x)]|\]
\end{lemma}
\begin{proof}
    We prove the claim for $y=x\sigma$ (i.e., $y$ is a direct child of $x$), and the claim follows by induction.
    Most of the hard work is already done in \cref{cor:decomp semantically decrease}. We use it and split to cases based on $\sigma\in \{0,1\}$.
    \begin{itemize}
        \item If $y=x0$ then by \cref{cor:decomp semantically decrease} we have $\varphi[\ns(\tau,\tilde V, B, y)]\subseteq \varphi[\ns(\tau,\tilde V,B,x)]$. Thus,  
        $|\varphi[\ns(\tau,\tilde V, B, y)]|\leq |\varphi[\ns(\tau,\tilde V,B,y')]|$.
        \item If $y=x1$ then by \cref{cor:decomp semantically decrease} we have $\varphi(\tau.\ctx(y))\cdot \varphi[\ns(\tau,\wHole{V}, B, y)] \subseteq \varphi[\ns(\tau, \wHole{V}, B, x)]$
        
        Since $V$ is a group, the function $v\mapsto \phi(\tau.\ctx(y))\cdot v$ is a bijection, and therefore
        \[
        |\varphi[\ns(\tau,\tilde V, B, y)]|=|\varphi(\tau.\ctx(y))\cdot \varphi[\ns(\tau,\wHole{V}, B, y)]|\le |\varphi[\ns(\tau, \wHole{V}, B, x)]|
        \]
    \end{itemize}
    In both cases we have $|\varphi[\ns(\tau,\tilde V,B,y)]|\leq |\varphi[\ns(\tau,\tilde V,B,x)]|$, and we conclude the lemma by induction.
\end{proof}

We proceed with the main technical part of group decomposition. Here, we bound the depth of a minimal decomposition starting from some node $x$ of an existing binary partial decomposition. The bound is given in terms of $|\varphi[\ns(\tau, \wHole{V}, B, x)]|$. 
\begin{lemma}[Decomposition Over a Group]
\label{lem:group decomp technical}
Consider a partial decomposition $\tau\in \pd(\wHole{V}, B)$ and a morphism $\varphi\colon\wHole{V}\rightarrow V$ where $V$ is a group.
Let $x\in \tau.\nds$ and denote $n=|\varphi[\ns(\tau,\wHole{V}, B,x)]|$, then there exists a full general decomposition $\tau_x\in\dec(\wHole{V},B, \tau.\frst(x),\varphi)$ such that $\depth(\tau_x)\leq 2n$.
Moreover, if $\unitel\in \varphi[\ns(\tau,\wHole{V}, B,x)]$ then $\depth(\tau_x)\leq 2n-1$.
\end{lemma}
\begin{proof}
We start with an intuitive overview of the proof.
If we can decompose $x$ with a context mapped to the neutral group element $\unitel$, then we decompose using $\unitel$ ``as much as possible'', and fold this structure to a single $\I$-node (i.e., of depth 1).
We show that in this case, the remaining trees have a smaller image under $\varphi$, and we can apply induction. 

If we cannot decompose $x$ with $\unitel$, we decompose it instead to a centipede of maximal depth, which folds to a single $\C$-node. Using our framework of minimal decompositions (\cref{def:min dec}) and the group structure, we show that we can apply the induction hypothesis to the children of this $\C$-node.
We now turn to the precise details.

We prove the claim by induction on $n=|\varphi[\ns(\tau,\wHole{V}, B,x)]|$. 
The base case is $n=0$, so $\varphi[\ns(\tau,\wHole{V}, B,x)]=\emptyset$. By \cref{lem:no next contexts implies B} this implies $\tau.\frst(x)\in B$. In particular, the singleton decomposition $\tau_x$ satisfies $\tau_x= \sing(\tau.\frst(x))\in \dec(\wHole{V}, B, \tau.\frst(x),\varphi)$ and $\delta(\tau_x)=0\le 2\cdot 0$, as required. Note that the important part here is that this decomposition indeed has its only leaf in $B$.

Assume correctness for all $m<n$, we prove for $n$. We split to two cases, depending on whether $\unitel\in \varphi[\ns(\tau,\wHole{V},B,x)]$, where $\unitel$ is the neutral element of the group $V$. 
Intuitively, in both cases we consider a decomposition $\tau'$ with two properties: first, that $\tau'$ is \emph{minimal below $x$} in the sense of \cref{def:min dec}. 

Second, that $\tau'$ is maximal in the sense that none of its leaves can be further decomposed while maintaining the minimality requirements.
We then use the case assumptions to fold the node at $x$ to either an $\I$-node or a $\C$-node, while using the properties of groups in order to apply the induction hypothesis below these new nodes.

\subparagraph*{First case: $\unitel\in \varphi[\ns(\tau,\wHole{V},B,x)]$.}
Consider the context set $S=\varphi^{-1}[\unitel]$. Intuitively, we try to decompose $\tau.\frst(x)$ as much as possible using only contexts in $S$, i.e., that map to $\unitel$. We then ``fold'' this partial decomposition into a single $\I$-node, and show that we can apply the induction hypothesis to its leaves.

Let $\tau'$ be a decomposition such that $\equpto{\tau}{x}{\tau'}$ and $\tau'$ is minimal below $x$ with respect to $\wHole{V},B$ and $S$. Note that such a decomposition exists, since if we fold $x$ to a singleton in $\tau$ (namely take $\tau[x\mapsto\sing(x)]$), then $\tau[x\mapsto\sing(x)]\leq \tau$ and thus it is also a partial decomposition, and it is trivially minimal below $x$ and $\equpto{\tau}{x}{\tau'}$.

Consider some node $y\in \tau'.\nds$ below $x$ (i.e., $x<y$) and write $y=y'\sigma$. By \cref{lem:decomposition decreasing} we have the following:
\begin{itemize}
    \item If $y=y'0$ then $\varphi[\ns(\tau',\wHole{V}, B,  y)]\subseteq \varphi[\ns(\tau',\wHole{V}, B,  y')]$.
    \item If $y=y'1$ then $\overbrace{\underbrace{\varphi(\tau'.\ctx(y'1))}_{=\unitel}\cdot \varphi[\ns(\tau',\wHole{V}, B, y'1)]}^{\varphi[\ns(\tau',\wHole{V}, B,  y)]}\subseteq \varphi[\ns(\tau',\wHole{V}, B, y')]$ where $\phi(\tau'.\ctx(y'1))=\unitel$ since all the decomposition below $x$ is mapped to $\unitel$.
\end{itemize}
By induction over $y$ this shows that $\varphi[\ns(\tau',\wHole{V}, B,  y)]\subseteq \varphi[\ns(\tau,\wHole{V}, B,  x)]$ for every $y$ in $\tau'.\nds$ below $x$.

Now consider the set $L=\tau'.\type^{-1}[\lf]\cap (x\cdot \bns)$ of leaves in $\tau'$ below $x$. Since $\tau'$ is assumed to be maximal to extension, for every leaf $z\in L$ we have $\unitel\notin \varphi[\ns(\tau',\wHole{V},B,z)]$ (otherwise we could extend $\tau'$).
Combining this with the containment showed above, we conclude that for every $z\in L$ we have
$|\varphi[\ns(\tau',\wHole{V},B,z)]|\leq|\varphi[\ns(\tau',\wHole{V},B,x)]\setminus\{\unitel\}|=n-1$.
Again, we rely here on the fact that $\ns(\tau',\wHole{V},B,x)=\ns(\tau,\wHole{V},B,x)$ since $\equpto{\tau}{x}{\tau'}$.

We can therefore apply the induction hypothesis to every leaf $z$ and obtain a full general decomposition $\tau_z$ of $\tau.\frst(z)$ such that $\depth(\tau_z)\le 2(n-1)$.

By the definition of an $\I$-node (\cref{def:general decomposition}) and since $\tau'$ is minimal by assumption (\cref{def:min dec}) we can apply \cref{cor:min to I} and ``fold'' $\tau'[x]$ to a single $\I$-node (since $\unitel$ is idempotent). We can then plug each decomposition $\tau_z$ to each child of the $\I$-node as per \cref{cor:general decomp closed substitution} and obtain a full general decomposition $\tau_x$ of $\tau.\frst(x)$ of height at most $1+2(n-1)=2n-1$, as required.

\subparagraph*{Second case: $\unitel\notin \varphi[\ns(\tau,\wHole{V},B,x)]$.}
Intuitively, in this case we decompose $\tau.\frst(x)$ as much as possible using a centipede structure, that we can then ``fold'' into a single $\C$-node. The challenge is then to show that we can apply the induction hypothesis to its leaves, which requires some case analysis.

Let $\tau'$ be a decomposition with the following properties:
\begin{itemize}
    \item $\equpto{\tau}{x}{\tau'}$,
    \item $\tau'$ is minimal below $x$,
    \item $\tau'[x]$ is a centipede (i.e., $\tau'[x].\nds\subseteq 0^*\cdot (\epsilon+1)$), and
    \item $\tau'[x]$ is maximal with these properties with respect to extension.
\end{itemize}
As in the previous case, such a decomposition exists, since $\tau[x\mapsto\sing(x)]$ satisfies the first three properties.

Consider the set $L=\tau'.\type^{-1}[\lf]\cap (x\cdot \bns)$ of leaves in $\tau'$ below $x$. By the definition of a $\C$ node (\cref{def:general decomposition}) we can fold $\tau'[x]$ so that its root is a $\C$ node whose children are bijectively matched with $L$. 
We now claim that every leaf $y\in L$ has a full general decomposition of depth at most $2n-1$.
To show this, we split to cases.
\begin{itemize}
    \item If $y\in x\cdot 0^*$ (i.e., $y$ is the ``down-most'' leaf), then by the maximality of $\tau$ we have that $\ns(\tau',\wHole{V},B, y)=\emptyset$ (otherwise we could extend $\tau'[x]$). By \cref{lem:no next contexts implies B} we have $\tau'.\frst(y)\in B$, so we can use the depth $0\le 2n-1$ singleton decomposition $\sing(\tau'.\frst(y))\in \dec(\wHole{V},B,\tau'.\frst(y),\varphi)$.
    \item If $y=x\cdot 0^k1$ for some $k$ and $\unitel\in \varphi[\ns(\tau',\wHole{V},B,y)]$, then by \cref{lem:group next ctx cardinality decreasing} we have $|\varphi[\ns(\tau',\wHole{V},B,y)]|\le n$. We can now apply the first case of the induction (which is already proved), and obtain a full general decomposition $\tau_y\in \dec(\wHole{V},B,\tau'.\frst(y),\varphi)$ with $\depth(\tau_y)\le 2n-1$, as required.
    \item If $y=x\cdot 0^k1$ for some $k$ and $\unitel\notin \varphi[\ns(\tau',\wHole{V},B,y)]$, we use the centipede structure as follows. By \cref{cor:decomp semantically decrease} we have
    $\varphi[\ns(\tau',\wHole{V}, B, x0^k)]\subseteq\varphi[\ns(\tau', \wHole{V}, B, x)]$. By the assumption that $\unitel\notin \varphi[\ns(\tau,\wHole{V},B,x)]$ and since $\ns(\tau',\wHole{V},B,x)=\ns(\tau,\wHole{V},B,x)$ we have $\unitel \notin \varphi[\ns(\tau',\wHole{V},B,x)]$ and by the containment above we get $\unitel \notin \varphi[\ns(\tau',\wHole{V},B,x0^k)]$. Denote $c_y=\tau'.\ctx(y)$, then we have $c_y\neq \unitel$.

    Again by \cref{cor:decomp semantically decrease} we have $c_y \cdot \varphi[\ns(\tau',\wHole{V},B,y)]\subseteq \varphi[\ns(\tau',\wHole{V},B,x0^k)]$. We claim that $c_y\notin c_y\cdot \varphi[\ns(\tau',\wHole{V},B,y)]$. Indeed, this follows from the fact that the action of a non-unit group element is a \emph{derangement}. 
    More precisely, if (by way of contradiction) we have $c_y\in c_y\cdot \varphi[\ns(\tau',\wHole{V},B,y)]$, then since $V$ is a group, by multiplying with $c_y^{-1}$ from the left, we have $\unitel\in \varphi[\ns(\tau',\wHole{V},B,y)]$, contradicting the assumption of this subcase. On the other hand, group action is bijective, so 
    $|\varphi[\ns(\tau',\wHole{V},B,y)]|=|c_y\cdot \varphi[\ns(\tau',\wHole{V},B,y)]|$.
    Combining the results above, we have:
    \[
    \begin{split}
        &|\varphi[\ns(\tau',\wHole{V},B,y)]|=|c_y\cdot \varphi[\ns(\tau',\wHole{V},B,y)]|
        \le |\varphi[\ns(\tau',\wHole{V},B,x0^k)]\setminus\{c_y\}|\\
        &=|\varphi[\ns(\tau',\wHole{V},B,x0^k)]|-1
        \le |\varphi[\ns(\tau', \wHole{V}, B, x)]|-1=n-1
    \end{split}
    \]
    Thus, by the induction hypothesis, there is a full general decomposition $\tau_y\in \dec(\wHole{V},B,\tau'.\frst(y),\varphi)$ with $\depth(\tau_y)\le 2(n-1)\le 2n-1$, as required.
\end{itemize}
We can then plug each decomposition $\tau_y$ to each child of the $\C$-node as per \cref{cor:general decomp closed substitution} and obtain a full general decomposition $\tau_x$ of $\tau.\frst(x)$ of height at most $1+2n-1=2n$, as required.
\end{proof}

We can now conclude our decomposition bound for groups.
\begin{lemma}
    \label{lem:group decomposition bound}
    For every finite group $V$ it holds that $\norm V\leq 2|V|-1$.
\end{lemma}
\begin{proof}
    Fix $\wHole{V}$, $B$ and a morphism $\varphi\colon\wHole{V}\to V$. 
    We show that every forest $f$ that is decomposable by $\wHole{V}$ to $B$ over $\varphi$ has a decomposition tree $\tau$ of depth at most $2|V|-1$.

    Since $f$ is decomposable, in particular its singleton decomposition satisfies $\sing(f)\in \pd(\wHole{V}, B)$. We now split to cases.

    \begin{itemize}
        \item If $\unitel\in \varphi[\ns(\sing(f), \wHole{V}, B, \epsilon)]$, then by \cref{lem:group decomp technical} there is a decomposition tree $\tau$ of $f$ with 
        $\depth(\tau)\le 2|\ns(\sing(f), \wHole{V}, B, \epsilon)|-1\le 2|V|-1$, where the last inequality is because $\ns(\sing(f), \wHole{V}, B, \epsilon)\subseteq V$.
        \item If $\unitel\notin \varphi[\ns(\sing(f), \wHole{V}, B, \epsilon)]$ then again by \cref{lem:group decomp technical} there is a decomposition tree $\tau$ of $f$ with 
        $\depth(\tau)\le 2|\ns(\sing(f), \wHole{V}, B, \epsilon)|\le 2(|V|-1)\le 2|V|-1$, where the last inequality is because 
        $\ns(\sing(f), \wHole{V}, B, \epsilon)\subseteq V\setminus \{\unitel\}$.
    \end{itemize}
\end{proof}

\subsection{Decomposition over Simple Semigroups}
\label{sec:simple semigroups}
In the proof of Simon's factorization theorem for words, a central part of the proof proceeds by splitting a word $w=\sigma_1\cdots \sigma_n$ into sections where each section starts in some fixed $\mr$-class and ends in a fixed $\ml$-class. Intuitively, the way to obtain such a decomposition is to fix a pair $(L,R)$ of an $\mr$-class and an $\ml$-class, respectively, and find all occurrences in the word of the immediate consecutive pair $LR$, 

For forests, this approach becomes much more complicated. In fact, it generally breaks down (as we show in \cref{sec:lower bound}). This is the technical reason for the introduction of $\mr$-aligned morphisms (\cref{def:unamb property}). 

The analogue of the $LR$ pairs for forests is to consider two consecutive plucked subforests. This is a natural generalization of $\ns$ (\cref{def:next contexts}), where we consider the available \emph{two} contexts to be plucked from a node $x$ (i.e., to $x1$ and $x11$).
\begin{definition}
    \label{def:nts}
    \label{def: next two contexts}
    Consider a decomposition tree $\tau$ over $\wHole{V}$ and $B$, and let $x\in \tau.\nds$. We define
    \[\begin{split}
    \nts(\tau, \wHole{V}, B, x)=\{(\tau'.\ctx(x1),\tau'.\ctx(x11)):& \equpto{\tau}{x}{\tau'},
    \tau'\in \pd(\wHole{V}, B),\\
    &x11\in \tau'.\nds\}
    \end{split}
    \]
    For a morphism $\varphi$, we denote
    \[\varphi[\nts(\tau,\wHole{V},B,x)]=\{(\varphi(C_1),\varphi(C_2)): (C_1,C_2)\in \nts(\tau,\wHole{V},B,x)\}\]
    When $\wHole{V}$ and $B$ are clear for context, we write e.g., $\nts(\tau,x)$.
\end{definition}
Using our framework of rotations (\cref{sec:rotations}), we can show that $\nts$ is well-behaved (``decreasing'', in a sense), similarly to \cref{cor:decomp semantically decrease}.
\begin{lemma}
\label{lem:decomp two semantically decreasing}
    Consider a decomposition $\tau$, a binary node $x\in \tau.\type^{-1}[\B]$ and a morphism $\varphi\colon\wHole{V}\rightarrow V$, then the following hold.
    \begin{enumerate}
        \item $\varphi[\nts(\tau,\wHole{V},B,x0)]\subseteq \varphi[\nts(\tau,\wHole{V},B,x)]$
        \item If $(u,v)\in \varphi[\nts(\tau,\wHole{V},B,x1)]$ then $(\varphi(\tau.\ctx(x1))\cdot u,v)\in \varphi[\nts(\tau,\wHole{V},B,x)]$
    \end{enumerate}
\end{lemma}
\begin{proof}
We prove each part of the lemma separately.
\subparagraph*{Item 1.}
Let $(C_1,C_2)\in \nts(\tau,x0)$, we show that $(\varphi(C_1),\varphi(C_2))\in \varphi(\nts(\tau,x))$.
We assume without loss of generality that $\tau$ attains these contexts, i.e., $\tau.\ctx(x01)=C_1$ and $\tau.\ctx(x011)=C_2$ (otherwise replace $\tau$ with an extension below $x$ that does attain them).
Intuitively, in this case we want to ``raise'' the branch $x01\to x011$ so that it starts from $x$ instead of $x0$. We do this by applying either an LR to RL rotation or a BT to TB rotation, depending on the references to $x1$ (similarly to \cref{lem:decomposition decreasing}).
        \begin{itemize}
            \item If $\tau.x00@*\to x1$, then we are under the reference structure of LR to RL rotation (\cref{lem:rotation LR to RL}). Specifically,  there is a decomposition $\tau'\in \pd(\wHole{V},B)$ with the following properties:
            \begin{enumerate}
                \item $\equpto{\tau}{x}{\tau'}$
                \item $\tau.\ctx(x01)\eqUnr \tau'.\ctx(x1)$, and in particular $\varphi(C_1)=\varphi(\tau'.\ctx(x1))$.
                \item $\tau'[x1]=\tau[x01]$, and in particular $C_2\in \ns(\tau'[x01])$ (since $\tau'[x1]$ is decomposed identically to $\tau[x01]$, i.e., with $C_2$):
                \[\tau'.\ctx(x11)=\tau'[x1].\ctx(1)=\tau[x01].\ctx(1)=\tau.\ctx(x011)=C_2\]
            \end{enumerate}
            We readily obtain $(\varphi(C_1),\varphi(C_2))\in \nts(\tau',x)$, and since $\equpto{\tau'}{x}{\tau}$ it follows that $(\varphi(C_1),\varphi(C_2))\in \nts(\tau,x)$, as required.

            \item If $\tau.x01@*\to x1$, then we are under the reference structure of BT to TB rotation. We follow the same logic as above.
            By \cref{lem:rotation BT to TB} there is a decomposition $\tau'\in \pd(\wHole{V},B)$ with the following properties:
            \begin{enumerate}
                \item $\equpto{\tau}{x}{\tau'}$
                \item $\tau'.\ctx(x1)=\tau.\ctx(x01)=C_1$
6                \item $\tau'[x10]=\tau[x01]$.
            \end{enumerate}
            We therefore have that $C_2=\tau.\ctx(x011)\in \ns(\tau',x10)$, so $\varphi(C_2)\in \varphi[\ns(\tau',x10)]$.
            Recall that by \cref{cor:decomp semantically decrease} that $\varphi[\ns(\tau',x10)]\subseteq \varphi[\ns(\tau,x)]$, so $\varphi(C_2)\in \varphi[\ns(\tau',x10)]$.
            Combining this with the second property above gives the desired
            $(\varphi(C_1),\varphi(C_2))\in \varphi[\nts(\tau,\wHole{V},B,x)]$.
        \end{itemize}

\subparagraph*{Item 2.}
        Let $(C_1,C_2)\in \nts(\tau,x1)$, we show that $(\varphi(\tau.\ctx(x1))\cdot C_1,C_2)\in \varphi[\nts(\tau,x)]$.
        We assume without loss of generality that $x111\in \tau.\nds$ (otherwise we can extend $\tau$ by factoring $x1$ by $C_1$ and then $x11$ by $C_2$).
        We therefore have $C_1=\tau.\ctx(x11)$ and $C_2=\tau.\ctx(x111)$.

        We now apply a TB to BT rotation. By \cref{lem:rotation TB to BT} there is$\tau'\in \pd(\wHole{V},B)$ with the following properties:
        \begin{enumerate}
            \item $\equpto{\tau}{x}{\tau'}$
            \item $\tau'.\ctx(x1) = \tau.\ctx(x1)\cdot \tau.\ctx(x11)=\tau.\ctx(x1)\cdot C_1$.
            \item $\tau'[x1]=\tau[x11]$.
        \end{enumerate}
        By Property 3, we have that $\tau'.\ctx(x11)=\tau.\ctx(x111)=C_2$.
        We therefore have $(\tau.\ctx(x1)\cdot C_1,C_2)\in \nts(\tau',x)$, and we have the desired $(\varphi(\tau.\ctx(x1))\cdot \varphi(C_1),\varphi(C_2))\in \varphi[\nts(\tau,x)]$ (again, since $\equpto{\tau}{x}{\tau'}$).
\end{proof}

Recall that our main interest in two consecutive plucked subforests is when these form a specific pair of $\mr$- and $\ml$-classes under the morphism $\varphi$. We now define this notion, as a sort of projection of $\nts$.

\begin{definition}[LR]
\label{def:lr}
Consider a partial decomposition $\tau$ over $\wHole{V}, B$, a node $x\in \tau.\nds$ and a morphism $\varphi\colon\wHole{V}\to S$. 
We define $\lr(\tau,\wHole{V}, B, x,\varphi)\coloneqq \{([u]_\ml,[v]_\mr): (u,v)\in \varphi[\nts(\tau,\wHole{V},B,x)]\}$.

We write $\lr(\tau,x)$ or $\lr(x)$ when the remaining elements are clear from context.
\end{definition}
\begin{remark}
    \label{rmk:lr}
    Notice that by \cref{def:green relations} for every $(l,r)\in \lr(\cdot)$ $l\cap r$ is an $\mh$-class of $S$.
    In addition, because the quotient sets $S/\ml,S/\mr$ are partitions, we get that for $(l_1,r_1),\neq(l_2,r_2)\in \lr(\cdot)$ we have that $l_1\cap r_1$ and $l_2\cap r_2$ are mutually exclusive.
\end{remark}

We now show that $\lr$ is ``decreasing'', in the spirit of \cref{cor:decomp semantically decrease}, but specialized for the case of simple semigroups.
\begin{lemma}[$\lr$ is decreasing]
    \label{lem:lr decreasing}
    Consider a simple or \zsimple semigroup $S$, a stable context set $\wHole{V}$ and a morphism  $\varphi\colon\wHole{V}\rightarrow S$. Let $\tau\in \pd(\wHole{V},B)$ and consider two nodes $x\leq y$ in $\tau.\nds$. 
    If $\zeroel\notin \varphi[\ns(\tau,\wHole{V},B,x)]$, then $\lr(y)\subseteq \lr(x)$.
\end{lemma}
\begin{proof}
    Let $\tau,x$ as in the statement of the lemma, with $\tau.\type(x)=\B$.
    We show that the inclusion $\lr(y)\subseteq \lr(x)$ holds for each of the two children of $x$ (i.e., $y\in \{x0,x1\}$). The result then follows by induction for any descendant $y$.
    \begin{enumerate}
        \item $\lr(x0)\subseteq \lr(x)$: follows directly from \cref{lem:decomp two semantically decreasing}.
        \item $\lr(x1)\subseteq \lr(x)$:
        Let $(u,v)\in \varphi[\nts(\tau, \wHole{V},B,x1)]$. By \cref{lem:decomp two semantically decreasing} we have $(\varphi(\tau.\ctx(x1))\cdot u,v)\in \varphi[\nts(\tau.\wHole{V},B,x)]$.
        Due to the assumption that $\zeroel\notin \varphi[\ns(\tau,\wHole{V},B,x)]$ we know that $\varphi(\tau.\ctx(x1))\cdot u \neq \zeroel$.
        Now, by \cref{lem:green lemma}, we have that $\varphi(\tau.\ctx(x1))\cdot u\ml u$, and we have that $([u]_\ml,[v]_\mr)\in \lr(x)$ as required.
    \end{enumerate}
\end{proof}

Before proceeding with the decomposition bound, we present a technical lemma that essentially lifts the ``egg-box'' structure of the non-zero $\mj$-class (\cref{lem:green lemma}) in a simple semigroup to the specific setting of $\nts(\cdot)$.
\begin{lemma}
\label{lem:nts props}
    Let $S$ be simple or \zsimple and consider a stable context set $\wHole{V}$ and a morphism $\varphi\colon\wHole{V}\rightarrow S$ and $\tau\in \pd(\wHole{V},B)$ such that $\zeroel\notin \varphi[\ns(\tau,\epsilon)]$. For every $x\in \tau.\nds$ and $(u_1,u_2)\in \varphi[\nts(\tau,x)]$ the following hold.
    \begin{enumerate}
        \item $u_1\neq \zeroel$, $u_2\neq \zeroel$, and $u_1u_2\neq \zeroel$.
        \item $u_1u_2\in [u_1]_\mr\cap [u_2]_\ml$
        \item $([u_1]_\ml \cap [u_2]_\mr) \cdot ([u_1]_\ml \cap [u_2]_\mr) = ([u_1]_\ml \cap [u_2]_\mr)$ (note the difference in $\ml$ and $\mr$ and $u_1,u_2$ from Item 2).
    \end{enumerate}
\end{lemma}
\begin{proof}
    \begin{enumerate}
        \item By the assumption that $\zeroel\notin \varphi[\ns(\tau,\epsilon)]$ and \cref{cor:top zeros} we immediately get $u_1\neq \zeroel$ and $u_2\neq \zeroel$.
        By \cref{cor:decomp semantically decrease} we have $u_1\cdot u_2\in \varphi[\ns(\tau,x)]$, thus also $u_1\cdot u_2\neq \zeroel$.
        \item By \cref{lem:green lemma} we now have $u_1u_2\mr u_1$ and $u_1u_2\ml u_2$, so $u_1u_2\in [u_1]_\mr\cap [u_2]_\ml$.
        \item Let some $a,b\in [u_1]_\ml\cap [u_2]_\mr$ then by \cref{lem:green lemma} and because $u_1u_2\neq \zeroel$ we have $ab\neq \zeroel$. Identically to Item 2 above, this implies (by \cref{lem:green lemma}) that $ab\in [b]_\ml\cap [a]_\mr$, but $[a]_\mr=[u_2]_\mr$ and $[b]_\ml=[u_1]_\ml$, so $ab\in [u_1]_\ml\cap [u_2]_\mr$. 
        We therefore have that $[u_1]_\ml\cap [u_2]_\mr$ is an $\mh$-class that is closed under multiplication. In particular, we have the inclusion $([u_1]_\ml \cap [u_2]_\mr) \cdot ([u_1]_\ml \cap [u_2]_\mr) \subseteq ([u_1]_\ml \cap [u_2]_\mr)$. 
        
        By \cref{lem:H class is a group} it follows that this class is a group, and therefore has an identity element, giving us the converse inclusion and concluding equality.
    \end{enumerate}
\end{proof}

We now proceed to prove the existence (and construct) a bounded general decomposition over simple (or $0$-simple) semigroup. Technically, this is one of the most involved proofs in this paper.

\begin{lemma}
\label{lem:simple decomp}
    Consider a simple or \zsimple semigroup $S$ 
    and an $\mr$-aligned morphism $\varphi\colon\wHole{V}\rightarrow S$, Then $\norm \varphi \leq 2|S|+1$.
    
    Moreover, every decomposable forest $f$ admits a full general decomposition $\tau$ (with $\depth(\tau)\leq 2|S|+1$) that has at most one $\zeroel$-valued $\I$-node and no other $\zeroel$ contexts.
\end{lemma}
\begin{proof}
Various elements of the proof are depicted in \cref{fig:simple semigroup proof}.
\begin{figure}[ht]
    \centering
    \includegraphics[width=0.4\linewidth]{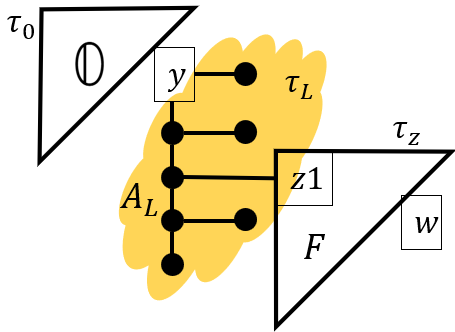}
    \caption{Illustration of the decomposition constructed in \cref{lem:simple decomp}.}
    \label{fig:simple semigroup proof}
\end{figure}

Consider a forest $f\in \tilH$ that is decomposable by $\wHole{V}$ to $B$. 
We start by getting rid of $\zeroel$, in case $S$ is \zsimple. To this end, we first decompose $f$ as much as we can with contexts mapped to $\zeroel$, so that the remaining leaves cannot be decomposed with $\zeroel$. We then fold this decomposition to a single $\I$-node.

Formally, let $\tau_0$ be a maximal-depth binary decomposition that is minimal (below $\epsilon\in \tau_0.\nds$, as per \cref{def:min dec}) with respect to the stable context set $\varphi^{-1}[\zeroel]$ (see the top triangle $\tau_0$ in \cref{fig:simple semigroup proof}).
Notice that $\sing(f)$ is such a minimal decomposition, so a maximal-depth one exists.
By \cref{cor:top zeros} we have that every leaf $x\in \tau_0.\type^{-1}[\lf]$ cannot be decomposed further by zero contexts (that is, $\zeroel\notin \varphi[\ns(\tau_0,x)]$).
Since $\tau_0$ is minimal below $\epsilon$, then by \cref{cor:min to I} we can collapse $\tau_0$ to a single $\I$-node. We refer to this collapsed tree (of depth $1$) as $\tau'_0$.

Now consider some leaf $x$ of $\tau_0$, and recall that $\zeroel\notin \varphi[\ns(\tau_0,x)]$, we construct a bounded general decomposition for $\tau_0.\frst(x)$ by induction over $|\lr(\tau,y)|$ (where $\tau$ is some extension of $\tau_0$ and $y\in \tau.\nds$). In \cref{fig:simple semigroup proof} the node $x$ is depicted as a node $y$, which conceptually can occur further below $\tau_0$, in an extension $\tau$.
We prove that for every such $\tau,y$ there is a full general decomposition for $\tau.\frst(y)$ with no $\zeroel_S$ valued $\I$-nodes, and of depth at most 
$1+\sum_{(l,r)\in \lr(\tau,y))}2\cdot|l\cap r|$
\subparagraph*{Base case: $\lr(\tau,y)=\emptyset$.} In this case, we have in particular that $\nts(\tau,y)=\emptyset$, i.e., $y$ cannot be decomposed by plucking two consecutive subforests to obtain a $y11$ node. By \cref{lem:decomp two semantically decreasing} it follows that $\nts(\tau,y0)=\emptyset$ as well, and by induction this holds for any descendant $y0^k$. We can therefore take a maximal down-centipede $\tau_y$ rooted at $y$, and we have that all its leaves are in $B$ (as per \cref{lem:no next contexts implies B}, otherwise we could obtain a $11$ child). 
    We can therefore fold $\tau_y$ to a single $\C$-node, and plug it into $\tau[y]$ (as per \cref{cor:general decomp closed substitution}), thus obtaining a decomposition of depth $1$.

In \cref{fig:simple semigroup proof} the base case would end after the highlighted centipede rooted at $y$.
\subparagraph*{Induction Step: $\lr(\tau,y)\neq \emptyset$.} 
    Fix some $(A_L, A_R)\in \lr(\tau,y)$. 
    Initially, we decompose $y$ with a down-centipede $\tau_L$ of maximal depth under several constraints (see the highlighted centipede in \cref{fig:simple semigroup proof}):
    \begin{description}
        \item[$\tau_L$ Constraint 1] Denote $\tau_1\coloneqq \tau[y\mapsto \tau_L]\in \pd(\wHole{V},B)$, then $\tau_1\in \pd(\wHole{V},B)$,
        That is, we must keep the decomposition partial even when plugged into $\tau$, so as not to eliminate the ability to extend it to a full decomposition.
        \item[$\tau_L$ Constraint 2] $\tau_1$ is minimal below $y$ with respect to the set $\varphi^{-1}[A_L]$. That is, we only use elements in $A_L$ for the decomposition. This is an analogue of marking the prefixes in a certain $\ml$-class in Simon's original setting of words. 
        \item[$\tau_L$ Constraint 3] For every $z1\in \tau_L.\nds$ we have $A_R\cap \varphi[\ns(\tau_1,yz1)]\neq \emptyset$. That is, every right child can be further decomposed by $A_R$. Again, this is analogous to the setting of words, where we look for consecutive $(A_L,A_R)$ prefixes.
    \end{description}
    Notice that there exists such a decomposition, since the singleton decomposition $\sing(\tau.\frst(y))$ satisfies the constraints. 
    In addition, observe that for the down-most leaf $y0^k$ of $\tau_L$ it holds that $(A_L,A_R)\notin \lr(\tau_1,y0^k)$, by the maximality of $\tau_L$.
    
    We proceed to the $A_R$ part of the decomposition, as follows.
    For every node $z1\in \tau_L.\nds$, we obtain a maximal-depth decomposition $\tau_z$ under the following constraints (see $\tau_z$ rooted at $z1$ in \cref{fig:simple semigroup proof}):
    \begin{description}
        \item[$\tau_z$ Constraint 1] Denote $\tau_2\coloneqq \tau_1[yz1\mapsto \tau_z]$, then $\tau_2\in \pd(\wHole{V},B)$. As above, we need to maintain decomposability.
        \item[$\tau_z$ Constraint 2] $\tau_2$ is minimal below $yz1$ with respect to the following stable \emph{factorizations set} (see \cref{def:min dec}):
        \[
        \begin{split}
        F=\{(C,f_1): & \varphi(C)\in A_L\cap A_R, \text{ and }\\
        &\forall C'\in\wHole{V} \text{ s.t }\varphi(C')\neq \zeroel, \text{ if } \unr(f_1)=C'(f'') \text{ for some $f''$, then }\varphi(C')\in A_R\}
        \end{split}
        \]
        Intuitively, the condition about $f_1$ in $F$ states that any decomposition of $f_1$ (or rather, its unraveling) into $C'(f'')$ must map $C'$ to $A_R$ (unless it is $\zeroel$).
        Notice that by \cref{lem:unr mor} and the fact that $F$ is defined with a condition on $\unr(f_1)$, it follows that $F$ is indeed a \emph{stable} factorizations set.
    \end{description}
    Again we have that $\sing(\tau_L.\frst(z1))$ is such a decomposition, and therefore a maximal-depth decomposition exists.
    
    We now turn to reason about the leaves of $\tau_z$. Consider a leaf $w$ of $\tau_z$ and let $f=\tau_z.\frst(w)$ be its forest. We claim that the following holds:
    \begin{equation}
    \label{eq:condition simple factorization}
    \forall C'\in\wHole{V} \text{ s.t }\varphi(C')\neq \zeroel, \text{ if }     \unr(f)=C'(f'') \text{ for some $f''$, then }\varphi(C')\in A_R
    \end{equation}
    That is, we claim that all the leaves satisfy the condition in Constraint 2 of $\tau_z$. We separate to cases.
    \begin{itemize}
        \item If $w=w'1$, i.e., $w$ is a right-child, the condition follows immediately by Constraint 2. Indeed, the factorization of $w'$ satisfies the constraint, and therefore $(\tau_z.\ctx(w'1),\tau_z.\frst(w'1))\in F$, but $\tau_z.\frst(w'1)=f$.
        \item If $w=w'0$, i.e., $w$ is a down-child, we start by showing that there exists \emph{some} context $C$ with $\varphi(C)\in A_R$ and $f'\in \tilH$ such that $f=C(f')$. 
        Since $\tau_z$ is a binary decomposition, then (by \cref{def:bin decomp}) $f=\tau_z.\frst(w'0)=\tau_z.\ctx(w'1)(\square_{\tau_z.\frst(w'1)})$. Denote $C=\tau_z.\ctx(w'1)$, then by Constraint 2 we have $\varphi(C)\in A_L\cap A_R$, and thus $\varphi(C)\in A_R$.
        
        Denote $f'=\square_{\tau_z.\frst(w'1)}$, we conclude that $f=C(f')$ and $\varphi(C)\in A_R$ (and $\varphi(C)\neq \zeroel$). 

        We now proceed to show \cref{eq:condition simple factorization}, i.e., that the above holds for \emph{every} context. 
        By \cref{lem:unr mor} we have $\unr(f)=\unr(C)(\unr(f'))$. Consider some context $C'$ and $f''\in \tilH$ such that $\varphi(C')\neq \zeroel$ and $\unr(f)=C'(f'')$. 
        Since $S$ is simple or \zsimple then by definition we have $\varphi(C')\mj \varphi(C)$. We now use the $\mr$-alignment of $\varphi$ (\cref{def:unamb property}) to conclude that $\varphi(C')\mr \varphi(C)$, so $\varphi(C')\in A_R$. Since this is true for every such $C'$, we conclude the property in \cref{eq:condition simple factorization}.

        \item If $w=\epsilon$, i.e., is the root (and only node) of $\tau_z$, then observe that $f=\tau_z.\frst(w)=\tau_1.\frst(yz1)$ (recall that $\tau_z$ is a decomposition of $\tau_1.\frst(yz1)$, in this case --  a singleton decomposition).
        The constraints on $\tau_L$ give us that $A_R\cap \varphi[\ns(\tau_1,yz1)]\neq \emptyset$, and in particular there exists some context $C$ with $\varphi(C)\in A_R$ such that $f=C(f')$ for some forest $f$. By the same logic as the previous case, we have that $\varphi(C)\neq \zeroel$. We can now repeat the same arguments as the previous case, and conclude the property in \cref{eq:condition simple factorization}.
    \end{itemize}

    Having established \cref{eq:condition simple factorization} for the leaf $w$, we now show that $(A_L,A_R)\notin \lr(\tau_2,yz1w)$ (recall that $\tau_2$ is obtained by plugging in $\tau_z$ in $\tau_1$ at $yz1$).
    Assume by way of contradiction that $(A_L,A_R)\in \lr(\tau_2,yz1w)$, so there are $(C_L,C_R)\in \nts(\tau_2,yz1w)$ with $\varphi(C_L)\in A_L$ and $\varphi(C_R)\in A_R$. By \cref{lem:nts props} we also have $\varphi(C_L)\neq \zeroel$ and $\varphi(C_R)\neq \zeroel$ (since $\zeroel$ was decomposed maximally already in $\tau_1$). In addition, by \cref{lem:unr mor} we have $\unr(f)=\unr(C_L)(f')$ for some $f'\in \tilH$. 
    This places us under the conditions of \cref{eq:condition simple factorization}, and therefore $\varphi(C_L)\in A_R$. 
    Thus, $\varphi(C_L)\in A_L\cap A_R$.

    We claim that we can now factor $w$ by the context $C_L$ while maintaining Constraint 2 of $\tau_z$, in contradiction to the maximality of $\tau_z$. 
    Indeed, $C_L\in \ns(\tau_2,yz1w)$, and factoring with $C_L$ yields $f=C_L(f'')$ for some forest $f''\in \tilH$ that is decomposable by $C_R$ where $\varphi(C_R)\in A_R$ (due to the assumption that $(C_L,C_R)\in \nts(\tau_2,yz1w)$).
    We can therefore write $\unr(f'')=\unr(C_R)(f''')$ with $\varphi(C_R)\in A_R$. 
    We now follow an identical argument as above: since $S$ is simple, every context $C'_R$ that satisfies $\unr(f'')=\unr(C'_R)(g)$ satisfies $\varphi(C_R)\mj \varphi(C'_R)$. Since $\varphi$ is $\mr$-aligned (\cref{def:unamb property}), we have $\varphi(C_R)\mr \varphi(C'_R)$, so $\varphi(C'_R)\in A_R$, and we conclude Constraint 2, leading to a contradiction.
    Thus, $(A_L,A_R)\notin \lr(\tau_2,w)$ for every leaf $w$.

    We now obtain a new decomposition $\tau_3$ by plugging in $\tau_z$ at $yz1$ for every right-child $z1\in \tau_L.\nds$. That is, $\tau_3=\tau_1[yz1\mapsto \tau_z]$ (for every $z1$). By the above, every leaf $w$ of $\tau_3$ below $y$ (i.e., $y\le w$) satisfies that $(A_L,A_R)\notin \lr(\tau_3,w)$. 
    Combining this with \cref{lem:lr decreasing}, we have that $|\lr(\tau_3,w)|<|\lr(\tau,y)|$ (recall that our induction is on $|\lr(\tau,y)|$). By the induction hypothesis we obtain for every leaf $w$ a general decomposition $\tau'_w\in \dec(\wHole{V},B,\tau_3.\frst(w),\varphi)$ of depth at most $\sum_{(l,r)\in \lr(\tau_3,w))}2|l\cap r|+1$.

    Intuitively, all the work thus far was in order to apply the induction hypothesis. In the remainder of the proof, we fold the prefix of the tree to a lower depth. This is done separately for $\tau_z$ and for $\tau_L$, where the latter is simply folding $\tau_L$ to a $\C$-node. We start with the former.

    Fix $z1\in \tau_L.\nds$ and consider $\tau_z$. Recall that by Constraint 2 of $\tau_z$ we have that all contexts in $\tau_z$ map to $A_L\cap A_R$.
    By \cref{lem:nts props} we have that $(A_L\cap A_R)\cdot (A_L\cap A_R) = (A_L\cap A_R)$, so $A_L\cap A_R$ is a sub-semigroup of $S$, with the following implications:
    \begin{itemize}
        \item The set $\wHole{V'}\coloneqq \varphi^{-1}[A_L\cap A_R]$ is a sub-semigroup of $\tilV$, and therefore a stable context set.
        \item By \cref{rmk:lr} $A_L\cap A_R$ is an $\mh$-class of $S$. Therefore, by \cref{lem:H class is a group} we have that $A_L\cap A_R$ is a group.
    \end{itemize}
    We now replace $\tau_z$ with a (possibly very different) decomposition $\tau'_z$ over the same \emph{leaves}, using the fact that all the contexts in $\tau_z$ come from a group, giving us the bound from \cref{lem:group decomposition bound}.
    
    Denote by $B' = \{\tau_z.\frst(w):w\in \tau_z.\type^{-1}[\lf]\}$ the forests in the leaves of $\tau_z$. We claim that $\tau_L.\frst(z1) = \tau_z.\frst(\epsilon)$ is decomposable by the (stable) context set $\wHole{V'}$ to the basis $B'$. 
    The nontrivial part required to establish decomposability is ensuring that all the contexts in $\tau_z$ are in $\wHole{V'}$, but also all the \emph{linking contexts} (\cref{def:linking context}). While the former holds by definition, the latter requires \cref{lem:min factorization dec no references}: since $\tau_z$ is minimal with respect to the factorization set $F$, it follows that it does not have inherited references. Hence, all the linking contexts are in $\wHole{V'}$ by the fact that $\wHole{V'}$ is stable, so we conclude decomposability.
    
    We restrict $\varphi$ to this context set by defining $\varphi'= \varphi|_{\wHole{V'}}\colon\wHole{V'}\rightarrow A_L\cap A_R$. Then, by \cref{lem:group decomposition bound} there exists a general decomposition $\tau'_z$ of $\tau_z.\frst(\epsilon)$ with leaves in $B'$, contexts in $\wHole{V'}$, according to $\varphi'$, and with depth at most $2|A_L\cap A_R|-1$.
    
    Notice that by the definitions of $\varphi',\wHole{V'}$ as restrictions of $\varphi$ and $\wHole{V}$ we have that $\tau'_z$ is also a general decomposition with respect to $\wHole{V}$ and $\varphi\colon\wHole{V}\rightarrow S$.
    Moreover, since $\zeroel_S\notin A_L\cap A_R$, then $\tau'_z$ has no $\zeroel$-valued $\I$-nodes.
    
    We are finally ready to complete the induction step, namely constructing a full general decomposition for $\tau,y$:
    \begin{itemize}
        \item For every leaf $w$ of $\tau_z$, let $\tau'_w$ be the decomposition for $w$ obtained by the induction hypothesis.
        Notice that $\tau'_w$ has no $\zeroel$ valued $\I$ node (by induction).
        \item For every leaf $z1$ of $\tau_L$, let $\tau''_z$ be the corresponding decomposition above. To each leaf of $\tau''_z$ with forest $f$ we plug a corresponding $\tau'_w$ that decomposes $f$ (note that each leaf of $\tau''_z$ indeed has such $\tau'_w$, but there may be leaves in $\tau'_w$ that do not appear in $\tau''_z$, since these two decompositions may differ, but share the basis $B'$).
        \item We fold $\tau_L$ to a single $\C$-node $\tau_C$, and plug each $\tau_z$ to its corresponding leaf in $\tau_C$ (again, there is a bijection between the leaves of $\tau_L$ and $\tau_C$).
    \end{itemize}
    The resulting resulting decomposition has no $\zeroel$ valued $\I$ nodes, and has depth bounded by:
    \[\begin{split}
    &\underbrace{1+\sum_{(l,r)\in \lr(\tau,y)\setminus\{(A_L,A_R)\}}(2\cdot|l\cap r|)}_{\text{inductive $\tau'_w$}} + \underbrace{2\cdot|l\cap r|)-1}_{\tau''_z} + \underbrace{1}_{\tau_C}\\
    &=1+\sum_{(l,r)\in \lr(\tau,y)}(2\cdot|l\cap r|)
    \end{split}
    \]
    as required to conclude the induction.

    We can now complete the proof of the lemma by bounding the overall depth of the decomposition.
    Observe that by \cref{rmk:lr} we have that $\{l\cap r: (l,r)\in \lr(\tau_0,x)\}$ are mutually exclusive, and thus \[1+\sum_{(l,r)\in \lr(\tau_0,x)}2|l\cap r| =1+ 2\sum_{(l,r)\in \lr(\tau_0,x)}|l\cap r| = 1+ 2\left|\bigcup_{(l,r)\in \lr(\tau_0,x)}l\cap r\right|\leq 2|S\setminus\{\zeroel\}|+1\] 
    We therefore have the following:
    \begin{itemize}
        \item If $S$ is \zsimple, the above is bounded by $2(|S|-1)+1=2|S|-1$. Then, the full decomposition includes the initial single $\I$-node $\tau'_0$, so the depth is at most $2|S|$.
        \item If $S$ is simple, the above is bounded by $2|S|+1$ (and there is no initial $\I$-node).
    \end{itemize}

    We conclude that for $\mr$-aligned $\varphi$ and simple or \zsimple $S$ we have $\norm{\varphi}\le 2|S|+1$. Moreover, in both cases we have at most a single $\zeroel$-valued $\I$-node as the initial node (namely $\tau'_0$).
\end{proof}

\subsection{Decomposition over null semigroups}
\label{sec:null semigroups}
The final ``atomic'' base case in our decomposition is that of \emph{null semigroups} (\cref{def:simple zsimple null semigroup}), i.e., those where $st=\zeroel$ for every $s,t\in S$. We show that such semigroups always allow decompositions of depth at most 2.
\begin{lemma}
    \label{lem:null semigroup decomposable depth 2}
    Consider a null semigroup $S$, then $\norm S \leq 2$.

    Moreover, every decomposable forest $f$ admits a full general decomposition $\tau$ (with $\depth(\tau)\leq 2$) that has at most one $\zeroel$-valued $\I$-node and no other $\zeroel$ contexts.
\end{lemma}
\begin{proof}
    Let $\varphi\colon\wHole{V}\rightarrow S$ be a morphism and consider a forest $f\in \tilH$ that is decomposable by $\wHole{V}$ to $B$.
    Let $\tau_0$ a maximal-depth decomposition rooted with $f$, which is minimal (as per \cref{def:min dec}) with respect to the context set $\varphi^{-1}[\zeroel]$, and notice that this is a stable context set (\cref{def:stable context set}).
    By \cref{cor:min to I} we can fold $\tau_0$ to a single $\I$ node $\tau'_0$.
    Moreover, by \cref{cor:top zeros} and the maximality of $\tau_0$, we have that $\zeroel\notin \ns(\tau_0,y)$ for every leaf $y\in \tau_0.\type^{-1}[\lf]$.
    
    We now show that each leaf $y$ can be extended to a full decomposition of depth 1 using a single $\C$-node. Consider a leaf $y$. 
    By \cref{lem:min full dec} there is $\tau_{0,y}\in \fd(\wHole{V},B)$ that is minimal below $y$ and $\equpto{\tau_{0,y}}{y}{\tau_0}$. 
    Let $\tau_{1,y}=\tau_{0,y}[y]$ by the subtree of $\tau_{0,y}$ rooted at $y$; we claim that $\tau_{1,y}$ is a down centipede.
    Indeed, assume by way of contradiction that some node $z\in \tau_{1,y}.\nds$ has $z11\in \tau.\nds$, then $\tau_{1,y}.\ctx(z11)\in \ns(\tau_{1,y},z1)$. 
    Since $S$ is a null semigroup, then by \cref{cor:decomp semantically decrease} we have that 
    $\zeroel=\varphi(\tau_{1,y}.\ctx(z1))\cdot \varphi(\tau_{1,y}.\ctx(z11))\in \varphi[\ns(\tau_{1,y},z)]$, but (again by \cref{cor:decomp semantically decrease}) this contradicts the maximality of $\tau_{0}$ (which guaranteed that $\zeroel$ is not in $\ns$ of any leaf).

    We can therefore fold each $\tau_{1,y}$ to a single $\C$-node, yielding an overall decomposition that consists of a singe $\I$-nodes whose leaves are single $\C$-nodes, hence depth at most $2$, and indeed has at most single $\zeroel$ valued $\I$ node.
\end{proof}

\section{A Bound on Decomposition Depth}
\label{sec:main theorem bound}
We are now ready to prove our main result, namely that every $\mr$-aligned morphism admits decompositions of bounded depth.

\begin{theorem}
\label{thm:main}
    Consider a finite semigroup $S$ and an $\mr$-aligned morphism $\varphi\colon\wHole{V}\rightarrow S$, then $\norm{\varphi}\leq 4|S| - 3$.
\end{theorem}
\begin{proof}
    Let $f$ be a decomposable forest.
    The proof is by induction on the semigroup size $|S|$. More precisely, our induction hypothesis is that every decomposable forest has a full general decomposition over $\varphi$ of depth at most $4|S|-3$ such that the decomposition has at most one $\zeroel$-valued $\I$-node and no other $\zeroel$ contexts. The proof has several base cases, corresponding to the decompositions of \cref{sec:decomposition theorems}.
\subparagraph*{Base case: $|S|=1$.}
In this case, $S=\{\unitel\}$ where $\unitel$ is idempotent and is also the zero element of $S$. Therefore, every full decomposition of $f$ that is minimal below $\epsilon$ can be transformed to a single $\I$-node. Since $f$ is decomposable, then such a minimal decomposition exists, so we have $\norm{\varphi}\le \norm{S}\le 1= 4|S|-3$, and the single $\zeroel$-valued $\I$-node requirement holds. 

\subparagraph*{Base case: $S$ is a null semigroup and $|S|>1$.}
By \cref{lem:null semigroup decomposable depth 2} we immediately have $\norm \varphi \leq \norm S \leq 2 \leq 4|S| - 3$ (as $|S| \geq 2$), and that the single $\zeroel$-valued $\I$-node requirement holds.

\subparagraph*{Base case: $S$ is simple or \zsimple and $|S|>1$.}
By \cref{lem:simple decomp} we immediately have $\norm \varphi \leq 2|S|+1 \leq 4|S| - 3$ (as $|S| \geq 2$), and that the single $\zeroel$-valued $\I$-node requirement holds.

\subparagraph*{Induction Step: general $S$ with $|S|>1$.}
In the following, we use subscripts to denote the zero element of different semigroups (e.g., $\zeroel_S$ for the zero of $S$).

Since $|S|>1$, there exists some $a\in S$ with\footnote{$S$ might not have a zero at all.} $a\neq \zeroel_S$ such that $a$ is minimal with respect to $\le_\mj$ (among the non-zero elements). 

Conceptually, we follow here a similar scheme to the proof of Simon's factorization for words: we inductively ``climb up'' the hierarchy of $\mj$-classes, each time considering the ideal generated by a $\mj$-minimal element, and taking the quotient over it.

Consider the ideal $M=S^1aS^1$. We show that $M$ is a minimal non-zero ideal of $S$ (with respect to containment). To this end, notice that $a\in M$, so $M\neq \{\zeroel_S\}$ (if such a zero even exists), and that by definition $M$ is an ideal. It remains to show that it is minimal. 
Consider some non-zero ideal $I$ such that $I\subseteq M$, we show that $M=I$. Let $b\in I\setminus\{\zeroel_S\}$, then $b\in S^1a^S$, so $b\le_\mj a$. 
By the minimality of $a$ with respect to $\le_\mj$ and since $b\neq \zeroel$, we have $a\mj b$. We can therefore write $a=xby$, but since $I$ is an ideal and $b\in I$ we have
\[M=S^1aS^1=S^1xbyS^1\subseteq S^1bS^1\subseteq S^1IS^1=I\]
as required. 

Next, we notice that $|M|>1$. Indeed, if $S$ has a zero $\zeroel_S$, then $\zeroel_S\in M$ and since $M$ is non-zero it follows that $|M|>1$. Otherwise, $S$ does not have a zero. If $M=\{a\}$ (by way of contradiction) then for every $b\in S$ we have $S^1abS^1\subseteq M$ and similarly $S^1baS^1\subseteq M$ (since $M$ is an ideal). This implies $ab=ba=a$, so $a$ is a zero of $S$, which is a contradiction.

We conclude that $M$ is a minimal non-zero ideal of $S$ and $|M|>1$.
In particular, by \cref{lem:rees cong} we can consider the quotient morphism $\quo_M\colon S\to S/M$.  
Furthermore, we can apply \cref{lem:rees unambiguous} and get that $\quo_M\circ \varphi$ is $\mr$-aligned. 

Now, since $|M|>1$, then $|S/M|=|S|-|M|+1<|S|$ (by the definition of Rees Congruence in \cref{def:rees}). 
We can therefore apply the induction hypothesis, and obtain a decomposition $\tau$ of $f$ with respect to $\quo_M\circ \varphi$ such that $\depth(\tau)\le \norm{\quo_M\circ \varphi}\le 4|S/M|-3$ and $\tau$ has at most one $\zeroel_{S/M}$-valued $\I$-node. 

Our goal is now to convert $\tau$ to a decomposition with respect to $\varphi$. To achieve this, we first notice that there is very little to change: in the spirit of \cref{rmk:identify singletons with elements Rees}, the set $S/M$ can be viewed as $\{M\}\cup S\setminus M$. In particular, for every $a\in S\setminus M$ we have $\quo_M(a)=a$, so the decomposition $\tau$ is already essentially with respect to $\varphi$, with one important exception: for all $C\in \wHole{V}$ such that $\varphi(C)\in M$ we have $\quo_M\circ \varphi(C)=M\in S/M$. 

Thus, the parts of $\tau$ that might need to be adapted are those whose contexts map to $M$.
Fortunately, note that $M=\zeroel_{S/M}$ by \cref{def:rees}, and by the induction hypothesis there is at most one $\I$-node in $\tau$ that is mapped to $\zeroel_{S/M}$ and no other $\zeroel_{S/M}$ contexts. 
Our task is therefore to convert the single $\zeroel_{S/M}$-valued $\I$-node of $\tau$, denoted $x$, to a decomposition over the same leaves with respect to $\varphi$. Let $g=\tau.\frst(x)$ and let $B'=\{\tau.\frst(xi)\mid i\in \bbN \wedge xi\in \tau.\nds\}$
be the set of forests in the children of the $\I$-node $x$. By \cref{def:full general decomposition}, there is a decomposition of $g$ to the basis $B'$ whose contexts are mapped by $\varphi$ to $M=\zeroel_{S/M}$. We wish to establish a bound on the depth of such a decomposition. We split to cases.

\begin{itemize}
    \item If $M$ has a zero $\zeroel_M$ and $M^2=\{\zeroel_M\}$, then $M$ is a null semigroup (\cref{def:simple zsimple null semigroup}). We can then use \cref{lem:null semigroup decomposable depth 2} to obtain a decomposition $\tau'$ of $g$ to the basis $B'$ over $\varphi$ of depth at most $2$. 

    Plugging $\tau'$ instead of the node $x$ in $\tau$, and connecting each of its leaves to an appropriate subtree of $\tau$ (as per \cref{cor:general decomp closed substitution}) we obtain a decomposition of $f$ with respect to $\varphi$ of depth at most
    \[\depth(\tau)\underbrace{-1+2}_{\text{replace $x$ with $\tau'$}}\le 4|S/M|-3+1\le 4(|S|-1)-3+1\le 4|S|-3\]
    
    \item If $M^2\neq \{\zeroel_M\}$, or $M$ does not have a zero, observe that $M^2$ is a non-zero ideal of $S$ (since $M$ is an ideal) and moreover $M^2\subseteq M$. 
    Since $M$ is minimal by assumption, we get $M^2=M$. Therefore, $M$ satisfies the conditions of \cref{lem:min ideal is simple}, so it is simple or \zsimple. 

    Recall that we are under the assumption that $S$ is not simple or \zsimple, and therefore $M\neq S$, so $|M|<|S|$. Then, using \cref{lem:simple decomp} we obtain a decomposition $\tau'$ of $g$ to the basis $B'$ over $\varphi$ of depth at most $4|M|-3$. 
    Plugging $\tau'$ in as the case above gives us a decomposition of $f$ with respect to $\varphi$ of depth at most
    \[\depth(\tau)-1+4|M|-3\le 4|S/M|-3-1+4|M|-3=4(|S|-|M|+1)-4+4|M|-3\le 4|S|-3\]
\end{itemize}
We therefore obtain a full general decomposition of $f$ over $\varphi$ of depth at most $4|S|-3$, as required. 
In order to complete the induction, it remains to show that there is at most one $\zeroel_S$-value $\I$-node. To this end, observe that $\zeroel_S$ does not appear in $\tau$ (before our change), since $\zeroel_S\in M$ (if it exists). Thus, the only possible place for $\zeroel_S$ to occur in the modified decomposition is within $\tau'$.

Since $M$ is an ideal of $S$, by \cref{lem:0I} we have $\zeroel_M=\zeroel_S$ and therefore by either \cref{lem:min ideal is simple} or \cref{lem:simple decomp} (depending on how $\tau'$ is obtained) we have that there is at most one $\zeroel_M$-valued $\I$-node in $\tau'$, concluding the induction and the proof.
\end{proof}

\begin{corollary}
\label{cor:main forest thm}
    Consider a forest language $L$ over some alphabet $A$ such that its syntactic forest morphism $\tup{\alpha_L,\beta_L}\colon\tup{\hf,\vf}\rightarrow \tup{H_L,V_L}$ satisfy that $\beta_L$ is $\mr$-aligned, then every forest $f\in \hf$ admits a general decomposition $\tau\in \dec(\vf,\bs,f,\beta_L)$ of bounded depth $4|V_L|-3$.
\end{corollary}
This concludes the proof of our main result. In the following we present an additional tool for obtaining decomposition depth bounds that is independent from $\mr$-alignment, but assumes that there are already-known bounds.

\subsection{A Reduction Scheme}
\label{sec:reduction}
The inductive step in the proof of \cref{thm:main} constructs a decomposition over a composition of morphism by ``expanding'' an $\I$-node with a decomposition over its leaves. This technique does not rely on $\mr$-alignment, and is a way to obtain general bounds for (possibly not $\mr$-aligned) morphisms. In this section we formalize this result.

Consider two finite semigroups $V_1,V_2$ and a semigroup morphism $\varphi\colon V_1\to V_2$. The following lemma gives an upper bound on $\norm{V_1}$ by means of $\norm{V_2}$ and the bound on sub-semigroups of $V_1$. Specifically, observe that for every idempotent element $e\in V_2$, we have that $\varphi^{-1}(e)$ is a sub-semigroup of $V_1$. We then have the following.
\begin{lemma}[Reduction Lemma]
\label{lem:reduction}
Consider two finite semigroups $V_1,V_2$ and a semigroup morphism $\varphi\colon V_1\to V_2$, then
for every context set $\wHole{V}$ and a morphism $\beta\colon\wHole{V}\to V_1$ we have 
 \[\norm{\beta} \leq \norm {\phi\circ \beta} \cdot \max_{\overset{e\in V_2}{e^2=e}}\norm{\varphi^{-1}[e]}\]
 In particular,
 \[\norm{V_1} \leq \norm {V_2} \cdot \max_{\overset{e\in V_2}{e^2=e}}\norm{\varphi^{-1}[e]}\]
\end{lemma}
\begin{proof}
    Assume 
    $\norm{\varphi\circ \beta}<\infty$
    (otherwise there is nothing to prove).
    In order to prove the claim, we need to show that for every $\wHole{V}$, $B$ and morphism $\beta\colon\wHole{V}\to V_1$, and for every forest $f$ that is decomposable by $\wHole{V}$ to $B$, there exists a decomposition $\tau_1\in \dec(\wHole{V},B,f,\beta)$ whose depth is at most 
    $\norm {\phi\circ \beta} \cdot \max_{\overset{e\in V_2}{e^2=e}}\norm{\varphi^{-1}[e]}$.
    
    Thus, fix such a context set $\wHole V$, a basis $B$, a morphism $\beta\colon\wHole{V}\to V_1$, and a decomposable forest $f$. 
    Observe that $\varphi\circ \beta\colon \wHole{V}\to V_2$ is also a morphism, and recall from \cref{rmk:decomposability} that $f$ is decomposable to $B$ regardless of the morphism. 
    Therefore, there exists a decomposition $\tau_2\in \dec(\wHole{V},B,f,\varphi\circ \beta)$ with 
    $\depth(\tau_2)\le \norm {\phi\circ \beta}$. 
    We now construct the desired $\tau_1$, starting with a high-level overview.

    A naive approach is to ask whether $\tau_2$ itself is in $\dec(\wHole{V}, B, f, \varphi)$.
    Indeed, this seems to almost work, since all the ``syntactic requirements'' of \cref{def:general decomposition,def:full general decomposition} hold: the contexts in $\tau_2$ are from $\wHole{V}$, its leaves are from $B$, and all $\C$-nodes have the required underlying binary decomposition structure with no inherited  references.
    
    However, a problem arises when considering $\I$-nodes. Specifically, there may be an $\I$-node $x\in \tau_2.\nds$ with idempotent value $e\in V_2$. 
    By \cref{def:general decomposition} this means that there exists a binary decomposition $\tau_x$ whose contexts are from $(\varphi\circ \beta)^{-1}[e]=\beta^{-1}[\varphi^{-1}[e]]$. 
    However, in order for $x$ to be an $\I$ node over $\beta$ alone, we need the contexts to be from $\beta^{-1}[e']$ for some idempotent $e'\in V_1$). Unfortunately, it is not guaranteed that $\varphi^{-1}[e]$ is a single idempotent element (it may contain several elements, and they are not necessarily idempotent). Moreover, it is not even guaranteed that there is another binary decomposition that corresponds to $x$ with an idempotent value in $V_1$.

    Intuitively, this already hints at the multiplicative factor in the statement of the lemma: we construct $\tau_1$ by following the structure of $\tau_2$ until encountering an $\I$-node $x\in \tau_2.\nds$. By \cref{def:general decomposition} we have that $\tau_2.\frst(x)$ is decomposable by the semigroup $\beta^{-1}[\varphi^{-1}[e]]$ to $\{\tau_2.\frst(y): \text{$y$ is a child of $x$ in $\tau_2.\nds$}\}$.
    Then, we replace the single node $x$ in $\tau_2$ with a general decomposition of depth at most $\norm{\varphi^{-1}[e]}$ whose leaves are children of $x$ in $\tau_2$. We then continue with the decomposition from these leaves. 
    Thus, each $\I$-node causes a blowup of depth at most $\norm{\varphi^{-1}[e]}$ for some idempotent $e\in V_2$, as required. 

    We proceed with the formal details.
    Define a function $\trec$ such that for every context set $\wHole{V}$, a basis $B$ and a morphism $\beta\colon\wHole{V}\rightarrow V_1$ transforms a decomposition $\tau\in \dec(\wHole{V},B,f,\varphi\circ \beta)$ to a decomposition $\trec(\tau)\in \dec(\wHole{V}, B, f, \beta)$. We define the function inductively on the depth of $\tau$, and prove that it satisfies the following:
    \begin{itemize}
        \item If $\tau\in \dec(\wHole{V},B, f, \varphi\circ \beta)$ then $\trec(\tau)\in \dec(\wHole{V}, B, f, \beta)$.
        \item If $\tau.\type(\epsilon)\neq \I$ then $\depth(\trec(\tau))\leq 1+\max_i\{\depth(\trec(\tau[i]))\}$ for  $i\in \tau.\nds\cap \bbN$.
        \item If $\tau.\type(\epsilon)=\I$ with idempotent value $v$ then $\depth(\trec(\tau))\leq \norm{\varphi^{-1}[v]}+\max_i\{\depth(\trec(\tau[i]))\}$ for $i\in \tau.\nds\cap \bbN$.
    \end{itemize}

    Let $\tau \in \dec(\wHole{V},B,f,\varphi\circ \beta$) for some $f\in \tilH$. We define $\trec(\tau)$ by induction on $|\depth(\tau)|$. The base case is $|\depth(\tau)|=0$, so $f\in B$ and $\tau=\sing(f)$. We define $\trec(\tau)\coloneqq \tau$ and the requirements trivially hold. 
    The inductive case is $|\depth(\tau)|>0$. 
    In this case $\tau.\type(\epsilon)\neq \lf$, and we split to cases according to $\tau.\type(\epsilon)$.
        \begin{itemize}
            \item $\B$: We define $\trec(\tau)\coloneqq \tau\left[\begin{array}{c}
                0\mapsto \trec(\tau[0])\\
                1\mapsto \trec(\tau[1])
            \end{array}\right]$.
            That is, we define  $\trec(\tau).\frst(\epsilon)=\tau.\frst(\epsilon)$ and $\trec(\tau).\ctx(1)=\tau.\ctx(1)$, which results in $\trec(\tau).\frst(0)=\tau.\frst(0)$ and $\trec(\tau).\frst(1)=\tau.\frst(1)$. We then inductively use the decompositions for the children. This results in a decomposition by \cref{cor:general decomp closed substitution}, so $\trec(\tau)\in  \dec(\tilde V,B,\tau_2.\frst(x), \beta)$. 
            By the inductive hypothesis and that because $\varphi[V_1]$ contains an idempotent i.e. $\max_{\overset{e\in V_2}{e^2=e}}\norm{\varphi^{-1}[e]}\geq 1$, we obtain the bound 
            \[\begin{split}
            &\depth(\trec(\tau))\leq 1+\max\{\depth(\trec(\tau[0])),\depth(\trec(\tau[1]))\} \leq \\
            &1+\max_{\overset{e\in V_2}{e^2=e}}\norm{\varphi^{-1}[e]}\cdot (\depth(\tau)-1) \leq \max_{\overset{e\in V_2}{e^2=e}}\norm{\varphi^{-1}[e]}\cdot \depth(\tau)
            \end{split}
            \]
            
            \item $\C$: This case is analogous to the previous case: 
            let $S\coloneqq \tau.\nds\cap \bbN$. We now define $\trec(\tau)\coloneqq \tau[y\mapsto \trec(\tau[y])]$ (for $y\in S$), so that $\trec(\tau).\type(\epsilon)=\C$ as well. The bound computation is identical to the above, with the maximum taken over $S$.
            
            \item $\I$: This is the interesting case. Let $S\coloneqq \tau.\nds \cap \bbN$ (i.e., the children of the $\I$-node at the root) and let $v\in V_2$ be the idempotent value of $\tau[\epsilon]$.
            By the definition of $\I$-node (\cref{def:general decomposition}), there exists a \emph{binary} decomposition $\tau'$ such that:
            \begin{itemize}
                \item $\tau'\in \fd((\varphi\circ \beta)^{-1}[v], \tau.\frst(S))$.
                \item $\tau'.\frst(\epsilon) = \tau.\frst(\epsilon)$.
            \end{itemize}
            Let $\wHole {V'}\coloneqq \beta^{-1}[\varphi^{-1}[v]]$, then
            
            $\tau.\frst(\epsilon)$ is decomposable by $\wHole{V'}$ to $\tau.\frst(S)$.
            Denote $\beta' = \beta|_{\wHole{V'}}\colon\wHole{V'}\rightarrow \varphi^{-1}[v]$.
    
            By \cref{def:semigroup decomp bound} there exists some decomposition $\tau''\in \dec(\wHole{V'}, \tau.\frst(S), \tau.\frst(\epsilon),\beta')$ such that $\depth(\tau'')\le \norm{\varphi^{-1}[v]}$.

            Since all the leaves of $\tau''$ have labels in $\tau.\frst(S)$, we can define a mapping $g\colon\tau''.\type^{-1}[\lf]\to S$ such that for every leaf $y$ in $\tau''$ it holds that $\tau''.\frst(y)=\tau.\frst(g(y))$. Note that $g$ is not necessarily bijective -- it is possible that $\tau''$ factors differently to $\tau'$. The only requirement is that all the leaves have forests in $\tau.\frst(S)$, as they do in $\tau'$.

            We can now define $\trec(\tau)\coloneqq \tau''[y\mapsto \trec(\tau_{g(y)})]$ (for $y\in \tau''.\type^{-1}[\lf]$).
            By the inductive hypothesis (on each $\trec(\tau_{g(y)}))$), we have that $\trec(\tau)\in \dec(\wHole V, B, \tau.\frst(\epsilon), \beta)$. Moreover, since $v$ is idempotent, the depth bounds hold as follows:
            \[
            \begin{split}
            &\depth(\trec(\tau)) \leq \depth(\tau'')+\max\left\{\depth(\trec(\tau[y])):y\in S\right\} \\
            &\le \norm{\varphi^{-1}[v]}+\max\left\{\depth(\trec(\tau[y])):y\in S\right\} \\
            &\le \max_{\overset{e\in V_2}{e^2=e}}\norm{\varphi^{-1}[e]}+ (\depth(\tau)-1)\cdot \max_{\overset{e\in V_2}{e^2=e}}\norm{\varphi^{-1}[e]}\\
            &=\depth(\tau)\cdot \max_{\overset{e\in V_2}{e^2=e}}\norm{\varphi^{-1}[e]}
            \end{split}
            \]
            as required.
        \end{itemize}
    Finally, since we start with $\tau$ with 
    $\depth(\tau)\le \norm{\phi\circ \beta}$, 
    we conclude that 
    $\depth(\trec(\tau))\le \norm {\phi\circ \beta}\cdot \max_{\overset{e\in V_2}{e^2=e}}\norm{\varphi^{-1}[e]}$. 
    Since this holds for every basis and context set, we conclude the lemma with
    \[\norm{\beta} \leq \norm {\phi\circ \beta} \cdot \max_{\overset{e\in V_2}{e^2=e}}\norm{\varphi^{-1}[e]}\]
\end{proof}

\section{Lower Bound without $\mr$-Alignment}
\label{sec:lower bound}
In this section we show that without the sufficient condition of $\mr$-alignment, there is a morphism that does not have a bounded decomposition depth. 
We briefly sketch the intuition for this counterexample, before giving the precise details. 

We restrict attention to forests that are binary trees. We then tailor a specific morphism to a forest algebra, with the following property: each idempotent element in the algebra can only be used to factor only one ``side'' of a subtree (i.e., either left or right). Ensuring this property is the crux of the example.  
Once this property is established, it follows that even $\I$-nodes can only, at best, decompose one side of a subtree. Then, if we start with a complete tree of depth $d$, each node decreases its depth by at most $\frac12$. This shows that any decomposition has depth at least $\log d$, which is unbounded.

We start by describing the forest language we use. The language can be simply described as a set of labeled trees, but unfortunately we must also formalize it as a morphism to a forest algebra, which is significantly more complicated.

We work over the alphabet $A=\{a,b,\ela,\elb,\ell\}$ and consider binary trees whose leaves are labeled by $\ela,\elb,\ell$ and internal nodes labeled $a,b$. We also assume leaves are the only children of their parent.
Our forest language includes forests $f$ that satisfy the following (see \cref{fig:counter example tree}):
\begin{itemize}
    \item For all binary nodes $x$, the left child ($x[0]$) is labeled $a$ and the right child ($x[1]$) is labeled $b$.
    \item All the leaves below the left child of the root (i.e., below $f[0]$) have labels in $\{\ela,\ell\}$.
    \item All the leaves below the right child of the root (i.e., below $f[1]$) have labels in $\{\elb,\ell\}$.
\end{itemize}

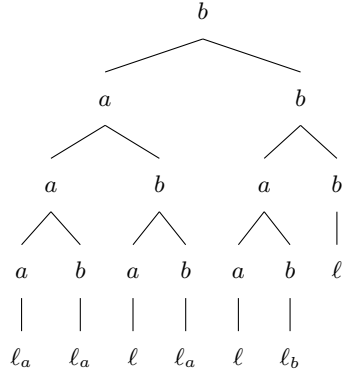
\begin{figure}[ht]
    \centering
    \begin{forest}
        [$b$
            [$a$
                [$a$
                    [$a$ [$\ell_a$]]
                    [$b$ [$\ell_a$]]
                ]
                [$b$
                    [$a$ [$\ell$]] 
                    [$b$ [$\ell_a$]]
                ]
            ]
            [$b$
                [$a$
                    [$a$ [$\ell$]]
                    [$b$ [$\ell_b$]]
                ]
                [$b$ [$\ell$]
                ]
            ]
        ]
    \end{forest}
    \caption{An tree in the language of \cref{sec:lower bound}. Note that in the left subtree of the root, all leaves are labeled in $\{\ela,\ell\}$, and in the right subtree in $\{\elb,\ell\}$.}
    \label{fig:counter example tree}
\end{figure}

Intuitively, a subtree whose leaves are labeled only with $\ell$ can be plugged anywhere in the tree. In contrast, subtrees that have a leaf labeled $\ela$ (respectively $\elb$) can only be plugged in the left (respectively right) subtrees of the root. We call the latter trees \emph{marked}.
This already suggests some inherent asymmetry in deciding how to decompose a node. 

In order to proceed with reasoning about decompositions of this language, we must formulate it as a morphism to a forest algebra.
We start by presenting the monoid $H$, with intuitive explanations about its elements. For each element we also bring a ``canonical'' example of a tree that is mapped to it. To distinguish elements in the free forest algebra from elements in the finite forest algebra, we denote the latter with bold red symbols.

\[H=\underbrace{\left\{\begin{array}{c}
     \sela,  \\
     \selb, \\
     \sell,
\end{array}\right\}}_{\text{Leaves}}
\cup 
\underbrace{\left\{\begin{array}{c c}
\alaT, & \blaT, \\
\alaF, & \blaF, \\
\albT, & \blbT, \\
\albF, & \blbF, \\
\alT,  & \blT,  
\end{array}\right\}}_{\text{One Root}}
\cup
\underbrace{\left\{\begin{array}{c}
     \latwoT,  \\
     \latwoF, \\
     \lbtwoT,\\
     \lbtwoF,
\end{array}\right\}}_{\text{Two Roots}}
\cup
\underbrace{\left\{\begin{array}{c}
     \TwRoot, \\
    \TwoRoot, \\
    \allell, \\
    \sbot, \\
    \sunit, \\ 
\end{array}\right\}}_{\text{Special}}
\]

The intuition behind the naming scheme and the types of elements (as is captured later by the morphism) is the following:
\begin{description}
    \item[Leaves:] each leaf in $\ela,\elb,\ell$  is mapped to its corresponding element in $\sela,\selb,\sell$.
    \item[One Root:] these map from single-root subtrees where either $\ela$ or $\elb$ occur (or neither, but not both). The first letter is the root of the subtree, the second marks whether the corresponding leaf exists in the subtree, and the third whether the subtree is in the language ($T$) or not ($F$). Thus, for example:
    \begin{itemize}
        \item $\alaT$ captures trees whose root is labeled with $a$, their left-subtree contains at least one $\ela$, and their right subtree contains only $\ell$ (and therefore they are in the language).

        An example of such a tree is $a(a(\ela)+ b(\ell))$.
        \item $\blaF$ captures trees whose root is labeled with $b$, and their right-subtree contains at least one $\ela$, and their right subtree contains only $\ell$ (and therefore they are not in the language).

        An example of such a tree is $b(a(\ell)+ b(\ela))$.
        \item $\alT$ captures subtree whose root is labeled with $a$, and all its leaves are labeled $\ell$ (and therefore they are in the language, which is always the case if all leaves are labeled $\ell$).

        An example of such a tree is $b(a(\ell)+ b(\ell))$.
    \end{itemize}
    \item[Two Roots:] these map from two-rooted subtrees where either $\ela$ or $\elb$ occur (but not both). 
    The first letter indicates a marking (either $\ela$ or $\elb$), $2$ indicates there are two roots, and $T$ and $F$ indicate whether the marking is on the correct side. Thus, for example:
    \begin{itemize}
        \item $\latwoT$ is mapped from forests whose left root has a leaf labeled $\ela$, and the right only has $\ell$.

        An example of such a tree is $a(\ela)+ b(\ell)$.
        \item $\lbtwoF$ is mapped from forests whose left root has $\elb$.

        Examples of such trees are $a(\elb)+ b(\elb)$, or $a(\elb)+ b(\ell)$.
    \end{itemize}

    \item[Special:] these map from some edge-cases:
    \begin{itemize}
        \item $\sbot$ is mapped from sub-forests that cannot be part of a tree in the language due to ill-formation, or violating the condition, e.g., $a(\elb)+b(\ela)$.

        \item $\sunit$ is an ad-hoc unit element, mapped to by the empty forest.

        \item $\allell$ is mapped from two-rooted forests whose leaves are all $\ell$, e.g., $a(\ell)+b(\ell)$.

        \item $\TwoRoot$ is mapped to from two-rooted forests that are fully marked, and therefore can become accepting only with a single root (but cannot serve as ``deeper'' subtrees), e.g., $a(\ela)+b(\elb)$.

        \item $\TwRoot$ is mapped to from fully marked trees that are in the language, but cannot serve as ``deeper'' subtrees, e.g., $a(a(\ela)+b(\elb))$.
    \end{itemize}
\end{description}

With these elements in mind, we now proceed to define $H$ as a monoid, and define the semantic context set $V$ and the morphism $\varphi$. While technically painful, the definitions follow immediately from the intuition of $H$.

The addition operation on $H$ is given by the following (where $\semsymb{?}$ is a wildcard):
\[h_1+h_2 = \begin{cases}
        \allell & h_1=\alT \wedge h_2 = \blT\\
        \TwoRoot & h_1=\semsymb{a\ela ?}\wedge h_2=\semsymb{b\elb?}\\
        
        \lbtwoT & h_1=\alT \wedge h_2=\semsymb{b\elb?}\\
        \latwoT & h_1=\semsymb{a\ela ?}\wedge h_2=\blT\\
        
        \latwoF & h_1\in \{\alT,\semsymb{a\ela ?}\}\wedge h_2=\semsymb{b\ela ?}\\
        \lbtwoF & h_1=\semsymb{a\elb ?}\wedge h_2\in\{\blT,\semsymb{b\elb ?}\}\\
        
        \sbot & h_1=\semsymb{a\elb ?}\wedge h_2=\semsymb{b\ela ?}\\
        \sbot & else 
    \end{cases}\]
We also define $\sunit$ defined as a neutral element.
One can verify that the intuition of the elements of $H$ is maintained through this addition. Note in particular that this already enforces trees that are at most binary, as all other additions immediately go to $\sbot$.

We proceed to define the monoid $V$ of contexts. To this end, we identify $V\subseteq H^H$, i.e., each context is an action on $H$. Specifically, we have $V$ to be the span of $V_{\text{letters}}\cup V_{\text{left}}\cup V_{\text{right}}\cup \{\textrm{id}\}$ within the semigroup $H^H$, with the following components.

\begin{itemize}
    \item The ``main'' elements of $V$ are $V_{\text{letters}}=\{\conta,\contb,\contela,\contelb,\contel\}$, each corresponding intuitively to the composition of its respective letter above a subtree. They are defined as follows.
    
\begin{align*}
\conta(h) &= \left\{
\begin{array}{ll}
\alaT & h=\sela\\
\albF & h=\selb\\
\alT    & h=\sell \vee h=\allell\\
\alaF & h=\latwoF\\
\alaT & h=\latwoT\\
\albF & h=\lbtwoF\\
\albT & h=\lbtwoT\\
\TwRoot   & h=\TwoRoot\\
\sbot & \text{otherwise}
\end{array}
\right.
\qquad
\contb(h) &= \left\{
\begin{array}{ll}
\blaF & h=\sela\\
\blbT & h=\selb\\
\blT    & h=\sell \vee h=\allell\\
\blaF & h=\latwoF\\
\blaT & h=\latwoT\\
\blbF & h=\lbtwoF\\
\blbT & h=\lbtwoT\\
\TwoRoot   & h=\TwoRoot\\
\sbot & \text{otherwise}
\end{array}
\right.
\end{align*}
and 
\[\contela(h)=\begin{cases}
    \sela & h= \sunit\\
    \sbot & \text{otherwise}
\end{cases}\quad 
\contelb(h)=\begin{cases}
    \selb & h= \sunit\\
    \sbot & \text{otherwise}
\end{cases} \quad 
\contel(h)=\begin{cases}
    \sell & h= \sunit\\
    \sbot & \text{otherwise}
\end{cases}\]
\item The elements $V_{\text{left}}$ and $V_{\text{right}}$ capture the addition from left and right for every element in $H$ (i.e., enable the functions $\inl$ and $\inr$). For every $h\in H$ define $\semsymb{r_h}(u)=u+h$ and $\semsymb{l_h}(u)=h+u$, then $V_{\text{left}}=\{\semsymb{l_h}\mid h\in H\}$ and $V_{\text{right}}=\{\semsymb{r_h}\mid h\in H\}$. We also introduce the convenient shorthands $\plusa=\semsymb{r_{\alT}},\plusb=\semsymb{r_{\blT}},\aplus=\semsymb{l_{\alT}},\bplus=\semsymb{l_{\blT}}$.

\item Finally, $\textrm{id}\in H^H$ is the identity function, and is the unit of $V$.
And we denote by $\contbot\in V$ such that for every $h\in H$ $\contbot(h)=\sbot$. (which is serves as a \emph{zero} of $V$, and it exists as for example $\contbot=(\aplus)^3\in V$)
\end{itemize}

We can finally define the morphism $\varphi\colon A^\triangle \to \tup{H,V}$ as the morphism lifted uniquely 
from the mapping $\sigma\mapsto \semsymb{f_{\sigma}}$ for $\sigma\in A=\{a,b,\ela,\elb,\ell\}$.

We now start analyzing the structure and idempotent elements of $V$, denoted $E(V)$. 
Recall that a forest $f$ is \emph{unmarked} if all its leaves are $\ell$. 
The following result studies prefix-contexts of unmarked forests. 
\begin{lemma}
\label{lem:ab idem calculation}
    Let $f$ be an unmarked forest and $C\neq \square$ a context such that $f=C\cdot f'$ for some $f'$. If $\varphi(C)$ is idempotent, then $\varphi(C)$ is one of the following:
    \[\contbot,\alT+\contb, \conta+\blT, \conta(\plusb),\contb(\aplus), \conta(\alT+\contb(\plusb)), \contb(\conta(\aplus)+\blT)\]
    We denote this set (without $\contbot$) by $E'(V)\subseteq E(V)\subseteq V$, and note these are all idempotent.

    In addition, for every $e\in E'(V)$ the following property holds (``absorbing from the left over $E'(V)$''):
    \[\forall e\in E'(V), v\in V\ e\cdot v\in E'(V)\rightarrow e\cdot v = e\]
    Moreover, the following properties hold for every $v\in E'(V)$ and $v'\in V$:    
    \[\begin{array}{l l c l}
        (P1) & v=\aplus \cdot v' \in E'(V) & \implies & v = \alT+\contb\\
        (P2) &v=\conta(\aplus)\cdot v' \in E'(V) &\implies &v=\conta(\alT+\contb(\plusb))\\
        (P3) & v=(\plusb)\cdot v'\in E'(V)& \implies& v=\conta+\blT\\
        (P4) & v=\contb(\plusb)\cdot v'\in E'(V)& \implies& v=\contb(\conta(\aplus)+\blT)
    \end{array}\]
\end{lemma}
\begin{proof}
    The proof of the lemma is by direct computation, following the definition of $\tup{H,V}$. 
    
    First, the set $E'(V)$ is computed as follows. Observe that $\{\aplus,\plusb,\conta,\contb, \contel\}$ is the set of all possible values for $\varphi(C)$, by the assumption on $f$ and $C$.
    Denote by $V_l$ the subsemigroup spanned by this set (which is finite and can be directly computed). We then compute $ E(V_l) = \{e\in V_l: e^2=e\}$ and we have $E'(V)= E(V_l)\setminus\{\contbot\}$.

    To show the absorbing property, consider for example the element $\conta(\plusb)$. We can compute the set $\conta(\plusb)\cdot V=\{\conta(\plusb)\cdot v:v\in V\}$, and then intersected it with $V_l$ to get $\{\contbot, \plusa,\conta(\plusb),\conta(\plusb(\conta)), \conta(\plusb(\conta(\aplus)))\}$. Notice that the only non $\contbot$ idempotent in this set is $\conta(\plusb)$, as required.

    For $(P1)$, we present the computation for $\plusa$. The computations for $(P2)-(P4)$ are similar. 
    By direct computation, notice that
    \[(\plusa\cdot V) \cap V_l =
    \{\contbot,\aplus,\allellplus,\aplus(\contb), \aplus(\contb(\aplus)),\aplus(\contb(\plusb))\}\] 
    and the only non $\contbot$ idempotent here is $\aplus(\contb)$ as required.
\end{proof}
The latter property shown in \cref{lem:ab idem calculation} is at the heart of our reasoning. Intuitively, it shows that if $v\in E'(V)$ has a prefix from a certain list, then we know exactly what $v$ is. This is used in the following to show that each context always decomposes a forest on a fixed ``side''.

We can now establish that $\varphi$ is indeed not $\mr$-aligned (\cref{def:unamb property}). For example, the forest $f=a(\ell)+b(\ell)$ can be factored by $C_1=a(\ell)+b(\square)$ as well as by $C_2=a(\square)+b(\ell)$. 
We have $\varphi(C_1)=\alT+\contb$ and $\varphi(C_2)=\conta+\blT$, and we claim $\alT+\contb\mj \conta+\blT$. Indeed, by \cref{lem:ab idem calculation} we have 
\[
\conta+\blT=(\conta+\blT)\cdot (a+\contb) \qquad \text{ and } \alT+\contb=(\alT+\contb)\cdot (\conta +\blT)
\]
However, again by \cref{lem:ab idem calculation}, there is no $r\mj \alT+\contb$ such that $\alT+\contb=(\conta+\blT)\cdot r$, since $\conta+\blT$ is absorbing from the left.

We now turn to formalize the intuition that every idempotent $e\in E'(V)$ has to ``pick a side'' of forests to decompose. Notice that we do not claim this in complete generality, but rather only to unmarked forests. The counterexample we construct is ultimately unmarked, and therefore this is sufficient.
\begin{lemma}
\label{lem:ab idem preserving}
    For every $e\in E'(V)$ there is an address $u\in \bbN\cup \bbN^2$ (i.e., of depth at most $2$) such that for every $f,f_1,f_0\in \tilH$ with $f\in \varphi^{-1}[\alT,\blT,\allell]$ and $C\in \varphi^{-1}[e]$ such that $f=C(f_1)$ and $f_0=C(\square_{f_1})$ we have that $f_0[u]=f[u]$ (i.e., $f$ and $f_0$ have identical subtrees at address $u$).
    
    In particular, factoring by $C$ does not pluck out any subforest of $f$ below $u$.
\end{lemma}
\begin{proof}
    We prove the claim by separating to cases according to $e\in E'(V)\setminus\{\mathrm{id}\}$. 
    \begin{itemize}
        \item $e=\alT+\contb$: notice first that $f$ has the structure $a(f_1)+b(f_2)\in \tilH$ (otherwise it cannot be factored by $e$). 
        Moreover, since $\varphi(f)\in \{\alT,\blT,\allell\}$) we actually have $\varphi(f)=\allell$, as this is the only mapping of such a structure to this set.
        Thus, $f$ has all its leaves labeled $\ell$, and in particular $\varphi(b(f_2))=\blT$.
        
        We claim that in this case we can choose $u=0$, i.e., that $f_0[u]=f[u]$. Assume by way of contradiction that $f_0[0]\neq f[0]$. Therefore, we have  $0\leq \spos(C)$, i.e., $\spos(C)$ occurs below $0$. 
        Since $C$ itself is nontrivial, then $C\neq \square$, so we can denote $C=(\square+b(f_2))\cdot C'$ for some $C'$. 
        Note that $\varphi(\square+b(f_2))= \plusb$, and since $\varphi$ is a morphism, then by $(P3)$ of \cref{lem:ab idem calculation} we get  
        \[\varphi(C)=\varphi(\square+b(f_2))\cdot \varphi(C')= \conta+\blT\neq \alT+\contb\]
        which is a contradiction. 
        Thus $f_0[0]=f[0]$ as required, and we conclude that we can only use $e=a+f_b$ to pluck out a subforest below the right root.
        \item $e=\conta+\blT$: this case is symmetric to $e=\alT+\contb$, but uses $(P1)$ of \cref{lem:ab idem calculation}.
        \item $e=\conta(\plusb)$: we follow a similar structure to the previous cases. Notice first that $f=a(a(f_1)+b(f_2))$. We claim that for $u=01$ it holds that $f_0[u]=f[u]$. 
        Indeed, otherwise we get either $C=a(\square)$ (which is not possible as $\varphi(a(\square))\neq e$), or that $C=a(a+\square)\cdot C'$ for some $C'$. Observe that $\varphi(a(a+\square))= \conta(\aplus)$, and 
        by using $(P2)$ of \cref{lem:ab idem calculation} we get 
         $\varphi(C)= \conta(\alT+\contb(\plusb))\neq e$, which is a contradiction. 
        \item $e=\contb(\aplus)$: symmetric to the previous case, by using $(P4)$ of \cref{lem:ab idem calculation}.
        \item 
        $e=\conta(\alT + \contb(\plusb))$: in this case we have $f=a(a(f_1)+b(a(f_3)+b(f_4)))$, and we claim that for $u=00$ we have $f_0[00]=f[00]$. Indeed, otherwise we get that $C=a(\square+b)\cdot C'$ for some $C'$ and because $\varphi(a(\square+b))\mapsto \conta(\plusb)$ we can use the absorption property in \cref{lem:ab idem calculation} and we get $\varphi(C)= \conta(\plusb)\neq e$, which is a contradiction.
        \item $e=\contb(\conta(\aplus)+\blT)$: symmetric to the previous case. 
    \end{itemize}
    We thus conclude the claim for every $e\in E'(V)$.
\end{proof}

The counterexample we construct utilizes a large complete binary tree, and shows that at each node in a decomposition, there remains a deep-enough complete binary subtree. To capture this we denote, for a given forest $f$ and an address $u$, the maximal depth of a complete binary subtree rooted at $u$ by $\mc(f,u)$, as well as $\mc(\epsilon)$ the maximal complete binary rooted at some root of $f$, and $\mc(f,\infty)$ as the maximal complete binary tree anywhere at $f$. Formally, we have the following.
\begin{definition}
\label{def: max comp bin tree counterexample}
Given a forest $f\in \hf$ and an address $u\in \dom(f)$ define
\[\mc(f,u)\coloneqq \min\{|v|-|u| + 1: v\in \dom(f)\wedge u\leq v\wedge v \text{ is maximal }\}\]

Define $\mc(f,\epsilon)=\min\{\mc(f,n):n\in \bbN\}$ and $\mc(f)=\mc(f,\infty)=\max\{\mc(f,u): u\in\dom(f)\}$.
\end{definition}

The following lemma relates the size of a maximal binary subtree in an unmarked forest $f$ with that of its children's in a node of a binary decomposition (\cref{def:bin decomp}).
\begin{lemma}\label{lem:mc bin}
    Consider $f,f_0,f_1\in \tilH$ with $f\in \varphi_H^{-1}[\alT,\blT,\alT+\blT]$ and $C\in \tilV$ such that $f=C(f_1)$ and $f_0=C(\square_{f_1})$.
    Then $\mc(f_0)+\mc(f_1)\geq \mc(f)$.
    Moreover, one of the following holds:
    \begin{enumerate}
        \item $\mc(f_0)=\mc(f)$, or
        \item $\mc(f_1)=\mc(f)$, or
        \item $\mc(f_0)+\mc(f_1,\epsilon) \geq \mc(f)$.
    \end{enumerate}
\end{lemma}
\begin{proof}
    Consider $u\in \dom(f)$ such that $\mc(f)=\mc(f,u)$ (i.e., $u$ witnesses the location of the maximal completely binary subtree), and denote $C'=f[u\mapsto \square]$.
    Let $f'=\tofd^{-1}(\tofd(f)[u])$ be the subforest of $f$ below $u$, so that $f=C'(f')$. 
    By the choice of $u$, we have $\mc(f',\epsilon)=\mc(f)$.
    Now split to cases according to the order of $C$ and $C'$ with respect to $\cleq$:
    \begin{enumerate}
        \item If $C$ and $C'$ are incomparable, then $f'$ is a subtree of $f_0$, so $\mc(f_0)=\mc(f)$.
        \item If $C \cleq C'$, then $f'$ is a subtree of $f_1$, hence $\mc(f_1)=\mc(f)$.
        \item If $C'\cleq C$, then by the definition of $C'$ and the fact that $f'$ is a tree (rather than a proper forest), we have that $\spos(C)=u\cdot v$ (for $v\neq \epsilon$).
        Notice that $\mc(f_0,u)=\min\{\mc(f),|v|+1\}$ and that $\mc(f_1,\epsilon)\geq \mc(f)-|v|$ (since at least the subtree starting from $u\cdot v$ is a complete binary tree in $f_1$). Overall we get that $\mc(f_0)+\mc(f_1,\epsilon)\geq \mc(f)$
    \end{enumerate}
\end{proof}

We are now ready for our main argument, showing that every node in a general decomposition of an unmarked tree decreases the depth of a maximal complete binary tree by at most $\frac12$.
\begin{theorem}
\label{thm:counter example}
    Consider a forest $f\in\varphi^{-1}[\alT]$ and a general decomposition $\tau$ such that $\tau.\frst(\epsilon)=f$. For every inner node $x\in \tau.\nds$ there is a child $y$ of $x$ such that $\mc(\tau.\frst(y)) \geq \frac12 \mc(\tau.\frst(x))$.
\end{theorem}
\begin{proof}
We prove the theorem by cases according to $\tau.\type(x)$.
\subparagraph*{If $\tau.\type(x)=\B$} If $\tau.\type(x)=\B$ then by \cref{lem:mc bin} we readily have $\mc(\tau.\frst(x0))+\mc(\tau.\frst(x1))\geq \mc(\tau.\frst(x))$, so at least one of the summands is $\ge \frac12 \mc(\tau.\frst(x))$. 

\subparagraph*{If $\tau.\type(x)=\C$\footnote{By \cref{lem:elimination of C nodes} we can actually first eliminate $\C$-nodes, but we rely on this argument in the $\I$-node case, and therefore prove it.}:}
    Denote by $\tau'$ the corresponding centipede binary decomposition (as per \cref{def:general decomposition}).
    Assume that $x0\in \tau.\nds$ corresponds to the down-most leaf in $\tau'$ (i.e., the ``index'' node), and that $xi\in \tau.\nds$ corresponds to $0^i1\in \tau'.\nds$. 
    Intuitively, we now repeat a similar argument to the proof of \cref{lem:mc bin} to show that one of these nodes has a large maximal complete binary subtree.
    
    Let $u\in \dom(\tau'.\frst(\epsilon))$ such that $\mc(\tau.\frst(x))=\mc(\tau'.\frst(\epsilon),u)$.
    \begin{itemize}
        \item If there exists some child $y=0^i1$ such that $\AncEmb(\tau',y,\epsilon)(0)\leq u$, then the entire complete subtree rooted at $u$ appears in $y$. Therefore, $\mc(\tau.\frst(xi))=\mc(\tau.\frst(x))$, and we are done.
        \item Otherwise, let $0^k$ be the down-most leaf of $\tau'$ (the index node), and notice there is an address $u'$ such that $\AncEmb(\tau,0^k,\epsilon)(u')=u$. Indeed, if none of the $0^i1$ are mapped above $u$, then $u$ remains in the residue $y$.
        If $\mc(\tau.\frst(0^k))=\mc(\tau.\frst(x))$ then we are done. Otherwise, the subtree rooted at $u'$ in $0^k$ is factored to some child $0^i1$, and we proceed.
        
        Consider some leaf of $\tau'.\frst(0^k)$ at address $v'$ below $u'$ (i.e., $u'\le v'$) of minimal depth. Thus, $\mc(\tau'.\frst(0^k),u')\ge |v'|-|u'|+1$. 
        Denote $v=\AncEmb(\tau',0^k,\epsilon)(v')$.  Intuitively, we now consider the leaf $0^i1$ where the subtree rooted at $v$ is factored. Since a centipede has no inherited references, then the \emph{entire} complete binary subtree rooted at $v'$ is factored at $0^i1$. Therefore $\mc(\tau'.\frst(0^i1),\epsilon)\ge \mc(\tau.\frst(x))-(|v'|-|u'|)$, from which we conclude the desired inequality
        \[
        \mc(\tau'.\frst(0^k),\epsilon)+\mc(\tau'.\frst(0^i1),\epsilon)\ge \mc(\tau.\frst(x))+1
        \]
        We proceed to give the formal argument.
        Note that  $v$ cannot be a leaf of $\tau'.\frst(\epsilon)$, since we would then have $\mc(\tau.\frst(\epsilon),u)= |v|-|u|+1=|v'|-|u'|+1=\mc(\tau'.\frst(0^k),v')\leq \mc(\tau.\frst(x0))$, which we already assume is not the case (otherwise we are done as above).
            
        Thus, we have that $\tau'.0^k@v\to0^i1$ for some $i<k$, and assume that $xi$ is the corresponding node in $\tau$.
        By the properties of $\AncEmb$ (\cref{lem:ancestor embedding respects holes and references} and \cref{lem:ancestor embedding not affected by ancestor}) we have that $\AncEmb(\tau',0^i1,\epsilon)(0)=v$. 
        More precisely, by \cref{lem:ancestor embedding respects holes and references} and \cref{def:ancestor embedding formal} we have $\AncEmb(\tau',0^k,0^i)(v')=\spos(\tau'.\ctx(0^i1))=\AncEmb(\tau',0^i1,0^i)(0)$ and thus by \cref{def:ancestor embedding formal} we have $\AncEmb(\tau',0^i1,\epsilon)(0)=\AncEmb(\tau'.0^i,\epsilon)\circ \AncEmb(\tau',0^i1,0^i)(0)=\AncEmb(\tau'.0^i,\epsilon)\circ \AncEmb(\tau',0^k,0^i)(v')=\AncEmb(\tau',0^k,\epsilon)(v')=v$.
        
        Now let $n,w\in \dom(\tau'.\frst(0^i1))$ such that $n\in \bbN$, and $\mc(\tau'.\frst(0^i1),\epsilon)=|w|-|n|+1=|w|$. Notice that $w$ is a leaf of $\tau'.\frst(0^i1)$, and that $\AncEmb(\tau',0^i1,\epsilon)(w)=v\spa w$.
            Because $\tau'$ has no inherited references (and $w$ is a leaf of $\tau'.\frst(0^i1)$) we have that $v\spa w$ is leaf of $\tau'.\frst(\epsilon)$, and notice that $|v\spa w| = |v|+|w|-1$.
            Overall we conclude that $\mc(\tau'.\frst(\epsilon),u)\leq |v\spa w|-|u|+1=|v|+|w|-|u| < (|v'|-|u'|+1)+|w|= \mc(\tau'.\frst(0^k),u')+\mc(\tau'.\frst(0^i1),\epsilon)$, and we conclude that $\mc(\tau.\frst(x))+1\leq \mc(\tau.\frst(x0))+\mc(\tau.\frst(xi),\epsilon)$.

    \end{itemize}
\subparagraph*{If $\tau.\type(x)=\I$:} let $\tau'$ a corresponding binary decomposition (as per \cref{def:general decomposition}).  
    We can further assume that $\tau'$ has no inherited references, by requiring that it is minimal (\cref{lem:min dec no references contexts,lem:min full dec}).

    We now focus on the left-most ``centipede'' structure induced by the nodes $(\{0\}^*\cup\{0\}^*\cdot\{1\})\cap \tau'.\nds$ (i.e, the left-most branch and it immediate right children).
    Since there are no inherited references, this is a centipede, and we can apply the same reasoning as the case of a $\C$-node. 

    This allows us to conclude that either the down-most leaf $y$ in $\tau'$ satisfies $\mc(\tau'.\frst(y))=\mc(\tau'.\frst(\epsilon))$ (in which case we are done), or there is some $i$ such that $\mc(\tau'.\frst(y))+\mc(\tau'.\frst(0^i1),\epsilon)\geq \mc(\tau.\frst(x))$.
    If $\mc(\tau'.\frst(y))\ge \frac12 \mc(\tau.\frst(x))$, then we are done. We therefore assume that this is not the case, so $\mc(\tau'.\frst(0^i1),\epsilon) > \frac 1 2 \mc(\tau.\frst(x))+1$. 
    Notice that unlike $y$ (which is a leaf, and therefore corresponds to a child of $\tau[x]$), the node $0^i1$ in $\tau'$ does not directly correspond to a child of $\tau[x]$. 
    
    We focus on the case that $\mc(\tau'.\frst(y))< \frac 12 \mc(\tau.\frst(x))$ and show that $\tau'$ has some leaf $z$ (in fact, the down-most leaf below $0^i1$) such that $\mc(\tau'.\frst(z))\ge \frac12 \mc(\tau.\frst(x))$. 
    Intuitively, this follows from \cref{lem:ab idem preserving} -- the idempotent value of the $\I$-node has to ``pick a side'' to factor by, which is maintained down the decomposition. Thus, when starting from a complete binary tree, at least half the tree survives all the way to the leaf (namely the side that is not picked). We make this intuition formal.
    
    Recall that we assume that $f\in \varphi^{-1}[\alT]$. By the definition of $\tup{V,H}$, it is easy to verify that all its subforest are therefore mapped to $\{\alT,\blT,\alT+\blT,\sell\}$. Moreover, forests mapped to $\sell$ are not decomposable at all, as the only forest mapped to $\sell$ is the letter $\ell\in A$).
    We can therefore apply \cref{lem:ab idem preserving}, so there exists an address $u\in \bbN\cup \bbN^2$ such that for every internal node $w$ in $\tau'$ we have $\tau'.\frst(w0)[u]=\tau'.\frst(w)[u]$. Applying this inductively from $0^i1$ down to the down-most leaf $0^i10^k$ yields $\tau'.\frst(0^i10^k)[u]=\tau'.\frst(0^i1)[u]$.
    
    Nest, since $|u|\leq 2$ (i.e., it is either at the root or at a child of the root) we have $\mc(\tau'.\frst(0^i10^k))\geq \mc(\tau'.\frst(0^i1,u))\geq \mc(\tau'.\frst(0^i1),\epsilon)-1$. That is, we lose at most depth $1$ between $0^i1$ and the leaf $0^i10^k$.
    
    Recall that we are under the assumption that $\mc(\tau'.\frst(y))< \frac 12 \mc(\tau.\frst(x))$  but we have 
    \[
    \mc(\tau'.\frst(y))+ \mc(\tau'.\frst(0^i1))\ge \mc(\tau.\frst(x))
    \]
    Therefore, we have $\mc(\tau'.\frst(0^i1))> \frac 12 \mc(\tau.\frst(x))$, and since these are integers we have $\mc(\tau'.\frst(0^i1))\ge  \frac 12 \mc(\tau.\frst(x))+1$.
    Combining this with the observation above, we have
    \[\mc(\tau'.\frst(0^i10^k))\geq \mc(\tau'.\frst(0^i1,u))-1\geq \frac12 \mc(\tau.\frst(x))\]
    as required.
\end{proof}

Equipped with\cref{thm:counter example}, we now have that the depth of a minimal general decomposition of an unmarked forest $f$ is at least $\log(\mc(f))$. Therefore, in order to complete our counterexample, we only need to show that there is a family of unmarked trees whose maximal complete binary subtree grows unboundedly. 
The latter is easy to construct: for every $n$, take a complete binary tree of depth $n$ where each left-child is labeled $a$ and each right-child labeled $b$. Then, connect to every leaf a child labeled $\ell$.
\begin{corollary}
    \label{cor: counter example}
    There is a morphism $\varphi$ that is not $\mr$-aligned such that for every $n\in \bbN$ there is a decomposable forest $f_n$ whose minimal decomposition depth is at least $n$.
\end{corollary}

\end{document}